\documentclass[a4paper, UKenglish, cleveref, autoref, thm-restate]{lipics-v2021}
\hypersetup{
  pdfborder={0 0 0}
}
\usepackage[T1]{fontenc}
\usepackage{lmodern} 
\usepackage{microtype}
\usepackage{graphicx}
\usepackage[dvipsnames]{xcolor}
\usepackage{hyperref}
\usepackage{amsmath}
\usepackage{amsthm}
\usepackage[normalem]{ulem}
\usepackage{todonotes}
\usepackage[skip=.5cm]{caption}
\usepackage{subcaption}
\usepackage{placeins}
\usepackage[export]{adjustbox}
\usepackage{soul}
\usepackage{xspace}
\usepackage{enumitem}
\usepackage[utf8]{inputenc}
\usepackage{booktabs}   
\usepackage{tabularx}  
\usepackage{makecell}  
\usepackage{tcolorbox}
\tcbuselibrary{skins}
\usepackage{booktabs}
\usepackage{tabularx}
\usepackage{makecell}
\usepackage{multirow}
\setlist[description]{
    labelindent = 0.5cm,
    leftmargin=0.8cm,
    format=\textcolor{black},
    itemsep=5pt,               
    font=\normalfont\bfseries\color{black},
}

\newcommand{\scup}[1]{\textnormal{\textsc{#1}}}
\DeclareMathOperator{\vc}{vc}
\DeclareMathOperator{\tw}{tw}
\newcommand{\pwfull}{\scup{Pairwise Distance Preserver}\xspace}
\newcommand{\pwshort}{\scup{PDP}\xspace}
\newcommand{\swfull}{\scup{Subsetwise Distance Preserver}\xspace}
\newcommand{\swshort}{\scup{SDP}\xspace}

\newcommand{\sw}{subsetwise}

\newcommand{\rsafull}{\scup{Rectilinear Steiner Arborescence}\xspace}
\newcommand{\rsashort}{\scup{RSA}\xspace}
\newcommand{\mccfull}{\scup{Multi--Colored Clique}\xspace}
\newcommand{\mccshort}{\scup{MCC}\xspace}
\newcommand{\alcfull}{\scup{3-Agreement List Coloring}\xspace}
\newcommand{\alcshort}{\scup{3-ALC}\xspace}
\newcommand{\mwcfull}{\scup{3-Multiway Cut}\xspace}
\newcommand{\mwcshort}{\scup{3-MWC}\xspace}
\newcommand{\bipwfull}{\scup{Vertex Cover-3 Bipartite PDP}\xspace}
\newcommand{\bipwshort}{\scup{VC3-BiPDP}\xspace}
\newcommand{\bicqfull}{\scup{Bipartite Multi--Colored Clique}\xspace}
\newcommand{\bicqshort}{\scup{BMCC}\xspace}

\newcommand{\nph}{NP-hard}

\newcommand{\wh}[1]{W$[{#1}]$-hard}
\newcommand{\fpt}{\textsc{FPT}}
\newcommand{\yes}{\textsc{Yes}}
\newcommand{\no}{\textsc{No}}

\newcommand{\defbox}[4]{
\begin{tcolorbox}[ enhanced, attach boxed title to top left={xshift=.3cm, yshift*=-4mm},  colback=white, colframe=black, colbacktitle=white, boxrule=0.6pt,
  title={\color{black}{#1} (#2)}, fonttitle=\Large, boxed title style={frame hidden, colframe=black, colback=lipicsYellow}]
	
    \textbf{Input:}  {#3}

    \medskip
	\textbf{Question:} {#4}
\end{tcolorbox}
}

\newcommand{\defboxsmalltitle}[4]{
\begin{tcolorbox}[ enhanced, attach boxed title to top left={xshift=.3cm, yshift*=-4mm},  colback=white, colframe=black, colbacktitle=white, boxrule=0.6pt,
  title={\color{black}{#1} (#2)}, fonttitle=\large, boxed title style={frame hidden, colframe=black, colback=lipicsYellow}]
	
    \textbf{Input:}  {#3}

    \medskip
	\textbf{Question:} {#4}
\end{tcolorbox}
}

\newcommand{\p}{\mathcal{P}}
\newcommand{\N}{\mathbb{N}}

\title{Finding Minimum Distance Preservers: A Parameterized Study}
\author{Kirill Simonov}{Department of Informatics, University of Bergen, Norway}{kirill.simonov@uib.no}{https://orcid.org/0000-0001-9436-7310}{}
\author{Farehe Soheil}{Hasso Plattner Institute (HPI), University of Potsdam, Germany}{farehe.soheil@hpi.de}{https://orcid.org/0000-0002-0504-8834}{Supported by the HPI Research School on Foundations of AI (FAI).}

\author{Shaily Verma}{Indian Institute of Technology Jodhpur,  India}{shailyverma@iitj.ac.in}{https://orcid.org/0000-0003-0076-6308}{}
\authorrunning{K. Simonov, F. Soheil, and S. Verma} 
\Copyright{Kirill Simonov, Farehe Soheil, and Shaily Verma}



\nolinenumbers

\usepackage{tcolorbox}
\begin{document}
\maketitle
\begin{abstract}
    For a given graph $G$ and a subset of vertices $S$, a \emph{distance preserver} is a subgraph of $G$ that preserves shortest paths between the vertices of $S$. 
    We distinguish between a \emph{subsetwise} distance preserver, which preserves distances between all pairs in $S$, and a \emph{pairwise} distance preserver, which preserves distances only between specific pairs of vertices in $S$, given in the input. While a large body of work is dedicated to upper and lower bounds on the size of distance preservers and, more generally, graph spanners, the computational complexity of finding the minimum distance preserver has received comparatively little attention.
    
    We consider the respective \swfull (\swshort) and \pwfull (\pwshort) problems and initiate the study of their computational complexity.
    We provide a detailed complexity landscape with respect to natural parameters, including the number of terminals, solution size, vertex cover, and treewidth.
    Our main contributions are as follows:
    \begin{itemize}
        \setlength{\itemsep}{0.5em}
        \item Both \pwshort and \swshort are \nph even on subgraphs of the grid. Moreover, when parameterized by the number of terminals, the problems are \wh{1}\ on subgraphs of the grid, while they become \fpt\ on full grids.
        \item \pwshort is \nph on graphs of vertex cover $3$, while \swshort is \fpt when parameterized by the vertex cover of the graph.  Thus, the vertex cover parameter distinguishes the two variants.
        \item Both problems are \fpt when parameterized by the number of terminals and the treewidth of the graph.
    \end{itemize}
\end{abstract}

\section{Introduction}

Graph spanners, introduced by Peleg and Sch\"{a}ffer~\cite{PelegS89}, aim to capture the distance structures of the graph in a smaller subgraph. 
Graph spanners can be used to speed up virtually any algorithm that relies on shortest-path distances~\cite{Cohen00,AwerbuchBCP98}, and have diverse domain-specific applications, for example in distributed computing~\cite{AwerbuchBCP98,Elkin05} or computational geometry~\cite{PaulChew89,Dobkin1990}.
Hence, finding sparse or low-weight spanners has remained a desirable research target~\cite{AHMED2020100253, BBC2003,bodwin2016better,kogan2025having,gajjar2017,CDKL20,abdolmaleki2020minimum}.

The ideal type of a spanner is a \emph{distance preserver}: a (preferably small) subgraph, preserving the shortest paths in the original graph exactly. It is easy to observe that no savings can be made if one aims to preserve all pairwise distances. In this case, each edge has to be included in the subgraph in order to retain the distance between its endpoints. However, the question becomes intriguing again if one only considers the distances between a specified subset of vertices, called the \emph{terminals}. Formally, we say that a \emph{subsetwise distance preserver} of the graph $G$ with a given set of terminals $S$, is a subgraph of $G$ such that for any pair of vertices in $S$ the distance is the same in the subgraph as in the whole graph $G$. In case only certain pairs of terminals are important, one may refine the definition further: For a given graph $G$ and a collection of vertex pairs $\mathcal{P}$, a \emph{pairwise distance preserver} is a subgraph of $G$ that preserves the distance between each pair of vertices in $\mathcal{P}$. For consistency between the two variants, we call the vertices of $\mathcal{P}$ terminals, and denote $S = \bigcup \mathcal{P}$.

Distance preservers were first introduced by Bollob\'{a}s, Coppersmith and Elkin~\cite{BBC2003} in the setting of preserving all sufficiently large pairwise distances. The more general definitions above were coined and studied by Coppersmith and Elkin~\cite{CD2006}.
Distance preservers share similar applications to graph spanners~\cite{BBC2003,CD2006}, providing general-purpose compression of the graph while keeping the relevant distances unchanged. For this reason, they have seen considerable attention within the literature on spanners.
Moreover, sparse distance preservers have been extensively used in the construction of sparse spanners~\cite{Pet09,ElkinFN17,KoganP22,kogan2025having} as well as distance oracles~\cite{ElkinS16,ElkinS23}. Hence, the search for small distance preservers is also meaningful for more general types of spanners and other objects.

Similar to general graph spanners, the majority of results in the literature on distance preservers focuses on obtaining tight lower and upper bounds on the size of the preserver. In their original work, Coppersmith and Elkin~\cite{CD2006} showed that there exists a distance preserver with linearly many edges whenever $|S| = O(n^{1/4})$ (in the subsetwise case) or $|\mathcal{P}| = O(n^{1/2})$ (in the pairwise case), where $n$ is the number of vertices in the graph. Over the last two decades, these bounds were improved and generalized~\cite{bodwin2016better,bodwin2021new,kogan2025having}.
In a related direction, Krauthgamer et al.\ study preserving terminal distances using minors, instead of subgraphs~\cite{Krauthgamer2014}. They show that $\Omega(|S|^2)$ vertices are necessary for a subsetwise distance preserving minor even when the graph is a grid, while a simple upper bound provides an $O(|S|^4)$-sized minor for any graph.

While the existing bounds provide strong guarantees on the size of the distance preserver in terms of the graph size $n$ or the number of terminals $|S|$, it is also meaningful to ask the following: \emph{What is the smallest distance preserver for the given instance?} This algorithmic question is the central focus of our work. Formally, we define the following two computational problems.
In \swfull (\swshort), we are given a graph $G$, a terminal subset $S$ and an integer $k$, and the task is to decide whether there exists a subgraph of $G$ with at most $k$ edges that preserves the distances between the vertices of $S$. 
Analogously, we define the \pwfull (\pwshort) problem where a collection of pairs of terminals $\mathcal{P}$ is given in addition. 
Note that \pwshort is a strict generalization of \swshort since the collection $\mathcal{P}$ may contain all pairs in $S$.

In contrast to the extensive work on bounds, the computational problems \pwshort and \swshort have received comparatively little attention. Gajjar and Radhakrishnan show that \swshort, and hence \pwshort, is NP-hard on general graphs, providing additional results for interval graphs in terms of the number of branching vertices in the preserver~\cite{gajjar2017}.
From an optimization viewpoint, Chlamt\'{a}\v{c} et al. provide a tight $O(n^{3/5+\epsilon})$-factor approximation for minimum-weight pairwise distance preservers~\cite{CDKL20}. Abdolmaleki et al. provide similar results for approximating the difference between the total distance upper bound and the size of the distance preserver~\cite{abdolmaleki2020minimum}.
Overall, despite the general-case NP-hardness, the question of computing a minimum distance preserver remains widely open.

Before we move on to our results, we also remark on the similarity between \pwshort/\swshort and other well-studied path-packing problems. \pwshort can be interpreted as finding $|\mathcal{P}|$ shortest paths in $G$, connecting the respective terminal pairs, so that the resulting collection of paths has the largest possible intersection, in terms of total number of edges used. The natural counterpart is then the \textsc{$k$-Disjoint Shortest Paths} problem, where the task is to find $k$ shortest paths between the respective terminal pairs that are vertex-disjoint, i.e., avoid any intersections. Lochet has recently shown that this problem is W[1]-hard when parameterized by the number of terminal pairs, but admits an XP algorithm~\cite{Lochet21}. A related problem of packing $k$ vertex- or edge-disjoint shortest cycles was studied by Bentert et al., who have shown that the problem admits an FPT algorithm on planar graphs, while a generalization of the problem becomes W[1]-hard on general graphs~\cite{BentertFGKLPRSS26}.

\subparagraph*{Our contribution}
In this work, our focus is to characterize the tractability of computing the minimum distance preserver depending on the properties of the instance, following the perspective of parameterized complexity. Our results are summarized in Table~\ref{tab:complexity-summary}.

Our starting point is the number of terminals $|S|$: Since the distance preserver is obtained by choosing a shortest path between each pair of terminals, the value of $|S|$ naturally influences the complexity of the problem. In particular, the number of terminals/pairs of terminals enters the known existential bounds, both on the size of the distance-preserving subgraphs~\cite{CD2006,bodwin2016better,bodwin2021new,kogan2025having} and minors~\cite{Krauthgamer2014}. On the other hand, related terminal-connectivity problems are often polynomial-time solvable for constantly many pairs, for example  \textsc{$k$-Disjoint Shortest Paths}~\cite{Lochet21}.

From the known structural observations~\cite{Krauthgamer2014,CD2006}, it can be quickly observed that \pwshort and \swshort admit a polynomial-time algorithm when the number of terminals is constant. In other words, the problems are XP when parameterized by $|S|$. But can this running time be improved? Our first main result answers this negatively, even for graphs that are quite restricted in terms of their planar embedding---induced subgraphs of the grid.

\begin{restatable}{theorem}{SwTerminalsWHard}
    \label{thm:sw-w1-terminals}
    \swshort\ is \wh{1} parameterized by the number of terminals, even on induced subgraphs of the grid.
\end{restatable}

\wh{1}ness for \pwshort follows immediately from Theorem~\ref{thm:sw-w1-terminals}, since \swshort is a special case of \pwshort. Note that our reduction also implies NP-hardness of \pwshort and \swshort on induced subgraphs of the grid.

On the other hand, our next result shows that the problems become FPT when the given graph is the full grid.

\begin{restatable}{theorem}{PwTerminalsGrid}
    \label{thm:pw-terminals-grid}
    \pwshort admits a $4^{|S|^2} \cdot |V(G)|^{O(1)}$-time algorithm when the input graph is a grid.
\end{restatable}
By inclusion, Theorem~\ref{thm:pw-terminals-grid} implies the same result for \swshort.

While the case of the grid may seem simple, as only the position of the terminals defines the instance, we observe that this is deceptive: \pwshort generalizes the \rsafull problem and hence is NP-hard even on grids~\cite{RSA}.
Additionally, the known bounds for sizes of distance-preserving minors are not better on grids than on general graphs~\cite{Krauthgamer2014}, hence the grids remain a challenging special case in the context of distance preservers.

Having established the boundary of tractability for grid-like graphs, we turn our attention to the next natural question: Can we hope to compute distance preservers efficiently on graphs that avoid large grids?
In particular, we may consider graphs of bounded treewidth---recall that a graph has bounded treewidth if and only if it excludes a grid minor of constant size, with a polynomial dependency between the treewidth and the maximum size of the grid minor~\cite{ChekuriC16}. We answer the question positively, showing that both \swshort and \pwshort are FPT when parameterized by $|S|$ plus the treewidth of the input graph.

\begin{restatable}{theorem}{pairwiseFPTbyTWK}
\label{thm:pw-tw-k-fpt}
    \pwshort admits an algorithm with running time $2^{O\bigl((\tw+|S|)^2\bigr)} \cdot |V(G)|^{O(1)}$, where $\tw$ is the treewidth of $G$.
\end{restatable}

The result of Theorem~\ref{thm:pw-tw-k-fpt} in turn raises a question of whether treewidth alone may be sufficient for an FPT algorithm, with an arbitrary number of terminals.
Unfortunately, our next result answers this negatively for \pwshort, even for a much more restrictive parameter.

\begin{restatable}{theorem}{VcParaNpHardness}
    \label{thm:pdp-vc-paraNphardness}
    \pwshort\ is \nph\ even on bipartite graphs with vertex cover $3$.
\end{restatable}

On the other hand, by a branching argument we observe that \swshort is FPT when parameterized by the vertex cover number of the graph.
\begin{restatable}{theorem}{SwVCFPT}
\label{thm:sw-vc-fpt-thm}
    \swshort admits a
        $2^{O(2^{\vc})} \cdot |V(G)|^{O(1)}$-time algorithm, where $\vc$ is the size of the minimum vertex cover of $G$.
\end{restatable}
Thus, the case of bounded vertex cover draws a clear line between \swshort and \pwshort.

Finally, one may ask whether small-sized distance preservers may be computed efficiently.
Since the solution size is always at least the number of terminals, this is a strictly more restrictive parameter, hence the result of Theorem~\ref{thm:sw-w1-terminals} is not sufficient. By a simple reduction from \textsc{Clique}, we observe that \swshort and \pwshort are both \wh{1} when parameterized by the solution size, in general graphs. 
\begin{restatable}{theorem}{SwSolnWHard}
    \label{thm:sw-w1-hard-soln}
    \swshort\ is \wh{1} parameterized by solution size. 
\end{restatable}
It is clear, however, that the result of Theorem~\ref{thm:sw-w1-hard-soln} cannot be strengthened to the same graph class as in Theorem~\ref{thm:sw-w1-terminals}. Indeed, if the solution size is bounded, then both the number of terminals and the distance between them is bounded, hence only vertices in small balls around each terminal are relevant. On subgraphs of the grid, these balls only contain a bounded number of vertices, which would immediately lead to a polynomial kernel for the problem.

\begin{table}[t]
\centering
\caption{Parameterized complexity of pairwise and subsetwise distance preservers on unweighted undirected graphs.}
\label{tab:complexity-summary}
\begin{tabular}{|l|p{3cm}|p{3cm}|}
\hline
\textbf{Parameter}
& \multicolumn{1}{c|}{\textbf{PDP}} 
& \multicolumn{1}{c|}{\textbf{SDP}} \\
\hline

\multirow{3}{*}{Number of terminals}
& \multicolumn{2}{c|}{\makecell{W[1]-hard on subgraphs of grid\\ (Thm.~\ref{thm:sw-w1-terminals})}}\\
\cline{2-3}

& \multicolumn{2}{c|}{\makecell{FPT on grids \\ (Thm.~\ref{thm:pw-terminals-grid})}}\\
\hline

\makecell[l]{Treewidth $+$ Number of terminals}
& \multicolumn{2}{c|}{\makecell{FPT \\ (Thm.~\ref{thm:pw-tw-k-fpt})}}\\
\hline

Vertex cover
& \makecell{NP-hard for $\vc=3$ \\ (Thm.~\ref{thm:pdp-vc-paraNphardness})}
& \makecell{FPT \\ (Thm.~\ref{thm:sw-vc-fpt-thm})} \\
\hline

Solution size
& \multicolumn{2}{c|}{\makecell{W[1]-hard \\ (Thm.~\ref{thm:sw-w1-hard-soln})}}\\
\hline

\end{tabular}
\end{table}

\subparagraph*{Related work}
For a comprehensive survey of literature on spanners, we refer to the work of Ahmed et al.~\cite{AHMED2020100253}. Following the standard definitions, distance preservers can be interpreted as $(1, 0)$-spanners (no stretch, no additive error).

Bollob\'as et al.\ introduced $D$-preservers, which aim to preserve all distances above the threshold $D$, and established tight asymptotic bounds for several variants, including Steiner and directed versions~\cite{BBC2003}.
Coppersmith and Elkin then introduced the now-standard terminal-pair variants: pairwise preservers and sourcewise/subsetwise preservers, and developed general upper bounds together with matching lower-bound regimes showing when linear-size preservers are possible~\cite{CD2006}.
Bodwin and Williams proved tightness of earlier upper bounds on weighted instances and established improved general upper bounds for the pairwise setting, sharpening the dependence on the number of terminal pairs~\cite{bodwin2016better}.
More recently, Bodwin obtained further bounds for directed preservers and for dense terminal-pair sets in undirected unweighted graphs (in terms of the Ruzsa--Szemer\'edi function), and refined lower bounds for subsetwise preservers~\cite{bodwin2021new}.
Kogan and Parter connect preservers with the study of spanners and hopsets, giving additional evidence that understanding preservers can inform the design of other distance-approximation and distance-sparsification objects~\cite{kogan2025having}.

The problem of finding the minimum $t$-spanner, i.e., the spanner with a multiplicative distortion of $t$, has been studied more extensively, also from the parameterized complexity perspective. 
The NP-hardness of the problem is due to Cai~\cite{Cai94} and Peleg and Sch\"{a}ffer~\cite{PelegS89}. Venkatesan et al. \cite{VenkatesanRMMP97} studied \textsc{Minimum $t$-Spanner} on various graph classes, including chordal graphs, convex bipartite graphs and split graphs, showing NP-hardness or polynomial-time algorithms depending on the value of $t$. Brandes and Handke showed NP-hardness of \textsc{Minimum $t$-Spanner} for $t \geq 5$ on planar graphs~\cite{BrandesH97}, and Kobayashi extended the hardness to $t \in \{2, 3, 4\}$, providing also an FPT algorithm when the parameter is the edge difference between the input graph and the spanner~\cite{KOBAYASHI201888}. Fomin et al. extended this perspective into the directed setting, showing an FPT algorithm when parameterized by $t$ plus the edge difference, and also W[1]-hardness for finding the minimum $t$-additive spanner with the same parameter~\cite{FominGLMSS22}.

\subparagraph*{Organization of the paper.} We introduce the necessary preliminaries in Section~\ref{sec:prelims}. Then, we present our results for grids and subgraphs of the grid in Section~\ref{sec:grids}. Section~\ref{sec:tw} is dedicated to the FPT algorithm parameterized by the treewidth and the number of terminals. Section~\ref{sec:vc} continues with the vertex cover results. The hardness proof for solution size is presented in Section~\ref{sec:solution_size}. We finish with a conclusion and open problems in Section~\ref{sec:conclusion}.

\section{Preliminaries}\label{sec:prelims}
We begin by introducing standard notation and concepts used throughout the paper.
\smallskip
All graphs considered here are simple, undirected, and unweighted. 
For a graph $G=(V,E)$, we additionally use $V(G)$ and $E(G)$ to denote its vertex and edge sets, respectively.
For a vertex subset $X \subseteq V(G)$, we write $G[X]$ for the subgraph of $G$ induced by $X$, i.e.,
$V(G[X]) = X$ and $E(G[X]) = \{uv \in E(G) \mid u,v \in X\}$.
For a vertex $v \in V(G)$, let $N_G(v)$ denote its neighborhood, and let $\delta_G(v)$ denote the set of edges incident to $v$ in $G$.
For an edge $e \in E(G)$, $G - e$ denotes the graph obtained by removing $e$ from $G$, that is $(V, E\setminus\{e\})$.
For vertices $u, v \in V(G)$, let $d_G(u,v)$ denote the length of a shortest path between $u$ and $v$ in $G$. 
When the graph $G$ is clear from the context, we drop the subscript. 
For a path $\pi$, $|\pi|$ denotes its length.
We denote an edge with endpoints $u$ and $v$ by $uv$.

For integers $n \ge 1$, we write $[n] := \{1,2,\dots,n\}$, and for a set $X$, we denote its cardinality by $|X|$.
Formally, $(\mathcal{F}_v)_{v \in V(G)}$ denotes a vector of sets where the $v$-th component is $\mathcal{F}_v$.
We adopt the parentheses and subscript notation to emphasize the index set clearly.

For a subgraph $H \subseteq G$, \emph{size} of $H$, denoted by $|H|$, is the number of edges in $H$. 

We also recall the relevant concepts from parameterized complexity.
For more thorough explanations and details, see~\cite{CyganFKLMPPS15}.
\subparagraph*{Parameterized complexity.} 
A \emph{parameterized problem} is a decision problem where each instance comes with a parameter $k$. An algorithm is \emph{fixed-parameter 
tractable (\fpt)} if it solves the problem in time $f(k)\cdot \mathrm{poly}(|I|)$ for some computable function $f$. 
A problem is $W[1]$-hard if it is unlikely to admit an FPT algorithm. 

A \emph{parameterized reduction} from a parameterized problem $A$ to a parameterized problem $B$ is a polynomial-time transformation that maps an instance $(I,k)$ of $A$ to an instance $(I',k')$ of $B$ such that 
(i) $(I,k)$ is a \yes-instance if and only if $(I',k')$ is a \yes-instance, and 
(ii) $k' \le g(k)$ for some computable function $g$.
The key difference from classical polynomial-time reductions is that the parameter in the target instance is bounded by a function of the original parameter.

Many \fpt\ algorithms exploit \emph{bounded search trees}, where recursive  branching is guided by the parameter to ensure that the total number of branches and the depth of the tree is bounded by $f(k)$ for some computable function $f$.

With the notations in place, we formally define the pairwise and subsetwise distance preserver problems.

\defbox{\pwfull}{\pwshort}
{An undirected graph $G=(V, E)$, a set of unordered terminal pairs
$\mathcal{P} \subseteq \bigl\{ \{u,v\} \,\bigm|\,  u,v \in V,\, u \neq v \bigr\}$, and an integer $k$.}
{Does there exist a subgraph $H \subseteq G$ such that for every $p=\{u, v\} \in \mathcal{P}$, $d_G(u, v) = d_H(u, v)$, and $|E(H)|\leq k?$}

\defbox{\swfull}{\swshort}
{An undirected graph $G=(V, E)$, a subset of vertices $S \subseteq V $ called \emph{terminals}, and an integer $k$.}
{Does there exist a subgraph $H\subseteq G$ such that for every pair of vertices $u,v \in S$, $d_G(u, v) = d_H(u, v)$, and $|E(H)|\leq k?$}
From now on, we denote an instance of \pwfull\ by $(G, \mathcal{P})$, and an instance of \swfull\ by $(G, S)$.
We refer to \pwfull\ and \swfull\ as \pwshort\ and \swshort, respectively.

Note that \swshort\ is a special case of \pwshort\ where 
$\mathcal{P} = \Bigl\{ \{u,v\} \,\bigm|\, \forall u,v \in S,\, u \neq v \Bigr\}$. 
The key difference is that \pwshort\ allows an arbitrary set of terminal pairs, possibly encoding a complex underlying structure, while \swshort\ always encodes a clique on $S$. 

Unlike the \textsc{Steiner Tree} problem, which only preserves connectivity among terminals, a \pwshort\ preserves shortest paths and may contain cycles.

Finally, we record a few simple observations about optimal preservers that will be used throughout the paper.

\begin{restatable}{observation}{AllEdgesNeeded}
    \label{obs:all-edges-needed}
    Let $H$ be a feasible solution (distance preserver) for an instance of \pwshort\ (resp.\ \swshort) of size at most $h$.
    If there exists an edge $e \in E(H)$ that does not lie on any shortest path between a pair of terminals, then $H'=\bigl(V(H), E(H) \setminus \{e\}\bigr)$ is also a feasible solution of size at most $h-1$.
\end{restatable}

By Observation~\ref{obs:all-edges-needed}, we may assume without loss of generality that every edge of an optimal \pwshort\ (resp.\ \swshort) lies on a shortest path between terminals.

\begin{observation}
\label{obs:remove-nonterminals-not-on-shortest-path}
Let $H$ be a feasible solution (distance preserver) for an instance of \pwshort\ (resp.\ \swshort) of size at most $h$.
If a non-terminal vertex $v \in V(H)$ does not lie on any shortest path between terminals, then $H' := (V(H) \setminus \{v\}, E(H) \setminus \delta_H(v))$ is
also a feasible solution of size at most $h$.
\end{observation}

Hence, we can safely remove non-terminal vertices that do not participate in any shortest terminal path.

\section{Grids and subgraphs of the grid}\label{sec:grids}

This section is dedicated to our results for grids and subgraphs of the grid. We start with the W[1]-hardness for subgraphs of the grid when parameterized by the number of terminals (Theorem~\ref{thm:sw-w1-terminals}), which is the main result of this section. Then, we show the FPT algorithm for grids with the same parameter (Theorem~\ref{thm:pw-terminals-grid}). Finally, we observe that \pwfull is NP-hard on grids by a reduction from \rsafull.

\subsection{\wh{1}ness by the number of terminals}
In this section, we show that \swshort\ is \wh{1} when parameterized by the number of terminals.
We give a parameterized reduction from the \bicqfull (\bicqshort) problem, which is  \wh{1} parameterized by the solution size~\cite{CyganFKLMPPS15}.
We then extend the reduction to obtain \wh{1}ness when parameterized by the number of terminal pairs,
even when the input graph is restricted to be an induced subgraph of the grid.

More formally, we prove Theorem~\ref{thm:sw-w1-terminals}, restated for convenience.
\SwTerminalsWHard*
Since \swshort is a special case of \pwshort, the next corollary follows immediately.
\begin{restatable}
{corollary}{PwTerminalsWHard}
    \label{coll:pw-w1-hard}
    \pwshort\ is \wh{1} parameterized by the number of terminal pairs, even on induced subgraphs of the grid.
\end{restatable}

We begin by first reminding the reader of the \bicqfull (\bicqshort) problem.

\defbox{\bicqfull}{\bicqshort}
{A bipartite graph $G=(L \cup R, E)$, an integer $k$, and partitions of $L$ and $R$ each into $k$ independent sets $(L_1, L_2, \dotso, L_k)$ and $(R_1, R_2, \dotso, R_k)$.}
{Does there exist a set $X \subseteq L \cup R$ such that $|X \cap L_i| = 1$ and $|X \cap R_i| = 1$ for every $i \in [k]$, and $G[X]$ is bipartite clique (biclique)?}

The \bicqshort\ problem parameterized by solution size is known to be \wh{1}.
This follows directly via a standard parameterized reduction from \mccfull\ (\mccshort), which is \wh{1}~\cite{CyganFKLMPPS15}.
Given an instance of \mccshort, we create two copies of the vertex set of the input graph and connect every vertex to its copy and the copies of its neighbors.
Clearly a $2k$-colored clique in the \bicqshort instance maps to a $k$-colored clique in the \mccshort instance, and vice versa.

Note that the partition of the vertex set into independent sets can be viewed as a proper vertex coloring, where each independent set corresponds to a distinct color.
A solution to the problem is therefore a biclique containing exactly one vertex of each color, which justifies the name ``multi-colored biclique''.
Throughout this reduction, we will use the terms \emph{color class} to refer to the sets $L_i$ and $R_j$ for $i,j \in [k]$.
Without loss of generality, we assume that $|L_i|=|R_j|=p$ for every $i,j \in [k]$ and $p \in \N$, 
as otherwise we can add isolated dummy vertices to the smaller sets as needed.
Our proof of Theorem~\ref{thm:sw-w1-terminals} proceeds in two steps.
We first, in~\autoref{sec:core-reduction}, describe a \emph{core reduction} that maps an instance $I_{\bicqshort}$ to an instance of \swshort\ on a \emph{planar} graph $G'$ together with a terminal set $S$, in polynomial time.
Then, in a separate \emph{embedding step}, we show how to modify the construction so that $G'$ is a subgraph of a grid graph while preserving the relevant distance properties.

\subsubsection{Core reduction}
\label{sec:core-reduction}
In this subsection, we first describe the core reduction from \bicqshort\ to \swshort.
Then, we prove several important properties about the structure of shortest paths 
in the constructed graph along with  properties of a subsetwise distance preserver.
Finally we prove the correctness of the reduction.

\paragraph*{Gadget construction.}
Let $I_{\bicqshort}=(G=(L\cup R, E),k)$ be an instance of \bicqshort\ with $|L \cup R|=n$, $L=\bigcup_{i=1}^k L_i$ and $R=\bigcup_{j=1}^k R_j$.
We construct an instance of \swshort,\ $I_{\swshort}=(G', S)$ such that $|S|=4k$ and $|V(G')|,|E(G')| \in O(n^2k+nk^2)$, as follows.

We create two types of gadgets to represent the color classes: 
for each $i\in[k]$ an \emph{$L_i$-column gadget}, oriented vertically, and an \emph{$R_i$-row gadget}, oriented horizontally.
These gadgets host the vertices of the original graph and their boundaries are defined by terminal vertices. 
The interaction between column and row gadgets will encode the edges of $G$.

Let $\Delta$, $\alpha$, and $\ell$ be length parameters to be fixed later such that:
\[\alpha \ll \ell \ll \Delta \quad \text{ with } \quad  \alpha,\ell,\Delta \in O(k+|V(G)|+|E(G)|).\]
For each $i \in [k]$, we initialize the $L_i$-column gadget by adding two terminal vertices $l^{\text{top}}_i$ and $l^{\text{bot}}_i$ to $G'$.
Then for every vertex $v \in L_i$, we introduce two vertices $w^{\text{top}}_v$ and $w^{\text{bot}}_v$, and connect them with a path of length $\ell$.
Next we connect $w^{\text{top}}_v$ to $l^{\text{top}}_i$ and $w^{\text{bot}}_v$ to $l^{\text{bot}}_i$ each via a path of length $\Delta$.

\smallskip
Similarly for each $i \in [k]$, we initialize the $R_i$-row gadget by adding two terminal vertices $r^{\text{left}}_i$ and $r^{\text{right}}_i$ to $G'$.
For every vertex $u \in R_i$, we introduce two vertices $w^{\text{left}}_u$ and $w^{\text{right}}_u$, and connect them with a path of length $\ell$.
We connect $w^{\text{left}}_u$ to $r^{\text{left}}_i$, and $w^{\text{right}}_u$ to $r^{\text{right}}_i$ each via a path of length $\Delta$.

We refer to the terminals $\{l_i^{\text{top}},l_i^{\text{bot}},r_i^{\text{left}},r_i^{\text{right}}\mid i\in[k]\}$ as \emph{endpoints}, since they serve as the endpoints of the row and column gadgets.

\smallskip
To ensure control over shortest paths and facilitate the subsequent grid embedding, 
we introduce four \emph{connector vertices}: $c^{\text{top}}, c^{\text{bot}}, c^{\text{left}}$, and $c^{\text{right}}.$
We then make the gadget endpoints adjacent to their corresponding connector, by an edge. 
That is, for each $i\in[k]$:
\begin{itemize}[leftmargin=2em]
    \item We connect $c^{\text{top}}$ to $l^{\text{top}}_i$ and $c^{\text{bot}}$ to $l^{\text{bot}}_i$ via an edge.
    \item We connect $c^{\text{left}}$ to $r^{\text{left}}_i$ and $c^{\text{right}}$ to $r^{\text{right}}_i$ via an edge.
\end{itemize}

Finally, we connect the connector vertices in a cyclic manner:
$c^{\text{top}}-c^{\text{left}}-c^{\text{bot}}-c^{\text{right}}-c^{\text{top}}$
with paths of length $2\Delta$.

The set of terminals for $I_{\swshort}$ is defined as
\[
    S
    = \{ l_i^{\text{top}}, l_i^{\text{bot}}, r_i^{\text{left}}, r_i^{\text{right}} \mid i \in [k] \},
\]
in other words,  for each color class in $G$ there are two terminals in $G'$.
This completes the gadgets framework and the vertex set of $G'$.
Refer to Figure~\ref{fig:sw-w1-hard-terminals-core-reduction-overview} for an illustration of the gadgets and an overview of the construction.

    \begin{figure}[!htb]
        \centering
        \includegraphics[width=.9\textwidth, ]{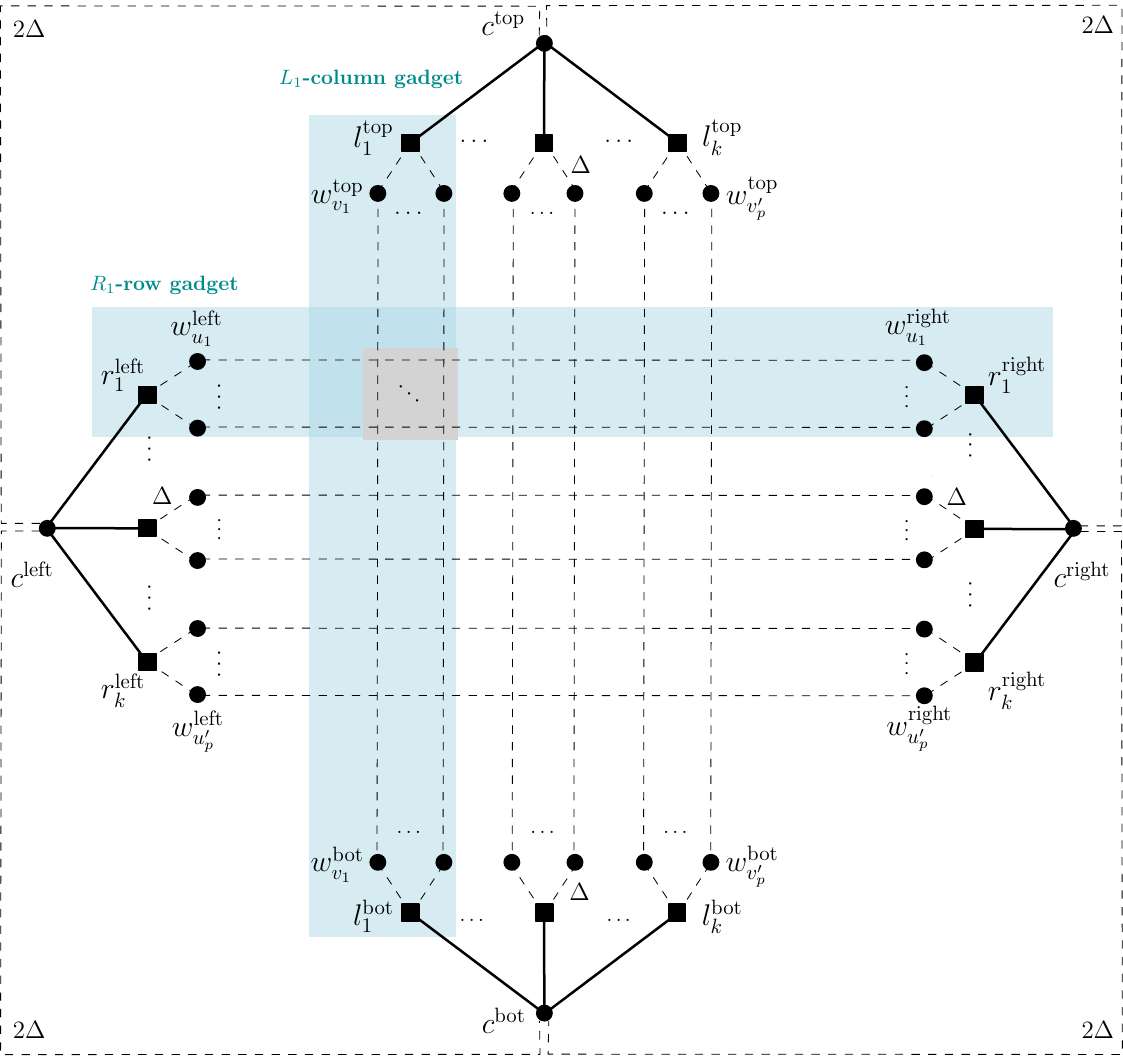}       
        \caption{
            The general overview of the construction.
            Terminal vertices are represented as squares, while the non-terminal vertices are represented as circles.
            Solid lines represent an edge while the dashed lines represent paths of length greater than one.
            For each color class in the \bicqshort\ instance, we create a corresponding gadget in the \swshort\ instance.
            The horizontal $R_k$-row and the vertical $L_k$-column gadgets are highlighted in light blue.
            Each $L_i$-column gadget has two terminals $l_i^{\text{top}}$, and $l_i^{\text{bot}}$ at its endpoints,
            and each $R_i$-row gadget similarly has two terminals $r_i^{\text{left}}$ and $r_i^{\text{right}}$ at its endpoints.
            All paths between \emph{connector} vertices have length $2\Delta$.
            Within each $L_i$-column gadget, for every $v\in L_i$, 
            the path between $w^{\text{top}}_v$ and $w^{\text{bot}}_v$ has length $\ell$, 
            the path between $w^{\text{top}}_v$ and $l^{\text{top}}_i$ has length $\Delta$, 
            and the path between $w^{\text{bot}}_v$ and $l^{\text{bot}}_i$ has also length $\Delta$.
            Analogously, within each $R_i$-row gadget, the same length properties hold for every $u \in R_i$. 
            Notice that endpoints of the gadgets are connected to the corresponding connector vertex via an edge.
            By the illustration, in $G$ it holds that $u \in R_1$, $u' \in R_k$, $v \in L_1$, and $v' \in L_k$.
            The gray highlighted area indicates the region where the paths in the $R_k$-row gadget 
            and the $L_k$-column gadget intersect, which
            corresponds to the edges between $R_k$ and $L_k$ in the original graph $G$.
        }
        \label{fig:sw-w1-hard-terminals-core-reduction-overview}
    \end{figure}
\smallskip
Next, we describe how these gadgets interact to represent the edge set of $G$.

For $u \in R$ and $v \in L$, let $\pi_u$ be the horizontal path from $w_u^{\text{left}}$ to $w_u^{\text{right}}$ 
and let $\pi_v$ be the vertical path from $w_v^{\text{top}}$ to $w_v^{\text{bot}}$, in $G'$.

In our construction, every horizontal path $\pi_u$ intersects every vertical path $\pi_v$.
We define the local graph structure at these intersection areas based on the adjacencies in $G$.

For every pair $(u,v)\in R\times L$, the vertices $z^1_{uv}, z^2_{uv}, z^3_{uv}$
are defined to be three \emph{consecutive} vertices on the vertical path $\pi_v$ in a bottom-to-top order.
Then, the intersections between $\pi_u$ and $\pi_v$ are defined as follows:
\begin{itemize}
    \item \textbf{Case $uv \in E(G)$}: 
    In this case, $\pi_u$ and $\pi_v$ share a common segment of length 2.
    Specifically, they intersect at all the three consecutive vertices $z^1_{uv}, z^2_{uv}, z^3_{uv}$ and
    share the edges $z^1_{uv}z^2_{uv}$ and $z^2_{uv}z^3_{uv}$.
    After this shared segment, $\pi_u$ continues two edges outside $\pi_v$.
    
    \smallskip
    \item \textbf{Case $uv \notin E(G)$}: 
    In this case, $\pi_u$ is routed to avoid collecting edges from $\pi_v$.
    Specifically, $\pi_u$ goes ``two edges up'' right before $\pi_v$, intersects $\pi_v$ at exactly one vertex $z^2_{uv}$, 
    and then continues one edge outside $\pi_v$. 
    This ensures that no edges are shared between $\pi_u$ and $\pi_v$ when $uv \notin E(G)$.
\end{itemize}
We refer to this triple, $z^1_{uv}, z^2_{uv}, z^3_{uv}$, as the \emph{local intersection area} of paths $\pi_u$ and $\pi_v$.
If $uv\in E(G)$ then $\pi_u$ collects the two edges $z^1_{uv}z^2_{uv}$ and $z^2_{uv}z^3_{uv}$ (so $\pi_u$ and $\pi_v$ share a length-2 segment),
whereas if $uv\notin E(G)$ then $\pi_u$ meets $\pi_v$ only at the middle vertex $z^2_{uv}$.

\medskip
Fix $v\in L$ and consider the order in which the horizontal paths $\{\pi_u : u\in R\}$ intersect $\pi_v$ from top to bottom.
For any two \emph{consecutive} horizontal paths $\pi_{u_1}$ (above) and $\pi_{u_2}$ (below) in this order, we subdivide $\pi_v$ by inserting a path of length $\alpha$
 between $z^3_{u_1v}$ and $z^1_{u_2v}$. 
Equivalently, along $\pi_v$ the local intersection areas appear as disjoint triples, separated by internally vertex-disjoint subpaths of length $\alpha$.

Moreover, before the first local intersection area and after the last one,
we subdivide each $\pi_u$ and each $\pi_v$ by inserting a path of length $\alpha$.

See figure~\ref{fig:sw-w1-hard-terminals-core-reduction-interaction} for an illustration of the gadget intersections and lengths.

Crucially, we construct these local intersections such that each horizontal path $\pi_u$ advances by exactly five edges 
to cross or pass one vertical path of $\pi_v$, regardless of whether the edge $uv$ exists in $G$ or not.
Moreover, the segments of length $\alpha$ along each vertical $\pi_v$ ensure that shortest paths do not benefit
from switching between different horizontal/vertical paths and gadgets.
Note that the presence of an edge in $G$ affects only the \emph{topology} of the intersection (shared edges vs. single vertex crossing),
not the length of each vertical or horizontal path.

\begingroup

\begin{figure}[!htb]
    \centering
    \begin{subfigure}[b]{0.24\textwidth}
        \centering
        \includegraphics[width=0.7\linewidth]{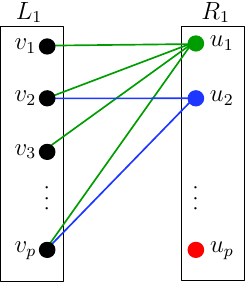}
        \caption{$G[L_1,R_1]$ with \\ $L_1=\{v_1,\dots,v_p\}$ and \\ $R_1=\{u_1,\dots,u_p\}$.}
        \label{fig:sw-w1-hard-terminals-core-reduction-bip-graph}
    \end{subfigure}
    \hfill
    \begin{subfigure}[b]{0.74\textwidth}
        \centering
        \includegraphics[width=\linewidth]{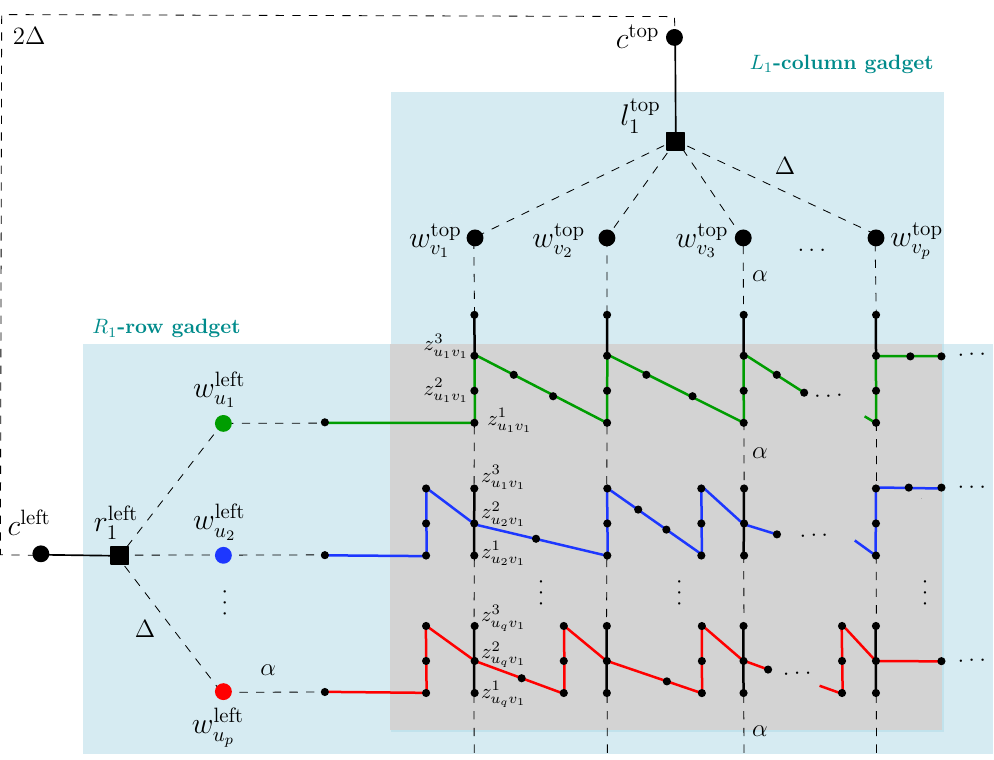}
        \caption{Local intersection pattern between the $R_1$-row gadget and the $L_1$-column gadget.}
        \label{fig:sw-w1-hard-terminals-core-reduction-interaction}
    \end{subfigure}

    \caption{
        Two views of the same $L_1$--$R_1$ interaction: 
        \textbf{(a)} the underlying induced bipartite graph $G[L_1,R_1]$, and \textbf{(b)} the corresponding local intersection structure in $G'$.
        Note that all horizontal paths intersect with all vertical paths at their corresponding local intersection areas.
        As $u_1$ is connected to all the vertices of $L_1$, in $G$, its corresponding horizontal path in $G'$ shares segments of length 2 with all vertical paths in the $L_1$-column gadget.   
        In contrast, $u_p$ is not connected to any vertex of $L_1$, so its horizontal path only intersects each vertical path at a single vertex.
        Vertex $u_2$ has a mixed interaction pattern, sharing segments of length 2 with vertical paths corresponding to $v_2$ and $v_p$
        while only intersecting at single vertices with others.
        Along each vertical path, the local intersection areas are separated by paths of length $\alpha$.
        Moreover, every vertical and horizontal path is subdivided by $\alpha$ at the beginning (and at the end).
        }
    \end{figure}
\endgroup
This concludes the construction of the graph $G'$ and terminal set $S$ for the \swshort\ instance $I_{\swshort}=(G', S)$.

To see that the reduction is polynomial in the size of $I_{\bicqshort}$, first we set the parameters $\alpha$, $\ell$, and $\Delta$.

First note that, $\alpha$ is inserted along every vertical path between consecutive local intersection areas, to prevent horizontal paths
from switching (using subpaths of vertical paths) to shorten their lengths.
Since each horizontal path pays $5$ edges per vertical path, by setting $\alpha \gg 5$, we ensure that switching to subpaths of vertical paths
makes the distance larger.
So $\alpha \in O(1)$ suffices.

Remember that $|R_i|=|L_i|=p$ for every $i \in [k]$ and let $n:=|V(G)|$.
In total, since there are $n$ paths and $2k$ gadgets, $\ell \in O(n+k)$ holds for each vertical path.
Without loss of generality, we can extend each horizontal path at the beginning and end to match this length.
Thus, setting $\ell \in O(n+k)$ suffices.
Finally, we set $\Delta \in O(n+k) \gg \ell$.

Now we bound $|S|$ and the size of $G'$.
By construction,
\[
    |S|=\bigl|\{l_i^{\text{top}},l_i^{\text{bot}},r_i^{\text{left}},r_i^{\text{right}} : i\in[k]\}\bigr|=4k.
\]

In every gadget, there are $p$ paths each of length $2\Delta + \ell \in O(n+k)$, which gives
$p \cdot O(n+k) = O(n(n+k))$ vertices and edges per gadget.
Since there are $2k$ gadgets, the total size of all gadgets is $O(k\cdot n(n+k)) = O(n^2k+nk^2)$.
By $\Delta \in O(n+k)$ we get $|V(G')|, |E(G')| \in O(n^2k+nk^2)$.

\medskip
\paragraph*{Structure of shortest paths in $G'$.} Before we prove the correctness of the reduction, we prove some important properties that describe
the shortest path structure in $G'$ between terminal pairs in $S$.

In the sequel, as mentioned in the construction, we assume $\alpha$ is chosen sufficiently large so
that no shortest path between terminals benefits from
detouring through multiple paths or gadgets.
Moreover, we assume $\ell<\Delta$ so that traversing a single row/column gadget
(length $2\Delta+\ell$ between its endpoints) is strictly shorter than going around half of the connector cycle (length $4\Delta$).

\begin{observation}
    \label{obs:sw-w1-hard-terminals-adjacent-same-connector}
    For every two terminals $s,t\in S$ adjacent to the same connector, every shortest $s$--$t$ path traverses the
    two edges connecting $s$ and $t$ to their common adjacent connector.
    
\end{observation}
The above observation follows directly from the construction, as each terminal is adjacent to exactly one connector vertex by an edge and $d_{G'}(s,t)=2.$ (Refer to \cref{fig:sw-w1-hard-terminals-core-reduction-overview})

We call two connector vertices \emph{consecutive} if they are consecutive on the connector--connector cycle. 
Then we have the following proposition:
\begin{proposition}
    \label{prop:sw-w1-hard-terminals-adjacent-consecutive-connectors}
    For every two terminals $s,t\in S$ adjacent to consecutive connectors, 
    every shortest $s$--$t$ path traverses the part of the connector cycle between the adjacent connectors.
\end{proposition}
\begin{proof} 
        Consider terminals $s$ and $t$ adjacent to $c^{\text{top}}$ and $c^{\text{left}}$ connectors.
        For any $i,j\in[k]$, traversing the part of the connectors cycle between their adjacent connectors gives the path
        \[
            l_i^{\text{top}}\to c^{\text{top}}\to c^{\text{left}}\to r_j^{\text{left}},
        \] 
        which has length $2\Delta+2$.
        Any other path between this pair that goes through gadgets has to change from a vertical to a horizontal direction,
        incurring a length at least $2\Delta+\alpha+\alpha$ 
        (as every horizontal and vertical path is subdivided by $\alpha$ at the beginning and the end)
        which is strictly larger.
        Similarly by symmetry:
        \[
            d_{G'}(l_i^{\text{top}},r_j^{\text{right}})=
            d_{G'}(l_i^{\text{bot}},r_j^{\text{left}})=
            d_{G'}(l_i^{\text{bot}},r_j^{\text{right}})=2\Delta+2.
        \]
\end{proof}

\begin{proposition}
    \label{prop:sw-w1-hard-terminals-same-gadget}
     For every two terminals $s=l_i^{\text{top}},\ t=l_i^{\text{bot}}$ or $s=r_i^{\text{left}},\ t=r_i^{\text{right}}$ where $i\in [k]$,
    every shortest $s$--$t$ path traverses one full internal path completely inside the corresponding column or row gadget.
\end{proposition}
\begin{proof}   
        For each $i\in[k]$, traversing a full path inside the gadget gives
        \[
            l_i^{\text{top}}\to(\text{one vertical path in }L_i)\to l_i^{\text{bot}}, \quad \text{ and } \quad r_i^{\text{left}}\to(\text{one horizontal path in }R_i)\to r_i^{\text{right}}
        \]
        both of which have length $2\Delta+\ell$.
        We prove the proposition for terminals in $L_i$-column gadget; the case for $R_i$-row gadget is symmetric.
        
        Let the above path be denoted by $\pi$, and let $\pi^*$ be a shortest $l_i^{\text{top}}\to l_i^{\text{bot}}$ path in $G'$.

        Assume $\pi^*$ switches between different vertical paths instead of traversing one full path.
        By construction, each such switching either
        (i) advances two horizontal edges while making (in the best case) only one unit of vertical progress, or
        (ii) advances two horizontal edges while making (in the best case) only two units of vertical progress.
        Both cases increase the path length by at least $1$ edge, hence $|\pi^*| \ge |\pi| + 1$, which is a contradiction.
        So, switching between different vertical paths does not yield a shorter $l_i^{\text{top}}$--$l_i^{\text{bot}}$ path 
        than staying on a single vertical path.

        For horizontal paths, this is even worse, as the vertical subpath between them is of length $\alpha \gg 5$, 
        so switching between different horizontal paths increases the path length by at least $\alpha-5$ edges, which is a contradiction as $\alpha$ is sufficiently large.
        
        Moreover by construction, all vertical paths between $l_i^{\text{top}}$ and $l_i^{\text{bot}}$ 
        have the same length $2\Delta+\ell$.

        Finally, any path that does not go via any gadget goes to a connector vertex, then takes half of the connector cycle to reach the opposite connector vertex and then enters the 
        other terminal; this has length at least $4\Delta+2$, which is strictly larger than $2\Delta+\ell$ since $\ell<\Delta$.

        Therefore, the proposition holds: traversing one full path inside the $L_i$-column gadget (analogously $R_i$-row gadget) 
        is the only shortest path and has length $2\Delta+\ell$.
\end{proof}

\begin{proposition}
\label{prop:sw-w1-hard-terminals-different-gadgets}
For every two terminals $s=l_i^{\text{top}},\ t=l_j^{\text{bot}}$ or $s=r_i^{\text{left}},\ t=r_j^{\text{right}}$ where $i,j\in [k]$,
every shortest $s$--$t$ path traverses one full internal path completely inside one of the two corresponding column gadgets or row gadgets.

\end{proposition}
\begin{proof}
    In particular, for all $i\neq j$,
\begin{align*}
d_{G'}(l_i^{\text{top}},l_j^{\text{bot}})\leq 2\Delta+\ell+2
    \quad
    &\text{via}\quad l_i^{\text{top}}\to(\text{traverse }L_i\text{-gadget})\to l_i^{\text{bot}}\to c^{\text{bot}}\to l_j^{\text{bot}}, \\
    &\text{or }\quad l_i^{\text{top}}\to c^{\text{top}} \to l_j^{\text{top}}\to(\text{traverse }L_j\text{-gadget})\to l_j^{\text{bot}},
\end{align*}
and likewise,
\begin{align*}
    d_{G'}(r_i^{\text{left}},r_j^{\text{right}})\leq 2\Delta+\ell+2
    \quad
    &\text{via}\quad r_i^{\text{left}}\to(\text{traverse }R_i\text{-gadget})\to r_i^{\text{right}}\to c^{\text{right}}\to r_j^{\text{right}}\\
    &\text{or }\quad r_i^{\text{left}}\to c^{\text{left}}\to r_j^{\text{left}}\to(\text{traverse }R_j\text{-gadget})\to r_j^{\text{right}}.
\end{align*}
The same holds, by symmetry, for the pairs $(l_i^{\text{bot}},l_j^{\text{top}})$ and $(r_i^{\text{right}},r_j^{\text{left}})$.

The proof is analogous to \autoref{prop:sw-w1-hard-terminals-same-gadget}:
the displayed paths give the upper bounds, while any alternative path either goes around the connector cycle (cost $\ge 4\Delta+2$)
or detours through additional gadgets/buffers (cost $\ge 2\Delta+\ell+2$ plus at least one $\alpha$ segment for horizontal gadgets), 
hence cannot be shorter under the standing parameter regime.
\end{proof}

With all the above properties of shortest paths between terminal pairs in $S$, 
we now state two important properties of a subsetwise preserver for $I_{\swshort}$.
\begin{corollary}
    \label{corr:sw-w1-hard-terminals-obvious-edges}
    Let $H$ be a distance preserver for the $I_{\swshort}$ instance.
    Then, by~\autoref{obs:sw-w1-hard-terminals-adjacent-same-connector} and~\autoref{prop:sw-w1-hard-terminals-adjacent-consecutive-connectors},
    $H$ must contain:
    \begin{itemize}
        \item all edges on the connector--connector cycle, and
        \item all edges between connector vertices  and gadget endpoints $S=\{ r_i^{\text{left}}, r_i^{\text{right}}, l_i^{\text{top}}, l_i^{\text{bot}} \mid i\in [k]\}$.
    \end{itemize}
\end{corollary}

(Terminology.) A subsetwise distance preserver $H$ is \emph{minimal} if no proper subgraph of $H$ is still
 a subsetwise distance preserver for $(G',S)$.

\begin{restatable}{lemma}{swWonePathPerGadget}
\label{lem:sw-w1-hard-terminals-min-preserver-one-path-per-gadget}
Let $H\subseteq G'$ be a \emph{minimal} subsetwise distance preserver for $I_{\swshort}$.
Then for every $i\in[k]$, $H$ contains  
\begin{itemize}
    \item the edges of exactly one  $l_i^{\text{top}} \to l_i^{\text{bot}}$ internal path inside the $L_i$-column gadget, and
    \item the edges of exactly one $r_i^{\text{left}} \to r_i^{\text{right}}$ internal path inside the $R_i$-row gadget.
\end{itemize}
\end{restatable}
\begin{proof}
Fix $i\in[k]$. Since $H$ preserves terminal distances, it must preserve $d_{G'}(l_i^{\text{top}},l_i^{\text{bot}})=2\Delta+\ell$.
By \autoref{prop:sw-w1-hard-terminals-same-gadget}, every shortest $l_i^{\text{top}}\to l_i^{\text{bot}}$ path 
is obtained by traversing one full internal path inside $L_i$,
so $H$ must contain the edges of (at least) one such internal path. The same argument applies to $r_i^{\text{left}}$ and $r_i^{\text{right}}$.

For uniqueness, suppose $H$ contains edges of two distinct internal $l_i^{\text{top}}\to l_i^{\text{bot}}$ paths $\pi_1$ , $\pi_2$ inside $L_i$.
Delete from $H$ all edges that belong exclusively to one of these two paths, namely $\pi_1$, obtaining $H'$.
Distances between terminals remain preserved: the pair $(l_i^{\text{top}},l_i^{\text{bot}})$ is still connected by $\pi_2$
(by \autoref{prop:sw-w1-hard-terminals-same-gadget}).
Any other terminal pair whose shortest path used $\pi_1$, can still use $\pi_2$ without changing length, as all internal paths in $L_i$ 
have the same length and the same endpoints by construction.

Thus $H'$ is still a distance preserver, contradicting minimality of $H$.
Hence exactly one internal path is present in each column gadget. The row-gadget case is symmetric.
\end{proof}

\paragraph*{Correctness of the reduction.}
With all the relevant shortest paths set, we now prove the correctness of the reduction.
\begin{restatable}{lemma}{sub-terminal-wh-core-correctness}
    \label{lem:sub-terminal-wh-core-correctness}
    Let $I_{\bicqshort}=\bigl(G=(L\cup R, E), k\bigr)$ be an instance of \bicqshort\ with
    $L=\bigcup_{i=1}^k L_i$ and $R=\bigcup_{j=1}^k R_j$.
    Let $I_{\swshort}=(G', S, k')$ be the instance of \swshort\ constructed as above.    
    Then, $I_{\bicqshort}$ has a biclique with $2k$ vertices if and only if $I_{\swshort}$ has a 
    subsetwise distance preserver of size at most $k' = 8\Delta + 4k + 4k\Delta + 2k\ell - 2k^2$.
\end{restatable} 
\begin{proof}
    To see the forward direction, suppose that $I_{\bicqshort}$ has a solution $X \subseteq L \cup R$
    with $|X|=2k$ vertices forming a biclique $C$.
    We construct a subsetwise distance preserver $H \subseteq G'$ for $I_{\swshort}$ as follows.

    \smallskip
     For each $i \in [k]$, let $v^*_i \in L_i \cap X$ and $u^*_i \in R_i \cap X$ be the vertices selected in the biclique. 
    Also, let $\pi_{u^*_i}$ be the horizontal path connecting $r_i^{\text{left}}$ to $r_i^{\text{right}}$ via $w^{\text{left}}_{u^*_i}$ 
          and $w^{\text{right}}_{u^*_i}$ in the $R_i$-row gadget and
          let $\pi_{v^*_j}$ be the vertical path connecting $l_j^{\text{top}}$ to $l_j^{\text{bot}}$ via $w^{\text{top}}_{v^*_j}$ 
          and $w^{\text{bot}}_{v^*_j}$ in the $L_j$-column gadget.
    
    We include in $H$ all the edges in $\pi_{u^*_i}$ and $\pi_{v^*_j}$, for all $i,j \in [k]$.
    Moreover, we take into $H$ all edges on the connector--connector cycle and the edges connecting the connector vertices to the gadget endpoints.
    
        To see that $H$ preserves all distances between terminal pairs in $S$,
        note that by~\autoref{obs:sw-w1-hard-terminals-adjacent-same-connector} and
        \autoref{prop:sw-w1-hard-terminals-adjacent-consecutive-connectors},
        the connector cycle together with the connector--endpoint edges, say $cx$, where $c$ is a connector and $x\in S$,
        preserves all shortest paths between terminals on the same side
        or on adjacent sides.
        By~\autoref{prop:sw-w1-hard-terminals-same-gadget} and \autoref{prop:sw-w1-hard-terminals-different-gadgets},
        shortest paths between terminals on opposite sides are realized via a full traversal of a single row/column gadget. Let $F$ be the induced subgraph on connector vertices and terminal set $S$, that is, $F$ exactly contains the connector cycle and the edges between connectors and terminals (Refer \cref{fig:sw-w1-hard-terminals-core-reduction-overview}). 
        Since $H$ contains the subgraph $F$ and for $i,j \in [k]$, it has $\pi_{u^*_i}$ and $\pi_{v^*_j}$ for traversal of
        row and column gadgets $R_i$ and $L_j$ , all terminal distances are preserved in $H$.

        By construction, the connector--connector cycle contributes $4\cdot 2\Delta = 8\Delta$ edges in total,
        and the connections between connectors and gadget endpoints contribute $4k$ edges in total to $H$.
        For every $i,j \in [k]$, since $u^*_i v^*_j \in E$ (as they form a biclique),
        $\pi_{u^*_i}$ shares two edges with the vertical path $\pi_{v^*_j}$ in each $L_j$-column gadget, with $j \in [k]$.
        
        Let $E^{\cap}$ be the set of all shared edges between the paths $\pi_{u^*_i}$ and $\pi_{v^*_j}$,
        for all $i,j \in [k]$.
        Then $|E^{\cap}| = 2k^2$, as each pair $\{i, j\}$ contributes two shared edges.

        Moreover, each $\pi_{u^*_i}$ (similarly $\pi_{v^*_j}$) has $2\Delta+\ell$ edges.         
        When summing over the lengths of all $\pi_{u^*_i}$ and $\pi_{v^*_j}$ paths, $E^{\cap}$ is counted twice.

        Thus, the total number of edges in $H$ is
        \[
            |E(H)|
                = 8\Delta + 4k + 2k(2\Delta+\ell) - |E^{\cap}|
                = 8\Delta + 4k + 4k\Delta + 2k\ell - 2k^2.
        \]
    For the backward direction, suppose that $I_{\swshort}$ has a subsetwise distance preserver $H \subseteq G'$ with at most 
    $k' = 8\Delta + 4k + 4k\Delta + 2k\ell - 2k^2$ edges.
    Without loss of generality, we can assume that $H$ is inclusion-wise minimal, as otherwise we could remove redundant edges.
    We show that $I_{\bicqshort}$ has a solution with $2k$ vertices forming a biclique.
    \smallskip

    By~\autoref{corr:sw-w1-hard-terminals-obvious-edges}, $H$ must include all edges on the connector--connector cycle and 
    the edges connecting the connector vertices to the gadget endpoints (or terminals), which is exactly the subgraph $F$ defined above. 
    This gives already $8\Delta + 4k$ edges in $H$.

    Thus, $H$ can include at most $4k\Delta + 2k\ell - 2k^2$ additional edges beyond the edges contributed by subgraph $F$.
    By~\autoref{lem:sw-w1-hard-terminals-min-preserver-one-path-per-gadget}, for each $i,j \in [k]$,
    $H$ contains exactly one horizontal path $\pi_{u^*_i}$ in the $R_i$-row gadget and exactly one vertical path $\pi_{v^*_j}$ 
    in the $L_j$-column gadget connecting the corresponding gadget endpoints.
    Each such path contributes at least $2\Delta+\ell$ edges, so the total contribution of these paths is at least $4k\Delta + 2k\ell$ edges.
    The $- 2k^2$ term in the budget $k'$ indicates that there must be at least $2k^2$ shared edges between these horizontal and vertical paths.
    There are $k$ horizontal and $k$ vertical paths in $H$, giving $k^2$ pairs in total, and each of these pairs can share at most 
    two edges (by the construction).
    Thus, in $H$, every horizontal path $\pi_{u^*_i}$ must share two edges with every vertical path $\pi_{v^*_j}$, for $i \in [k]$.
    By the construction of $G'$, this is possible only if for every $i,j \in [k]$, the edge $u^*_i v^*_j$ exists in $G$.
    Thus, the set $X = \{u^*_i, v^*_j \mid i,j \in [k]\}$ forms a biclique in $G$ with $2k$ vertices.
\end{proof}

\subsubsection{Embedding into the grid}
In this subsection, we show how to embed the planar reduction explained in subsection~\ref{sec:core-reduction} into a subgraph of a grid while
preserving the necessary distance properties.
This embedding is crucial for transferring hardness results from general planar graphs
to grid-like structures.

We work with the (integer) grid graph $\Gamma$ whose vertex set is $\mathbb{Z}^2$ and where two vertices are adjacent 
iff their $\ell_1$-distance is $1$.
Our embedding produces a finite subgraph $G^{\square}\subseteq \Gamma$. Note that a subgraph of the grid can always be 2-subdivided to become an induced subgraph of the grid twice the size. This modification multiplies all distances by two, and does not change the collection of shortest paths between any two original vertices. Hence it is sufficient to show a reduction that constructs a (not necessarily induced) subgraph of the grid.

\medskip
Remember that initially, in the \bicqshort problem, we have a bipartite graph $G=(R\cup L, E)$ 
with $R= \bigcup_{i=1}^k R_i$ and $L= \bigcup_{i=1}^k L_i$ 
such that $|L_i|=|R_j|=p$ for every $i,j \in [k]$.
Throughout, we assume both $p$ and $k$ are powers of two as otherwise we can add dummy isolated vertices to the smaller sets as needed.

Let $G'$ be the planar graph constructed in the core reduction (subsection~\ref{sec:core-reduction}), and let $G^{\square}$ denote the grid-embedded graph 
that we construct below.

First we briefly recall the main features of the construction of $G'$ that are relevant for the embedding step.

\medskip
The graph $G'$ contains $k$ row gadgets and $k$ column gadgets, hence $2k$ gadgets in total.
The \emph{gadget endpoints} in $G'$ are precisely the terminal vertices
\[
    \{\,l_i^{\text{top}},\, l_i^{\text{bot}},\, r_i^{\text{left}},\, r_i^{\text{right}} \mid i\in[k]\,\},
\]
i.e., the two endpoints of each $L_i$-column gadget and the two endpoints of each $R_i$-row gadget.

For every $u\in R$ the construction of $G'$ introduces two non-terminal vertices $w_u^{\text{left}}$ and $w_u^{\text{right}}$,
and for every $v\in L$ it introduces non-terminals $w_v^{\text{top}}$ and $w_v^{\text{bot}}$; 
we refer to these vertices collectively as the \emph{$w$-vertices}.

\smallskip
Within each $R_i$-row gadget, the relevant traversals between $r_i^{\text{left}}$ and $r_i^{\text{right}}$ are realized by 
(planar) paths that are intended to behave horizontally,
whereas within each $L_i$-column gadget, the traversals between $l_i^{\text{top}}$ and $l_i^{\text{bot}}$ are realized by (planar) 
paths intended to behave vertically.

\medskip
More importantly recall from the core reduction that terminal-to--$w$ distances in $G'$ are uniform within each gadget: for every $i\in[k]$,
every $u\in R_i$, and every $v\in L_i$,
\[
    d_{G'}\!\bigl(r_i^{\text{left}}, w_u^{\text{left}}\bigr)=d_{G'}\!\bigl(r_i^{\text{right}}, w_u^{\text{right}}\bigr)=\Delta
    \qquad\text{and}\qquad
    d_{G'}\!\bigl(l_i^{\text{top}}, w_v^{\text{top}}\bigr)=d_{G'}\!\bigl(l_i^{\text{bot}}, w_v^{\text{bot}}\bigr)=\Delta.
\]
This uniformity is crucial: it ensures that, when a shortest terminal-to-terminal path traverses a gadget, 
the choice of which horizontal/vertical path to use inside the gadget does not affect the resulting terminal distances.
In a grid graph, however, we cannot realize a length-$\Delta$ adjacency by a single long edge.
Accordingly, in the construction of $G^{\square}$, to achieve this uniformity, we realize only the endpoint--to--$w$ connection paths
by grid-embedded binary trees of equal root-to-leaf lengths; we do \emph{not} embed an entire row/column gadget as a binary tree.

\smallskip
Formally, for each gadget endpoint $t\in S$ and its associated set of $w$-vertices $W_t$ (on the same side), 
we replace the terminal--to--$w$ connections of $G'$ by a rooted binary tree in the grid, with root $t$ and leaves exactly $W_t$.
This implies that all $w$-vertices associated with the same endpoint are equidistant from it.
The only side-specific differences below are the placement (spacing) of the leaves and the resulting height/width bounds of the embedding region,
which we treat separately for the row and column gadgets.

\smallskip
To avoid ambiguity between the planar construction and the grid embedding, we adopt the following convention throughout this section:
unprimed symbols (e.g., $\pi_u,\pi_v$) refer to paths in $G'$, and their grid-embedded counterparts in $G^{\square}$ are denoted by the same symbol with superscript~$\square$
(e.g., $\pi_u^{\square},\pi_v^{\square}$).
Moreover, we set $\Delta' \gg 6kp$ to be a length parameter throughout this section.

\paragraph*{Terminal--to--$w$ embedding.}
\begin{description}[leftmargin=2em]
        \item[Row gadgets.]
        Let $W^{\text{left}}=\bigl\{w^{\text{left}}_{u} \in G' \mid u \in R\bigr\}$, and 
        $W^{\text{right}}=\bigl\{w^{\text{right}}_{u} \in G' \mid u \in R\bigr\}$
        be the set of all $w^{\text{left}}$-vertices and $w^{\text{right}}$-vertices on the left and right ends of the row gadgets in $G'$.

        \smallskip
        We place all vertices of $W^{\text{left}}$ on the grid column with first coordinate $0$, i.e., at points of the form $(0,y)$.
        Ordering $W^{\text{left}}$ from top to bottom as $\bigl(w^{\text{left}}_{u_1},\dots,w^{\text{left}}_{u_{kp}}\bigr)$, 
        we assign them coordinates $(0,y_1),\dots,(0,y_{kp})$ with $y_1=6kp$
        such that $y_1>y_2>\dots>y_{kp}$ and $y_i-y_{i+1}=6$ for every $i\in[kp-1]$, equivalently
        \[
            d_{\Gamma}\!\bigl((0,y_i),(0,y_{i+1})\bigr)=6.
        \]
        Let $\beta_1=\Delta'+4kp$; analogously, we place all vertices of $W^{\text{right}}$ on the grid column with first coordinate $\beta_1$,
        i.e., at points of the form $(\beta_1,y)$.
        Ordering $W^{\text{right}}$ from top to bottom as $\bigl(w^{\text{right}}_{u_1},\dots,w^{\text{right}}_{u_{kp}}\bigr)$,
        we assign them coordinates $(\beta_1,y_1),\dots,(\beta_1,y_{kp})$.
        Note that the $y$-coordinates of the vertices in $W^{\text{right}}$ are the same as those of the vertices in $W^{\text{left}}$, i.e., 
        for $i \in [kp]$, $w^{\text{left}}_{u_i}$  and $w^{\text{right}}_{u_i}$ are placed at the same grid rows.

        \smallskip
        This way, since $p$ is a power of two and, in the above top-to-bottom order, 
        the vertices corresponding to any fixed $R_i$-row gadget form a contiguous block of exactly $p$ consecutive vertices 
        in $W^{\text{left}}$ (and similarly in $W^{\text{right}}$), we may treat each such block independently.
        In particular, for every $i\in[k]$ we embed the \emph{connections} between $r_i^{\text{left}}$ and the $p$ vertices of the $i$-th block of 
        $W^{\text{left}}$ with a rooted binary tree contained in a grid region of height $\lceil \log(6p) \rceil$ and width $6(p-1)$,
        whose root is $r_i^{\text{left}}$ and whose leaves are exactly these $p$ $w^{\text{left}}$-vertices; 
        analogously, we embed the \emph{connections} between $r_i^{\text{right}}$ and the $i$-th block of $W^{\text{right}}$ 
        with a rooted binary tree
        with root $r_i^{\text{right}}$ and the same height and width.
        
        Figure~\ref{fig:grid-embedding-binary-tree-row} illustrates an example of the above embedding for two terminals $r_i^{\text{left}}$ 
        and $r_i^{\text{right}}$ each with $8$ $w$-vertices.

        Consequently, for every $i\in[k]$ and all $u,u'\in R_i$,
        \begin{align}
                \label{eq:grid-row-gadget-tree-distance}
                &d_{G^{\square}}\!\bigl(r_i^{\text{left}}, w_u^{\text{left}}\bigr)=d_{G^{\square}}\!\bigl(r_i^{\text{left}}, w_{u'}^{\text{left}}\bigr) = 6(p-1)/2+ \lceil \log(6p) \rceil
            \qquad\text{and}\qquad  \notag\\
            &d_{G^{\square}}\!\bigl(r_i^{\text{right}}, w_u^{\text{right}}\bigr)=d_{G^{\square}}\!\bigl(r_i^{\text{right}}, w_{u'}^{\text{right}}\bigr) = 6(p-1)/2+ \lceil \log(6p)\rceil.
        \end{align}

        Note that since all the $r_i^{\text{left}}$-vertices (similarly, all the $r_i^{\text{right}}$-vertices) are roots of binary trees of the same height
        and their corresponding $w$-vertices are placed at the same grid column, we have that all of the $r_i^{\text{left}}$-vertices 
        (similarly, all of the $r_i^{\text{right}}$-vertices) are placed at the same grid column.

        We connect all the $r_i^{\text{left}}$-vertices (similarly, all the $r_i^{\text{right}}$-vertices) with a vertical straight grid path.

    \item [Column gadgets.]
        Let $W^{\text{top}}=\bigl\{w^{\text{top}}_{v} \in G' \mid v \in L\bigr\}$ and
        $W^{\text{bot}}=\bigl\{w^{\text{bot}}_{v} \in G' \mid v \in L\bigr\}$
        be the sets of all $w^{\text{top}}$-vertices and $w^{\text{bot}}$-vertices on the top and bottom ends of the column gadgets in $G'$.

        Let $x_1=1+\Delta'/2$, and $\beta_2=3+y_1+\Delta'/2$.
        We place all vertices of $W^{\text{top}}$ on the grid row with second coordinate $\beta_2$, i.e., at points of the form $(x,\beta_2)$.
        Ordering $W^{\text{top}}$ from left to right as $\bigl(w^{\text{top}}_{v_1},\dots,w^{\text{top}}_{v_{kp}}\bigr)$,
        we assign them coordinates $(x_1,\beta_2),\dots,(x_{kp},\beta_2)$ such that $x_1<x_2<\dots<x_{kp}$ and
        $x_{i+1}-x_i=4$ for every $i\in[kp-1]$; equivalently, 
        \[
            d_{\Gamma}\!\bigl((x_i,\beta_2),(x_{i+1},\beta_2)\bigr)=4
        \]

        Let $\beta_3= \beta_2-\Delta'-6kp$; analogously, we place all vertices of $W^{\text{bot}}$ on the grid row with second coordinate $\beta_3$, i.e., at points of the form $(x,\beta_3)$.
        Ordering $W^{\text{bot}}$ from left to right as $\bigl(w^{\text{bot}}_{v_1},\dots,w^{\text{bot}}_{v_{kp}}\bigr)$,
        we assign them coordinates $(x_1,\beta_3),\dots,(x_{kp},\beta_3)$ such that $x_1<x_2<\dots<x_{kp}$.
        Note that the $x$-coordinates of the vertices in $W^{\text{bot}}$ are the same as those of the vertices in $W^{\text{top}}$, i.e.,
        for $i \in [kp]$, $w^{\text{top}}_{v_i}$ and $w^{\text{bot}}_{v_i}$ are placed on the same grid column.
        
        With the same argument as for the row gadgets, the vertices corresponding to any fixed $L_i$-column gadget form a contiguous block
        of exactly $p$ consecutive vertices in $W^{\text{top}}$ (and similarly in $W^{\text{bot}}$).
        For every $i\in[k]$ we embed the \emph{connections} between $l_i^{\text{top}}$ and the $p$ vertices of the $i$-th block of 
        $W^{\text{top}}$ with a rooted binary tree contained in a grid region of height $\lceil \log(4p) \rceil$ and width $4(p-1)$,
        whose root is $l_i^{\text{top}}$ and whose leaves are exactly these $p$ $w^{\text{top}}$-vertices; 
        analogously, we embed the \emph{connections} between $l_i^{\text{bot}}$ and the $i$-th block of $W^{\text{bot}}$ with a rooted binary tree
        with root $l_i^{\text{bot}}$ and the same height and width.
        Figure~\ref{fig:grid-embedding-binary-tree-column} illustrates an example of the above embedding for two terminals 
        $l_i^{\text{top}}$ and $l_i^{\text{bot}}$ each with $8$ $w$-vertices.

        Consequently, for every $i\in[k]$ and all $v,v'\in L_i$,
        \begin{align*}
            &d_{G^{\square}}\!\bigl(l_i^{\text{top}}, w_v^{\text{top}}\bigr)=d_{G^{\square}}\!\bigl(l_i^{\text{top}}, w_{v'}^{\text{top}}\bigr) = 4(p-1)/2+ \lceil \log(4p) \rceil
            \qquad\text{and}\qquad \\
            &d_{G^{\square}}\!\bigl(l_i^{\text{bot}}, w_v^{\text{bot}}\bigr)=d_{G^{\square}}\!\bigl(l_i^{\text{bot}}, w_{v'}^{\text{bot}}\bigr) = 4(p-1)/2+ \lceil \log(4p)\rceil.
        \end{align*}
        To make the terminal--to--$w$ distances equal to those in the row gadgets (\cref{eq:grid-row-gadget-tree-distance}), 
        for the binary trees corresponding to a column gadget, we attach each root by a straight vertical grid path of appropriate length.
        This way, all $l_i^{\text{top}}$-vertices (similarly, all $l_i^{\text{bot}}$-vertices) are placed at the same grid row.

        We connect all the $l_i^{\text{top}}$-vertices (similarly, all the $l_i^{\text{bot}}$-vertices) with a horizontal straight grid path.

        With the same reasoning as for the "row gadget" case, all $r_i^{\text{top}}$-vertices 
        (similarly, all $r_i^{\text{bot}}$-vertices) are placed at the same grid row.

\end{description}

\begingroup
\begin{figure}[!htb]
    \centering
    \begin{subfigure}[t]{0.6\textwidth}
        \centering
        \includegraphics[width=\linewidth]{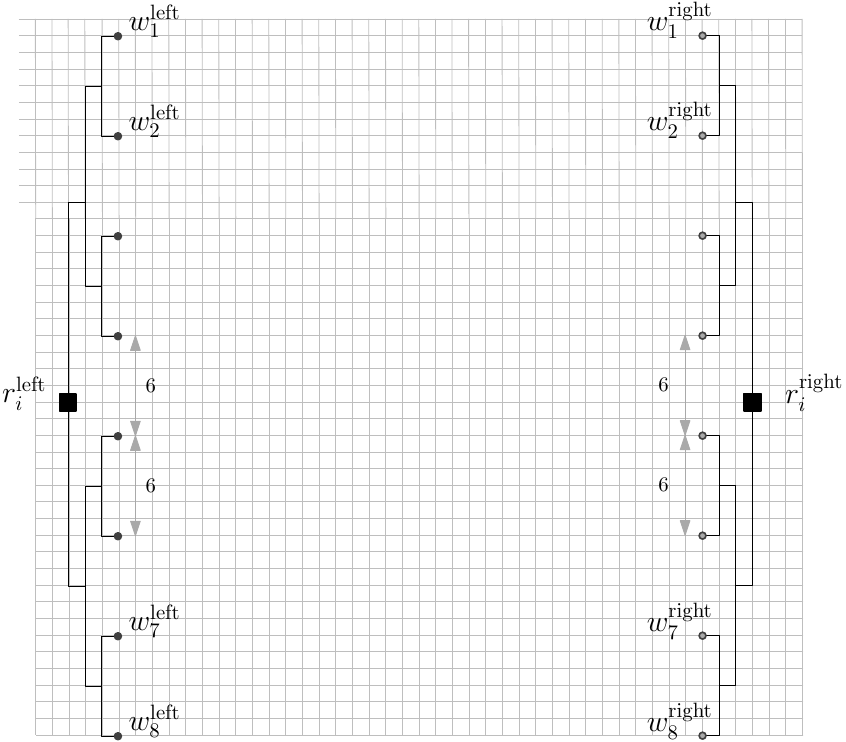}
        \caption{
            Example of the binary-tree embedding for the connections between $r_i^{\text{left}}$ and $r_i^{\text{right}}$
            and their corresponding $w$-vertices. 
        }
        \label{fig:grid-embedding-binary-tree-row}
    \end{subfigure}
    \hfill
    \begin{subfigure}[t]{0.38\textwidth}
        \centering
        \includegraphics[width=\linewidth]{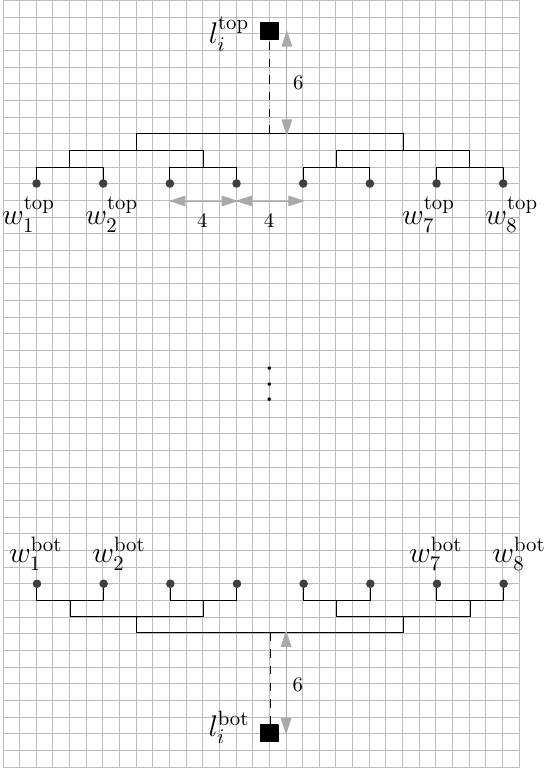}
        \caption{
            Example of the binary-tree embedding for the connections between $l_i^{\text{top}}$ and $l_i^{\text{bot}}$
            and their corresponding $w$-vertices.
        }
        \label{fig:grid-embedding-binary-tree-column}
    \end{subfigure}
    \caption{Binary-tree embeddings used to realize the terminal--to--$w$ connections in the grid embedding.
    In this example $|L_i|=|R_i|=8$. 
    In the horizontal trees, all leaves have the same distance of $23$ to their corresponding terminals $r_i^{\text{left}}$ and $r_i^{\text{right}}$.
    For root--to--leaf distances in the vertical trees to be also $23$, straight vertical paths of length $6$ are attached to connect the roots 
    to the remainder of the trees. 
    This equalizes the common terminal--to--$w$ distance.}
\end{figure}\endgroup

\medskip
    With the terminal--to--$w$ connections embedded as above, finally, we connect all the following pairs of gadget endpoints
    by simple grid paths of length $1.5 \Delta'$ to complete the main frame of $G^{\square}$:
    \[
       r_1^{\text{left}} \to l_1^{\text{top}}, 
       \quad l_k^{\text{top}} \to r_1^{\text{right}}, 
       \quad r_k^{\text{right}} \to l_k^{\text{bot}}, 
       \quad l_1^{\text{bot}} \to r_k^{\text{left}}
    \]
    
    Note that, by the binary tree embeddings and~\cref{eq:grid-row-gadget-tree-distance}, the grid distance between any of the above pairs is 
    at most $\Delta'+6kp+2\lceil \log(6kp) \rceil$.
    By assumption $\Delta' \gg 6kp$, so we can connect the above pairs by grid paths of the intended length $1.5 \Delta'$.
    Refer to~\cref{fig:grid-embedding-overview} for an overview of the whole resulting grid embedding.
\begingroup
\begin{figure}
    \centering
    \includegraphics[scale=0.6]{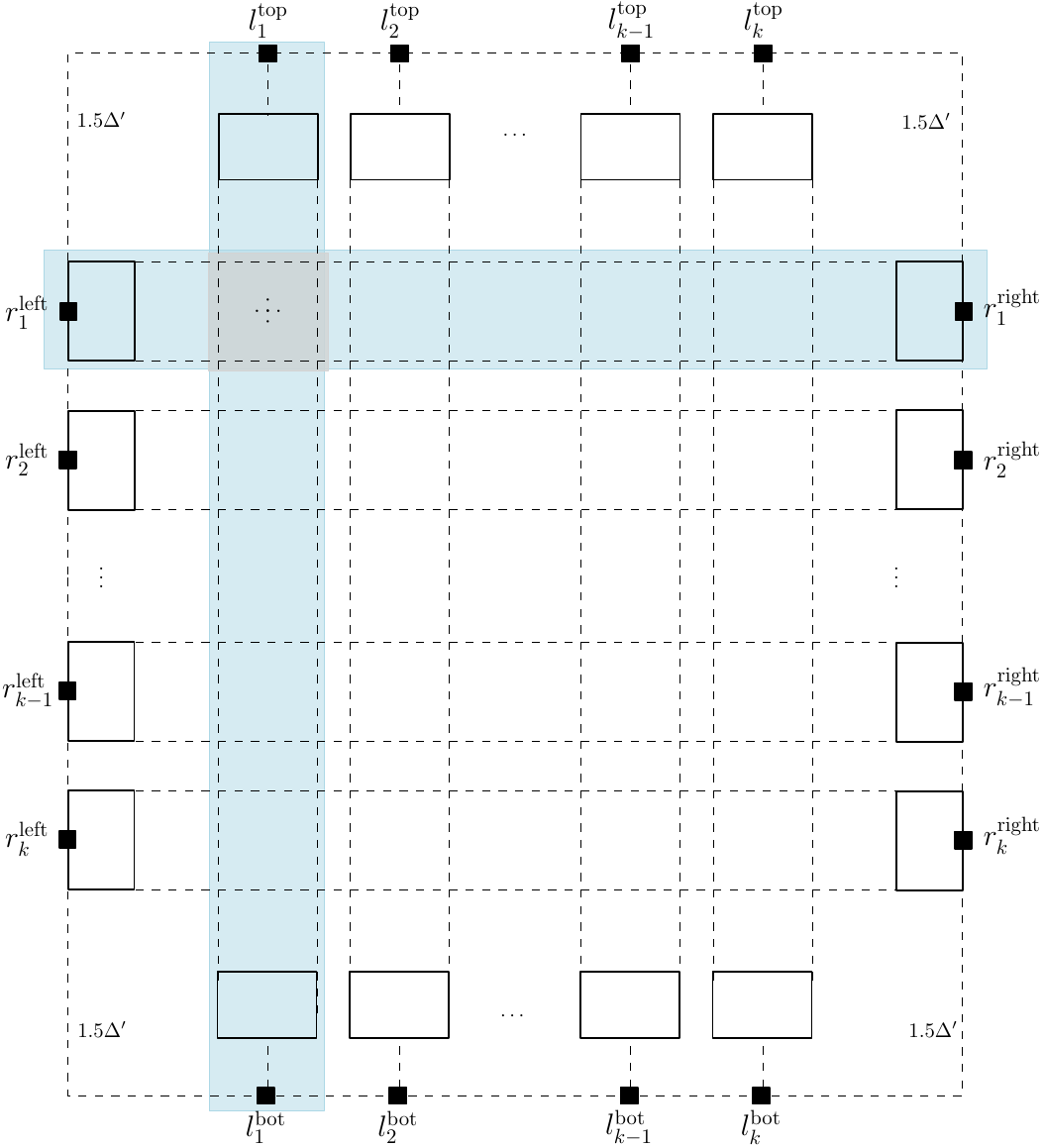}
    \caption{
        The general overview of the grid embedding of $G'$.
        Terminal vertices are represented as squares, while the non-terminal vertices are represented as circles.
        Dashed lines represent paths of length greater than one.
        The horizontal $R_k$-row and the vertical $L_k$-column gadgets are highlighted in light blue.
        Each rectangle represents a binary tree rooted at a terminal with leaves being the $w$-vertices, on the corresponding side:
        left, top, right and bottom.
        The top and bottom binary trees are connected to their corresponding root/terminal via paths of 
        appropriate lengths to equalize distances.
        All the $r_i^{\text{left}}$-vertices (similarly, all the $r_i^{\text{right}}$-vertices) are connected via a vertical straight grid path.
        All the $l_i^{\text{top}}$-vertices (similarly, all the $l_i^{\text{bot}}$-vertices) are connected via a horizontal straight grid path.
        All pairs $r_1^{\text{left}} \to l_1^{\text{top}}, l_k^{\text{top}} \to r_1^{\text{right}}, r_k^{\text{right}} \to l_k^{\text{bot}} $, and
        $ l_1^{\text{bot}} \to r_k^{\text{left}}$ are connected by grid paths of length $1.5 \Delta'$.
        }
    \label{fig:grid-embedding-overview}
\end{figure}
\endgroup

\medskip

With the main frame of the gadgets in place, we now embed the internal structure of each row/column gadget inside the corresponding grid region.
Before doing so, we first introduce a simple tool for realizing the length-$\Delta'$ paths in $G'$ by grid paths of the same length.
For convenience, assume $\Delta'$ is a multiple of $4$,
otherwise we increase it by at most $3$, which does not affect the asymptotic bounds and preserves the reduction.

\smallskip
\begin{definition}[rightward horizontal zig-zag path]
    A \emph{rightward horizontal zig-zag path} of length $\Delta'$ from a grid vertex $a$ is any simple grid path obtained by repeating the 
        pattern $(\rightarrow,\uparrow,\rightarrow,\downarrow)$, one edge in each direction, exactly $\Delta'/4$ times.
\end{definition}

This path has total length $\Delta'$ and its other endpoint lies exactly $\Delta'/2$ units to the right of $a$ (with the same $y$-coordinate),
i.e., their $\ell_1$-distance in $\Gamma$ is $\Delta'/2$.
Analogously, a leftward horizontal zig-zag path repeats $(\leftarrow,\uparrow,\leftarrow,\downarrow)$.
Similarly, a \emph{downward vertical zig-zag path} of length $\Delta'$ (resp.\ upward) is obtained by repeating $(\downarrow,\rightarrow,\downarrow,\leftarrow)$
(resp.\ $(\uparrow,\rightarrow,\uparrow,\leftarrow)$) exactly $\Delta'/4$ times; it has total length $\Delta'$ and its endpoint lies exactly $\Delta'/2$ units below (resp.\ above) the start.

\medskip
To embed the vertical paths in $G^{\square}$, for each $v\in L$, 
we first attach to $w^{\text{top}}_v$ a downward vertical zig-zag path of length $\Delta'$ and denote its endpoint by $t^{\text{top}}_v$.
Symmetrically, we attach to $w^{\text{bot}}_v$ an upward vertical zig-zag path of length $\Delta'$ and denote its endpoint by $t^{\text{bot}}_v$.
Thus,
\[
d_{G^{\square}}\!\bigl(w^{\text{top}}_v,t^{\text{top}}_v\bigr)=d_{G^{\square}}\!\bigl(w^{\text{bot}}_v,t^{\text{bot}}_v\bigr)=\Delta'.
\]
Then we connect $t^{\text{top}}_v$ and $t^{\text{bot}}_v$ by a vertical simple grid path $\pi_v^{\square}$ of length $8kp$ that repeats the pattern
 $(\downarrow,\downarrow,\downarrow,\downarrow, \leftarrow,\downarrow,\rightarrow,\downarrow)$ of length $8$ for each $u \in R$, which is
 exactly $kp$ times.

Note that the above pattern has length $8$, but advances only $6$ vertical units on the grid.
By construction, for every $v\in L$, its corresponding $w$-vertices, $w^{\text{top}}_v$ and $w^{\text{bot}}_v$, are placed on 
the same grid column, with vertical distance of $\Delta'+6kp$, which accommodates the zig-zag paths and the vertical path $\pi_v^{\square}$.

\medskip
To embed the horizontal paths in $G^{\square}$, for each $u\in R$, 
we first attach to $w^{\text{left}}_u$ a rightward horizontal zig-zag path of length $\Delta'$ and denote its endpoint by $t^{\text{left}}_u$.
Symmetrically, we attach to $w^{\text{right}}_u$ a leftward horizontal zig-zag path of length $\Delta'$ and denote its endpoint by $t^{\text{right}}_u$.
Thus,
\[
d_{G^{\square}}\!\bigl(w^{\text{left}}_u,t^{\text{left}}_u\bigr)=d_{G^{\square}}\!\bigl(w^{\text{right}}_u,t^{\text{right}}_u\bigr)=\Delta'.
\]
Then we connect $t^{\text{left}}_u$ and $t^{\text{right}}_u$ by a horizontal simple grid path $\pi_u^{\square}$ of length $8kp$ as follows.
Let $\{v_1, \dotso, v_{kp}\}$ be the vertices of $L$ ordered from left to right according to the order of their corresponding 
$w^{\text{top}}$-vertices in $W^{\text{top}}$ on the grid.
Fix $u \in R$. For each $i\in[kp]$,
 
\begin{itemize}
    \item If $uv_i \in E(G)$, we extend $\pi_u^{\square}$ by a path of length $8$ that goes one edge 
    for each of the following directions in order: (refer to~\cref{fig:grid-embedding-edge-intersection})
     \[
        (\rightarrow,\downarrow,\rightarrow,\uparrow, \uparrow,\rightarrow,\downarrow,\rightarrow)
    \]
    \item If $uv_i \notin E(G)$, we extend $\pi_u^{\square}$ by a path of length $8$ 
    that goes one edge for each of the following directions in order: (refer to~\cref{fig:grid-embedding-non-edge-intersection})
     \[
        (\rightarrow,\downarrow,\rightarrow,\uparrow,\rightarrow,\downarrow,\rightarrow, \uparrow)
    \]
\end{itemize}
        Note that each such pattern has total length $8$ but yields a net horizontal progress of $4$.
        Since $\pi_u^{\square}$ concatenates one such pattern for each of the $kp$ vertices of $L$, the total “core” horizontal displacement between
        $t_u^{\text{left}}$ and $t_u^{\text{right}}$ is $4kp$.
        By construction, $w_u^{\text{left}}$ and $w_u^{\text{right}}$ lie on the same grid row at horizontal distance $\Delta'+4kp$,
         which implies that their horizontal distance accommodates both zig-zag paths and the core path $\pi_u^{\square}$.
    \begingroup
    \begin{figure}[!htb]
        \centering
        \begin{subfigure}[b]{.49\textwidth}
            \centering
            \includegraphics[width=0.5\linewidth]{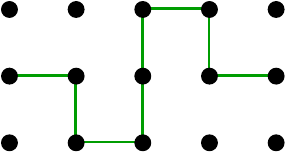}
            \caption{Path of length $8$ attached to $\pi_u$ when $uv_i\in E(G)$.}
            \label{fig:grid-embedding-edge-intersection}
        \end{subfigure}
        \hfill
        \begin{subfigure}[b]{0.49\textwidth}
            \centering
            \includegraphics[width=0.5\linewidth]{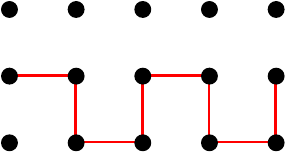}
            \caption{Path of length $8$ attached to $\pi_u$ when $uv_i\notin E(G)$.}
            \label{fig:grid-embedding-non-edge-intersection}
        \end{subfigure}

        \caption{
            The subpaths attached to construct $\pi_u$ for each $v_i\in L$ depending on whether $uv_i$ is an edge or not in $G$.
            The path patterns are
            $(\rightarrow,\downarrow,\rightarrow,\uparrow, \uparrow,\rightarrow,\downarrow,\rightarrow)$ and 
            $(\rightarrow,\downarrow,\rightarrow,\uparrow,\rightarrow,\downarrow,\rightarrow, \uparrow)$ respectively.
            }
        \end{figure}
    \endgroup        
        \medskip
        The above length-$8$ patterns are chosen so that, for every fixed $u\in R$, 
        the resulting grid path $\pi_u^{\square}$ meets the column-gadget paths
        \(\{\pi_{v}^{\square} : v\in L\}\) in the same left-to-right order as the corresponding planar path $\pi_u$ meets \(\{\pi_v : v\in L\}\) in $G'$.
        
        Moreover, for each $i\in[kp]$, the local intersection between $\pi_u^{\square}$ and $\pi_{v_i}^{\square}$ 
        depends only on whether $uv_i\in E(G)$:
        if $uv_i\in E(G)$ then the intersection subgraph contains exactly two shared grid edges, whereas if $uv_i\notin E(G)$ 
        then it contains exactly one shared grid edge.

        This is enforced by the fixed horizontal spacing of $4$ between consecutive vertices in $W^{\text{top}}$ (similarly in $W^{\text{bot}}$), and the
        vertical spacing of $6$ between consecutive vertices in $W^{\text{left}}$ ( similarly in $W^{\text{right}}$) 
        which align with each other  and provides sufficient room for the two distinct length-$8$ routing patterns.
        Figure~\ref{fig:grid-embedding-path-intersection} illustrates the local intersections explained above in $G^{\square}$.
\begingroup
\begin{figure}[!htb]
    \centering
    \begin{subfigure}[b]{0.22\textwidth}
        \centering
        \includegraphics[width=0.7\linewidth]{content/figures/sw-w1-hard-terminals-core-reduction-bip-graph.pdf}
        \caption{$G[L_1,R_1]$ with \\ $L_1=\{v_1,\dots,v_p\}$ and \\ $R_1=\{u_1,\dots,u_p\}$.}
        \label{fig:grid-embedding-bip-graph}
    \end{subfigure}
    \hfill
    \begin{subfigure}[b]{0.76\textwidth}
        \centering
        \includegraphics[width=\linewidth]{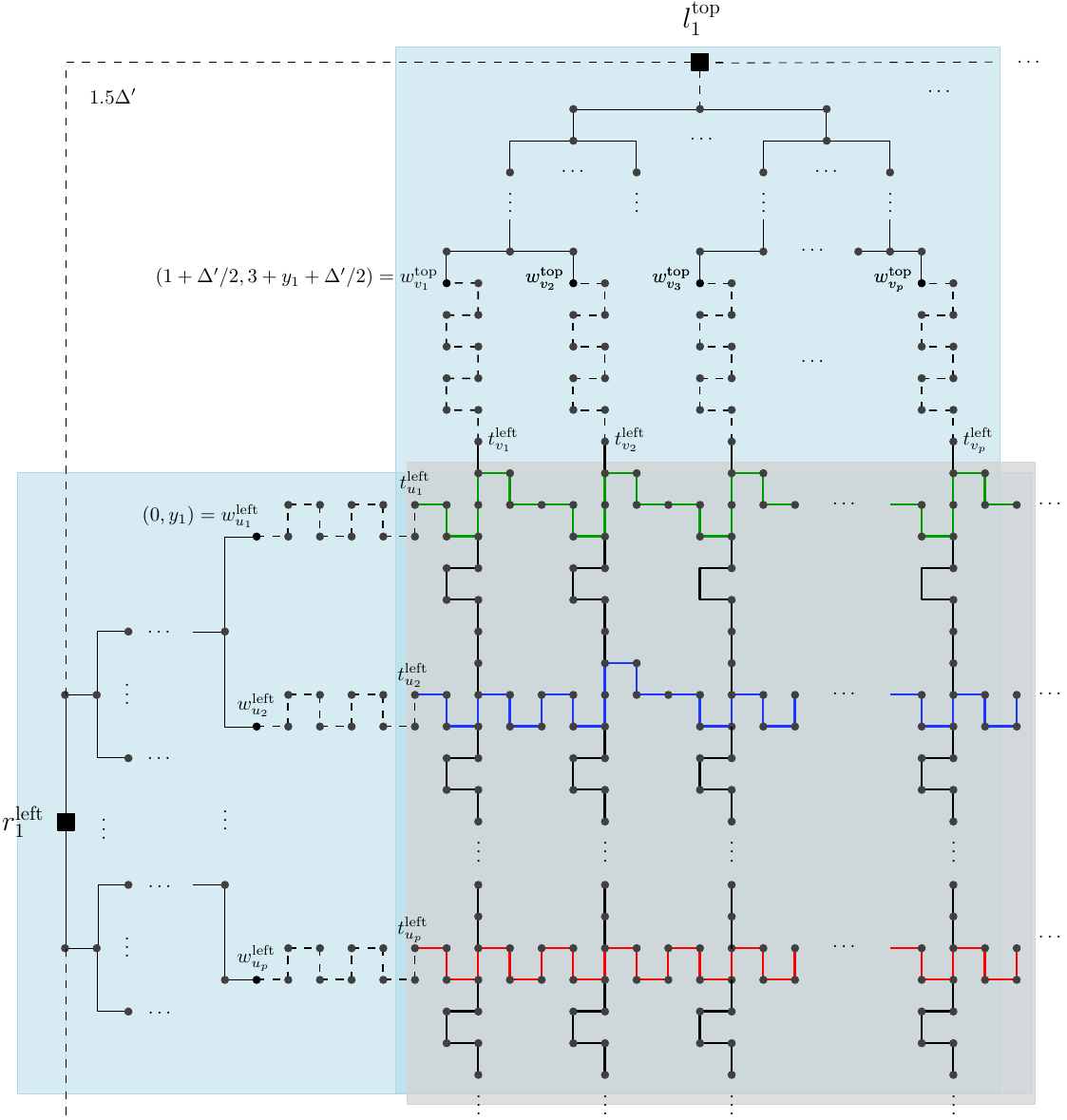}
        \caption{Local intersection pattern between the $R_1$ horizontal gadget and the $L_1$ vertical gadget in $G^{\square}$.}
        \label{fig:grid-embedding-path-intersection}
    \end{subfigure}
    \caption{
        Two views of the same $L_1$--$R_1$ interaction: 
        \textbf{(a)} the underlying induced bipartite graph $G[L_1,R_1]$, and \textbf{(b)} the corresponding local intersection structure in $G^{\square}$.
        Note that all horizontal paths intersect with all vertical paths at their corresponding local intersection areas.
        As $u_1$ is connected to all the vertices of $L_1$, in $G$, its corresponding horizontal path in $G'$ shares segments of length 2 with all vertical paths in the $L_1$-column gadget.   
        In contrast, $u_p$ is not connected to any vertex of $L_1$, so its horizontal path only shares one edge with each vertical path.
        Vertex $u_2$ has a mixed interaction pattern, sharing segments of length 2 with vertical paths corresponding to $v_2$ and $v_p$
        while only one edge with the others.
        The $t$-vertices are the endpoints of the zig-zag paths attached to the $w$-vertices.
        All the horizontal and vertical paths have the same length of $8kp+ 2\Delta'$.
        }
    \end{figure}
\endgroup

The set of terminals in $G^{\square}$ is exactly the same as that of $G'$, which 
are precisely 
\[
S=\{\,l_i^{\text{top}},\, l_i^{\text{bot}},\, r_i^{\text{left}},\, r_i^{\text{right}} \mid i\in[k]\,\},
\]
and this completes the construction of the grid embedding $G^{\square}$.

To see that the reduction is polynomial in the size of $I_{\bicqshort}$, first we set the parameter $\Delta'$.
In our construction we only need $\Delta' \gg c\cdot kp$ for some $c >6$.
So $\Delta' \in O(n)$.

Further note that the depth of all the binary trees
is of $O(\log(6kp))=O(\log(n))$.
So the number of vertices in each tree is $O(p)=O(n)$, and the size of the binary trees is polynomial in $n$.

All paths are of length $O(\Delta')+O(kp) \in O(n)$, so the size of the grid embedding is polynomial in $n$.

\paragraph*{Structure of Shortest Paths in $G^{\square}$.}
As the next step, we prove important properties that describe
the shortest path structure in $G^{\square}$ between terminal pairs in $S$.
We define $I^{\square}_{\swshort}=(G^{\square}, S)$ to be the \emph{instance} of \swshort\ induced by the embedded graph $G^{\square}$ and terminal set $S$.

\smallskip
In the sequel, we assume $\Delta' \gg 10kp$.
We define sides of terminals to be `left', `top', `right', and `bottom' according to their position in the construction, as illustrated in \cref{fig:grid-embedding-overview},
and by adjacent sides, we mean pairs of sides that are adjacent in the clockwise order (e.g., left and top are adjacent, while left and right are not).
Moreover, we call the paths of length $1.5 \Delta'$ connecting the four pairs 
    \[
       r_1^{\text{left}} \to l_1^{\text{top}} 
       \quad l_k^{\text{top}} \to r_1^{\text{right}} 
       \quad r_k^{\text{right}} \to l_k^{\text{bot}} 
       \quad l_1^{\text{bot}} \to r_k^{\text{left}}
    \]
the \emph{side connector} paths.

Analogously, we call every vertical gadget in $G^{\square}$ corresponding to $L_i$ an $L_i$-column gadget, 
and every horizontal gadget in $G^{\square}$ corresponding to $R_i$ an $R_i$-row gadget.

\begin{observation}
    \label{obs:grid-embedding-same-side-terminals}
    For every two terminals $s,t\in S$ on the same side, every shortest $s$--$t$ path traverses the straight grid path on the
    common row/column connecting them in $G^{\square}$.
\end{observation}
The above observation follows directly from the construction.
Indeed, for each side, all terminals on that side lie on a single grid row or a single grid column, and we explicitly include in $G^{\square}$
the straight grid segment connecting consecutive terminals on that line (see \cref{fig:grid-embedding-overview}).
Hence the straight $s$--$t$ subpath along that row/column is present in $G^{\square}$ and has length $d_{\Gamma}(s,t)$.
Since $G^{\square}\subseteq \Gamma$, every $s$--$t$ path in $G^{\square}$ has length at least $d_{\Gamma}(s,t)$, and therefore this straight subpath is
a shortest $s$--$t$ path in $G^{\square}$ (cf.~\cref{fig:grid-embedding-overview}).

\begin{proposition}
    \label{prop:grid-embedding-adjacent-side-terminals}
    For every two terminals $s,t\in S$, on the adjacent sides,
    every shortest $s$--$t$ path traverses their side connector path.
\end{proposition}
\begin{proof} 
        Consider one such terminal pair $s=r^{\text{left}}_{kp}$ and $t=l^{\text{top}}_{kp}$ 
        with maximum distance.
        By symmetry there are $8$ such pairs, for all of which the same proof applies.

        Any $s$--$t$ path either uses the corresponding side-connector path or enters (and later exits) at least one gadget region.

        On the one hand, by construction there is an $s$--$t$ walk using the side-connector path whose length is $1.5\Delta'+O(kp)$.

        On the other hand, any route that goes through gadgets must traverse at least two zig-zag attachment paths of length $\Delta'$ (one to enter some
        horizontal/vertical core and one to exit), and hence has length at least $2\Delta'$.
        Since $\Delta' \gg 10kp$, we have $2\Delta' > 1.5\Delta' + O(kp)$, so such a detour cannot be shortest.        
    \end{proof}

\begin{proposition}
    \label{prop:grid-embedding-same-gadget}
     For every two terminals $s=l_i^{\text{top}},\ t=l_i^{\text{bot}}$ or $s=r_i^{\text{left}},\ t=r_i^{\text{right}}$ where $i\in [k]$,
    every shortest $s$--$t$ path traverses one full internal path completely inside the corresponding vertical or horizontal gadget.
\end{proposition}
\begin{proof}   
        For each $i\in[k]$, traversing a full path inside the gadget gives
        \[
            l_i^{\text{top}}\to(\text{one vertical path })\to l_i^{\text{bot}}, \quad \text{ and } \quad r_i^{\text{left}}\to(\text{one horizontal path })\to r_i^{\text{right}}
        \]
        both of which have length $2\Delta'+8kp$.
        We prove the proposition for terminals in a vertical gadget; the case for a horizontal gadget is symmetric.
        
        Let the above path be denoted by $\pi$, and let $\pi^*$ be a shortest $l_i^{\text{top}}\to l_i^{\text{bot}}$ path in $G^{\square}$.

        Assume $\pi^*$ switches between different vertical paths instead of traversing one full path.
        By construction, each such switching advances $4$ edges horizontally while making (in the best case) only $2$ units of vertical progress,
        hence increasing the path length by at least $2$ edges, which is a contradiction as $\pi^*$ is a shortest path.

        For horizontal paths, this is even worse, as the vertical subpaths between them are of length $6$ without any horizontal progress, 
        so switching between different horizontal paths increases the path length by at least $5$ edges.
        
        Moreover by construction, all vertical paths between $l_i^{\text{top}}$ and $l_i^{\text{bot}}$ 
        have the same length $2\Delta'+8kp$.

        Finally, any path that does not go via any gadget traverses at least two side connector paths which gives the length at 
        least $3\Delta'$ and is strictly larger than $2\Delta'+8kp$ since $\Delta'$ is sufficiently large.

        Therefore, the proposition holds: traversing one full path inside the vertical gadget corresponding to $L_i$ 
        (analogously horizontal gadget corresponding to $R_i$) 
        is the only shortest path and has length $2\Delta'+8kp$.
\end{proof}

\begin{proposition}
\label{prop:grid-embedding-different-gadgets}
For every two terminals $s=l_i^{\text{top}},\ t=l_j^{\text{bot}}$ or $s=r_i^{\text{left}},\ t=r_j^{\text{right}}$ where $i,j\in [k]$,
every shortest $s$--$t$ path traverses one full internal path completely inside one of the two corresponding vertical gadgets or horizontal gadgets.

\end{proposition}
\begin{proof}
    In particular, for all $i\neq j$,
\begin{align*}
d_{G^{\square}}(l_i^{\text{top}},l_j^{\text{bot}})&\leq 2\Delta'+\bigl(|i-j|\cdot(p+6)\bigr)+8kp \leq 2\Delta'+9kp+6k \\
    \quad
    &\text{via}\quad l_i^{\text{top}}\to(\text{traverse }L_i\text{-gadget})\to l_i^{\text{bot}}\to (\text{traverse on the common row})\to l_j^{\text{bot}}, \\
    &\text{or }\quad l_i^{\text{top}}\to (\text{traverse on the common row}) \to l_j^{\text{top}}\to(\text{traverse }L_j\text{-gadget})\to l_j^{\text{bot}},
\end{align*}
and likewise,
\begin{align*}
    d_{G^{\square}}(r_i^{\text{left}},r_j^{\text{right}})&\leq 2\Delta'+\bigl(|i-j|\cdot(p+4)\bigr)+8kp \leq 2\Delta'+9kp+4k \\
    &\text{via } r_i^{\text{left}}\to(\text{traverse }R_i\text{-gadget})\to r_i^{\text{right}}\to (\text{traverse on the common column})\to r_j^{\text{right}}\\
    &\text{or }\, r_i^{\text{left}}\to (\text{traverse on the common column}) \to r_j^{\text{left}}\to(\text{traverse }R_j\text{-gadget})\to r_j^{\text{right}}.
\end{align*}
The same holds, by symmetry, for the pairs $(l_i^{\text{bot}},l_j^{\text{top}})$ and $(r_i^{\text{right}},r_j^{\text{left}})$.

The proof is analogous to \autoref{prop:grid-embedding-same-gadget}.
\end{proof}

With all the above properties of shortest paths between terminal pairs in $S$, 
we now state two important properties of a subsetwise preserver for $I^{\square}_{\swshort}$.
\begin{corollary}
    \label{corr:grid-embedding-obvious-edges}
    Let $H$ be a distance preserver for $I^{\square}_{\swshort}$.
    Then, by~\autoref{obs:grid-embedding-same-side-terminals} and~\autoref{prop:grid-embedding-adjacent-side-terminals},
    $H$ must contain:
    \begin{itemize}
        \item all edges on the straight paths connecting terminals on the same sides, namely the two straight vertical paths 
        connecting 
        \[
            r_1^{\text{left}} \to r_{k}^{\text{left}} \quad \text{and} \quad r_1^{\text{right}} \to r_{k}^{\text{right}}
        \]  
        and two straight horizontal paths connecting
         \[
            l_1^{\text{top}} \to l_{k}^{\text{top}} \text{ and } l_1^{\text{bot}} \to l_{k}^{\text{bot}}
         \]
        \item all the connector paths between adjacent sides.
    \end{itemize}
\end{corollary}

Using the above properties, we can prove that every minimal subsetwise preserver for $I^{\square}_{\swshort}$ 
must contain exactly one full path inside each gadget, 
and no switching between different paths inside the same gadget is possible, that is:
\begin{lemma}
\label{lem:grid-embedding-min-preserver-one-path-per-gadget}
Let $H\subseteq G^{\square}$ be a \emph{minimal} subsetwise distance preserver for $I^{\square}_{\swshort}$.
Then for every $i\in[k]$, $H$ contains  
\begin{itemize}
    \item the edges of exactly one  $l_i^{\text{top}} \to l_i^{\text{bot}}$ internal path inside the $L_i$-column gadget, and
    \item the edges of exactly one $r_i^{\text{left}} \to r_i^{\text{right}}$ internal path inside the $R_i$-row gadget.
\end{itemize}
\end{lemma}
Which is proved analogously to~\autoref{lem:sw-w1-hard-terminals-min-preserver-one-path-per-gadget}
 using the properties of shortest paths in $G^{\square}$ described above.

\medskip

The observations and propositions above are the exact analogues of the corresponding shortest-path statements for $G'$
(with the planar parameters replaced by their embedded counterparts, in particular $\Delta$ by $\Delta'$ and the internal gadget traversal length by $2\Delta'+8kp$).
Moreover, the grid embedding preserves the edge-vs-non-edge distinction at each horizontal/vertical intersection in the following quantitative sense:
for every $u\in R$ and $v\in L$, the embedded paths $\pi_u^{\square}$ and $\pi_v^{\square}$ share exactly two grid edges if $uv\in E(G)$,
and exactly one grid edge if $uv\notin E(G)$ (hence the two cases differ by precisely one shared edge).

Consequently, a biclique in $G$ yields a savings of $k^2$
edges in the resulting preserver subgraph, and conversely achieving this full $k^2$ savings
implies that $G$ has a biclique. This concludes the proof of Theorem~\ref{thm:sw-w1-terminals}.

\FloatBarrier 
\subsection{FPT algorithm on grids}

Here, we prove Theorem~\ref{thm:pw-terminals-grid}, restated next.

\PwTerminalsGrid*

The main component of the proof is the following structural lemma. For a grid graph $G = (V, E)$ and a terminal set $S \subseteq V$, let the Hanan grid be the subgraph of $G$ that contains only the edges and vertices on the rows and columns occupied by vertices of $S$. We may assume that the topmost row of $G$ contains a terminal, since otherwise no vertex of that row may be on any shortest path between two terminals, and can be safely removed. The same argument holds for the other three sides of the grid. 
We call the vertices lying on the intersections of the rows and columns containing terminals \emph{intersections} of the Hanan grid; observe that every vertex of the Hanan grid that is not an intersection has degree $2$, and that terminals are intersections by definition. 
Refer to~\autoref{fig:Hanan-grid} for an illustration.

The original motivation behind the Hanan grid is that the optimal rectilinear Steiner tree of a collection of terminals on the grid may be assumed to be contained in the Hanan grid of these points~\cite{Hanan66}. We show that the same holds for minimum distance preservers on the grid.

\begin{figure}[!ht]
    \centering
    \begin{subfigure}[t]{.49\textwidth}
        \centering
        \includegraphics[width=0.65\textwidth]{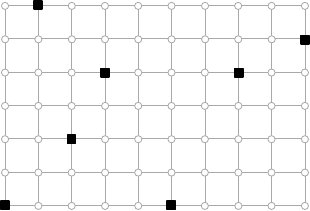}
        \caption{Grid graph $G=(E, V)$ with $|S|=6$ terminals.}
        \label{fig:Hanan-1}
    \end{subfigure}
        \hfill
    \begin{subfigure}[t]{.49\textwidth}
        \centering
        \includegraphics[width=0.65\textwidth]{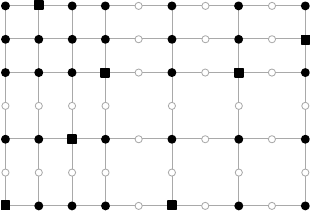}
        \caption{The Hanan grid of $G$ induced by $S$.}
    \label{fig:Hanan-2}
    \end{subfigure}
    \caption{An instance with a grid graph $G$ and its Hanan grid. Terminals are represented by black squares. The vertices of the grid graph $G$ are shown as circles with gray borders, while the intersections of the Hanan grid are represented as filled black circles. 
     Observe that the Hanan grid consists exclusively of the edges and vertices lying on rows and columns of $G$ that contain at least one terminal. Furthermore, any vertex of the Hanan grid that is not an intersection, has degree~$2$.}
    \label{fig:Hanan-grid}
\end{figure}
\begin{lemma}\label{lemma:Hanan}
    Let $G = (V, E)$ be a grid graph, and let $\mathcal{P}$ be a collection of pairs in $V$.
    There exists a minimum pairwise distance preserver of ($G$, $\mathcal{P}$) that is a subgraph of the Hanan grid of $S = \bigcup \mathcal{P}$.
\end{lemma}
\begin{proof}
    Let $H$ be a minimum distance preserver for $(G, \mathcal{P})$. We say that a vertex is a \emph{branching vertex of $H$}, if it has degree at least $3$ in $H$ or if it is a terminal. We say that $P$ is an elementary path in $H$ if it is an induced path of $H$, its endpoints are branching vertices, and no internal vertices are branching vertices. First we show a technical claim about elementary paths.
    \begin{claim}
    \label{claim:Hanan-all-disjoint}
        The elementary paths of $H$ are internally vertex-disjoint, and cover the edges of $H$.
    \end{claim}
    \begin{proof}
        Assume two distinct elementary paths $P_1$ and $P_2$ share an internal vertex $v'$. 
        Find inclusion-wise maximal subpath $P'$ of $P_1$ containing $v'$ that is also a subpath of $P_2$.
        Let $u$, $v$ be the endpoints of $P'$. If $P'$ is equal to $P_1$, then either it is also equal to $P_2$, or one of $u$ or $v$ is both an internal vertex of $P_2$ and a branching vertex, which is a contradiction.
    
        Otherwise, one of $u$ and $v$ is not the endpoint of $P_1$ and hence is not a branching vertex. Without loss of generality, let $u$ be this vertex, and let $w$ be the neighbor of $u$ on $P_1$ that is not on $P'$. By construction, $u$ belongs to $P_2$, and hence is either an endpoint of $P_2$, which is a contradiction since $u$ is an internal vertex of $P_1$, or $P_2$ continues into a vertex $t$ after $u$, and $w \ne t$ since $P'$ is the maximal subpath contained in both $P_1$ and $P_2$. Then, $u$ has degree at least $3$ in $H$, so $u$ is a branching vertex, contradicting that $P_1$ is an elementary path.
        
        For the second part of the claim, consider an edge of $H$. If both its endpoints are branching vertices, it is an elementary path by definition. Otherwise, take an endpoint that is not a branching vertex, it has degree two. Extend the current path with the other edge incident to this vertex. Continue until both endpoints will become branching vertices, at which point the path is an elementary path. Note that we cannot obtain a cycle, since then this cycle contains at most one terminal, 
        and removing the cycle from $H$ results in a smaller distance preserver, contradicting the minimality of $H$.
    \end{proof}
    
    We argue that it is sufficient to show that every branching vertex of $H$ is an intersection of the Hanan grid. Indeed, assume that this holds, and consider an elementary path $P$ in $H$ connecting two branching vertices $u$ and $v$ which is not contained in the Hanan grid. Since $u$ and $v$ are intersections of the Hanan grid, there also exists a shortest path $P'$ that connects $u$ and $v$ along the Hanan grid. Replacing $P$ by $P'$ in $H$ yields a new distance preserver that has at most as many edges as $H$, since the length of $P'$ is at most the length of $P$. By repeating this operation with every elementary path, we obtain a new distance preserver that is not larger and is contained in the Hanan grid. The latter holds since by~\autoref{claim:Hanan-all-disjoint} every edge of $H$ is covered by some elementary path.
    
    Hence, it remains to show that any branching vertex of $H$ is an intersection of the Hanan grid. Assume there exists a branching vertex $w$ in $H$ that is not an intersection. Hence, either the row of $w$ does not contain a terminal, or the column of $w$; in particular, $w$ is not a terminal but has degree $3$ or $4$ in $H$. Without loss of generality, assume the row of $w$ contains no terminal. Let $u$ and $v$ be the endpoints of the inclusion-wise maximal path in $H$ containing $w$ that is contained in the row of $w$, and let $P_{uv}$ be that path. Observe that the row of $w$ is neither on top nor on the bottom of the grid, since those rows contain terminals. Let $A$ be the vertices on the row above that are adjacent to the vertices of $P_{uv}$ in $H$; similarly, let $B$ be the vertices on the row below that are adjacent to $P_{uv}$ in $H$. See Figure~\ref{fig:grid-fpt-A-B} for an illustration.
    
    \begin{figure}[!ht]
        \centering
        \begin{subfigure}[t]{.315\textwidth}
            \centering
            \includegraphics[width=\textwidth]{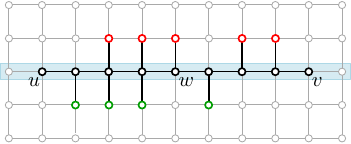}
            \caption{The vertex $w$, the path $P_{uv}$, and the sets $A$ and $B$.}
            \label{fig:grid-fpt-A-B}
        \end{subfigure}
            \hfill
        \begin{subfigure}[t]{.32\textwidth}
            \centering
            \includegraphics[width=\textwidth]{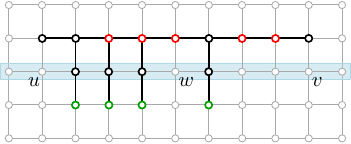}
            \caption{Constructing $H'$ by shifting $P_{uv}$ one row upwards.}
            \label{fig:grid-fpt-A}
        \end{subfigure}
        \hfill
        \begin{subfigure}[t]{.32\textwidth}
            \centering
            \includegraphics[width=\textwidth]{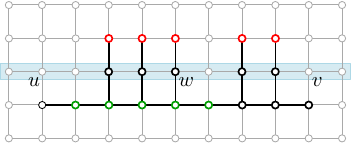}
            \caption{Constructing $H''$ by shifting $P_{uv}$ one row downwards.}
            \label{fig:grid-fpt-B}
        \end{subfigure}
        \caption{A branching vertex $w$ of $H$ that is not an intersection of the Hanan grid, the inclusion-wise maximal horizontal path $P_{uv}$ that contains $w$, sets $A$ and $B$, and modified graphs $H'$ and $H''$.
            The row of $w$ is highlighted in light blue. 
            The vertices of $A$ and $B$ are shown in red and green, respectively. 
            In \textbf{(a)}, the affected edges of $H$ are depicted in black, whereas in \textbf{(b)} and \textbf{(c)}, the edges of $H'$ and $H''$, respectively, are shown in black. 
            All remaining grid edges and vertices are shown in gray.
        }
        \label{fig:grid-fpt}
    \end{figure}
    
    We now consider two subgraphs of $G$, $H'$ and $H''$, illustrated in Figures~\ref{fig:grid-fpt-A} and~\ref{fig:grid-fpt-B}. $H'$ is obtained from $H$ by the following procedure: first, remove the edges of $P_{uv}$, and add edges parallel to them one row higher. Then, for each vertex of $B$ that has no vertex of $A$ in the same column, add an edge from the vertex one step above it to the vertex two steps above. For each vertex of $A$ that has no vertex of $B$ in the same column, remove the edge going downwards from this vertex. The construction of $H''$ is symmetric with edges of $P_{uv}$ shifted to the row below. 
    By construction, the size of $H'$ is at most $|E(H)| + |B| - |A|$, and the size of $H''$ is at most $|E(H)| + |A| - |B|$. We show that both $H'$ and $H''$ are also distance preservers.

    \begin{claim}
        $H'$ and $H''$ are pairwise distance preservers for $(G, \mathcal{P})$.
    \end{claim}
    \begin{proof}
        We show the claim for $H'$, the proof for $H''$ is analogous. Let $(s, t) \in \mathcal{P}$, and assume there is no shortest $(s, t)$-path in $H'$. Since there was a shortest $(s, t)$-path $P$ in $H$, it had to go through the vertices of $P_{uv}$, since only the edges adjacent to this path were removed in $H'$. The row of $w$, which contains $P_{uv}$, does not contain terminals, hence one of $s$ or $t$ is above it and one is below. Without loss of generality, let $s$ be above the row of $w$ and let $t$ be below. Let $s'$ be the last vertex of $P$ before intersecting $P_{uv}$, and let $t'$ be the first vertex on $P$ after leaving $P_{uv}$. Observe that $s'$ must be in $A$ and $t'$ must be in $B$. We now claim that $s'$ and $t'$ are shortest-path-connected in $H'$, which implies that $s$ and $t$ are also connected by the shortest path in $H'$. Indeed, consider the vertex $t''$ that is on the same column as $t'$ but two steps above, i.e., in the row of $s'$. $H'$ contains the two vertical edges between $t'$ and $t''$: the edge above $t'$ is contained by definition of $B$, then either $t''$ is in $A$ and the edge downwards from $t''$ thus belongs to $H$ and $H'$, or $t''$ is not in $A$, in which case this edge was added by construction of $H'$. Since $P_{uv}$ was lifted one row higher in $H'$, $t''$ and $s'$ are connected by the horizontal path in $H'$. This concludes the proof.
    \end{proof}

    Since both $H'$ and $H''$ are distance preservers, and since $H$ was a minimum distance preserver, it has to hold that $|A| = |B|$, and that $|E(H)| = |E(H')| = |E(H'')|$. Since this holds for any minimum distance preserver $H$ and any branching vertex $w$ that is not an intersection, we may apply the following procedure exhaustively. If $H$ has a branching vertex $w$ that is not an intersection, replace $H$ with $H'$ and repeat. (For $w$ that has no terminal in its column, we define $H'$ to be the one obtained by shifting the maximal path within the column leftwards.) Since these operations only shift branching vertices that are not intersections upwards and leftwards, but never in the opposite directions, the process is finite. Therefore, in the resulting distance preserver, all branching vertices lie on the intersections, which concludes the proof of the lemma.
\end{proof}

Finally, we now provide the proof of~\autoref{thm:pw-terminals-grid}.
\begin{proof}[Proof of~\autoref{thm:pw-terminals-grid}]
    By Lemma~\ref{lemma:Hanan}, there exists a solution that only uses the edges of the Hanan grid. Observe that the Hanan grid consists of $k^2$ intersections, connected by at most $2k^2$ straight-line segments. Clearly, any minimal distance preserver for $(G, \mathcal{P})$ either contains any of these straight-line segments completely or not at all. For each segment, branch whether it is included in the solution. By the above, one of the guesses will coincide with the minimum distance preserver. There are at most $2^{2k^2}$ choices for which segments to include, and each candidate solution can be verified in $|V(G)|^{O(1)}$ time. Hence, the algorithm that enumerates the choices and returns the smallest among those that are distance preservers, fulfills the statement of the theorem.
\end{proof}

For the running time of Theorem~\ref{thm:pw-terminals-grid}, observe that the grid lower bound construction of Krauthgamer et al.~\cite{Krauthgamer2014} shows that there is an instance where the minimum distance preserver is required to use $\Omega(k^2)$ vertices and edges of the Hanan grid.
\subsection{PDP is NP-hard on grids}
Here, we show that computing a minimum-size \pwshort\ is \nph, even when restricted to the special class of grid graphs.
Our proof proceeds by a polynomial-time reduction from the \rsafull\ problem, which is known to be NP-hard~\cite{RSA}.

\smallskip
Throughout this section, all points are assumed to lie on integer grid coordinates.
This restriction is without loss of generality.
Although integrality is not part of the original definition of \rsashort, the reduction from \textsc{3SAT}, given in~\cite{RSA}, produces only instances in which all terminal coordinates are integers.
Hence, the hardness result continues to hold under this assumption.

We recall the formal definition of the problem.
\defbox{\rsafull}{\rsashort}
{A set $P=\{p_1,p_2,\dots,p_n\} \subseteq \N^2$ of $n$ points, called terminals, in the first quadrant of the Cartesian plane, where $p_i=(x_i,y_i)$ for each $i$, and integer $k$.}
{Does there exist a directed Steiner tree $T=(V,E)$ rooted at the origin $(0,0)$ such that $P \cup \{(0,0)\} \subseteq V$, every edge in $E$ is a horizontal or vertical line segment directed left-to-right or bottom-to-top, and the total edge length of $T$
is at most $k$?}
Note that the underlying grid and the instance of \rsashort are undirected.
However, direction arises implicitly from the requirement that $T$ is a shortest-path tree rooted at $(0,0)$.
For any terminal $p=(x,y)\in P$, by definition, every shortest path from $(0,0)$ to $p$ consists of exactly $x$ horizontal steps to the right and $y$ vertical steps upwards.
Consequently, along any such path, each horizontal line segment connects a vertex with smaller $x$-coordinate to one with larger $x$-coordinate,
and each vertical line segment connects a vertex with smaller $y$-coordinate to one with larger $y$-coordinate.
We therefore orient every edge of $T$ in the direction induced by the root-to-terminal shortest paths, that is, left-to-right for horizontal edges and bottom-to-top for vertical edges.

Moreover, the key difference between an \rsashort\ and a rectilinear \textsc{Steiner} tree is that an \rsashort\ is additionally required to be a shortest-path tree with respect to the origin, whereas a Steiner tree is only required to preserve connectivity. 
~\autoref{fig:rsa-def} illustrates the edge directions implicitly induced by shortest paths from the root, as well as the distinction between an \rsashort
and a rectilinear \textsc{Steiner} tree.

\begin{observation}
\label{obs:rsa-shortest-path-length}
    Let $T$ be an \rsashort with terminal set $P=\{p_1,p_2,\dots,p_n\}$.
    For every terminal $p_i=(x_i,y_i)\in P$, the directed path from $(0,0)$ to $p_i$ in $T$ has length exactly $x_i+y_i$.
\end{observation}

In this section, we use $d(p_1,p_2)$ without subscript to denote the length of a shortest path from $p_1$ to $p_2$ on the underlying integer grid graph. 

\begin{figure}[h]
    \centering
    \begin{subfigure}[t]{.49\textwidth}
        \centering
        \includegraphics[width=0.75\textwidth]{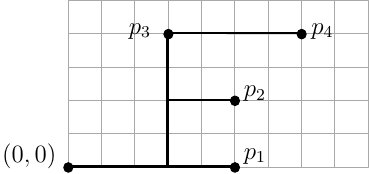}
        \caption{An \rsashort\ of total length $15$.}
        \label{fig:rsa}
    \end{subfigure}
        \hfill
    \begin{subfigure}[t]{.49\textwidth}
        \centering
        \includegraphics[width=0.75\textwidth]{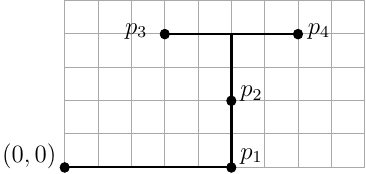}
        \caption{A rectilinear Steiner tree of total length $13$}
    \label{fig:steiner-tree}
    \end{subfigure}
    \caption{ \textbf{(a)} Illustration of an \rsashort\  and \textbf{(b)} illustration of a rectilinear \textsc{Steiner} tree for terminals $P = \{(5,0), (5,2), (3,4), (7,4)\}$.
    In the \rsashort, each edge on the unique path from the origin to a terminal
    $p_i = (x_i,y_i)$ is oriented left-to-right or bottom-to-top, and this path realizes the shortest possible length $x_i + y_i$. 
    In contrast, the rectilinear \textsc{Steiner} tree preserves connectivity but does not guarantee shortest paths; the path from $(0,0)$ to $p_3$
    has length $11 > d\bigl((0,0),p_3\bigr)=7$.}
    \label{fig:rsa-def}
\end{figure}

In~\cite{RSA}, authors show that \rsashort\ is \nph.
Using this result, we prove the following hardness theorem for \pwshort:

\begin{restatable}{theorem}{PwNpHardnessGrid}
\label{thm:pw-np-hardness-grid}
 \pwshort\ is \nph \ on grid graphs.
\end{restatable}
\begin{proof}
    We prove the theorem by a polynomial-time reduction from \rsashort.
    
    Let $I_{\rsashort}=(P, k)$ be an instance of \rsashort\ where $P=\{p_1,p_2,\dots,p_n\}$ and $p_i=(x_i, y_i)$.
    Let $p_0=(0,0)$ denote the origin in the Cartesian plane.
    We construct an instance $I_{\pwshort}=(\mathcal{P}', k')$ of \pwshort, by setting $k' = k$ and 
    \[
        \mathcal{P}' = \bigl\{\, \{p_0,p_i\} \mid i \in [n] \,\bigr\}.
    \]
    This completes the construction.
    
    We show that $I_{\rsashort}$ admits a solution of total length at most $k$ if and only if $I_{\pwshort}$ admits a solution of total length at most $k'$.
    
    To see the forward direction, assume that $I_{\rsashort}$ admits a solution $T$ of total length at most $k$.
    By definition, $T$ is a shortest-path tree rooted at $p_0$, and thus connects $p_0$ to every terminal $p_i\in P$ by a path of length $x_i+y_i=d(p_0,p_i)$.
    Therefore, $T$ preserves the distances for all pairs in $\mathcal{P}'$, and hence is a feasible solution for $I_{\pwshort}$ with total length at most $k=k'$.

    To see the other direction, assume that $I_{\pwshort}=(\mathcal{P}', k')$ admits a pairwise distance preserver $H$ of total length at most $k'$.

    By~\autoref{obs:all-edges-needed}, every edge $e \in E(H)$ lies on at least one shortest path from $p_0$ to some terminal $p_i$, with $i\in [n]$.
    We orient each edge along such a path, from $p_0$ towards $p_i$.
    This defines an implicit direction for every edge of $H$.

    \begin{claim}
        \label{claim:no-down-left}
        Every edge of $H$ is oriented either bottom-to-top or left-to-right.
    \end{claim}
    
    \begin{proof}
        Conversely, suppose $e$ is an edge in $H$ that is oriented top-to-bottom.
        By~\autoref{obs:all-edges-needed}, $e$ lies on some shortest path.
        If $e$ lies on a shortest path $\pi$ from $p_0$ to $p_i$, then
        \[
        d_{H}(p_0, p_i) \ge x_i + y_i + 1 = d(p_0, p_i)+1,
        \]
        contradicting the fact that $H$ preserves the shortest distance between $p_0$ and $p_i$.  
        Hence, $e$ does not lie on any shortest path and can be removed without affecting shortest distances and connectivity.   
        The argument for existence of a right-to-left edge is similar.
    \end{proof}
    
    \begin{corollary}
        Every edge of $H$ has a unique implicit direction, as determined by the shortest paths from $p_0$ to the terminals.
    \end{corollary}

    By Claim~\ref{claim:no-down-left}, we may assume that $H$ contains only left-to-right and bottom-to-top edges.
    Consequently, every path from $p_0$ to any vertex in $H$ is \emph{monotone}, that is, it contains no top-to-bottom or right-to-left steps.
    
    \begin{claim}
        \label{claim:no-cycle}
        If $H$ contains a cycle $C$, then there exists a pairwise distance preserver $H'$ with $|E(H')| \le |E(H)|$ and strictly fewer cycles.
    \end{claim}
    \begin{proof}
        Let $C$ be a cycle in $H$.
        By Claim~\ref{claim:no-down-left}, all edges of $C$ are oriented left-to-right or  bottom-to-top, and hence $C$ forms an orthogonal polygon.
        Let $a=(x_a,y_a)$ and $b=(x_b,y_b)$ denote the bottom-left and top-right corners of $C$, respectively.
        
        The cycle $C$ consists of two internally disjoint monotone paths $C^{\uparrow}$ and $C^{\downarrow}$ connecting $a$ to $b$.
        Since both paths are monotone, each contains exactly $(x_b-x_a)$ rightward and $(y_b-y_a)$ upward edges, and therefore $|C^{\uparrow}|=|C^{\downarrow}|$.

        Let $P_b \subseteq P$ be the set of terminals whose shortest paths from $p_0$ in $H$ pass through $b$.
        Since every path from $p_0$ to a terminal in $H$ is monotone, for every terminal $p=(x_p,y_p) \in P_b$ we have $x_p \ge x_b$ and $y_p \ge y_b$.
        Thus, for each such terminal $p \in P_b$, any shortest path from $p_0$ to $p$ may traverse $C^{\downarrow}$ instead of $C^{\uparrow}$ without increasing its length.
        An illustration of this configuration is given in Figure~\ref{fig:orthogonal-cycle}.
        
        Let $b_1=(x_b-1,y_b)$ be the immediate predecessor of $b$ along $C^{\uparrow}$.
        Removing the edge $b_1b$ preserves all required shortest distances from $p_0$ to all terminals in $P_b$ and eliminates the cycle $C$.
        Since no edges are added, no new cycles are introduced.
        See figure~\ref{fig:orthogonal-cycle-kill} for an illustration.
        
        \begin{figure}[h]
            \centering
            \begin{subfigure}[t]{.49\textwidth}
                \centering
                \includegraphics[width=0.77\textwidth]{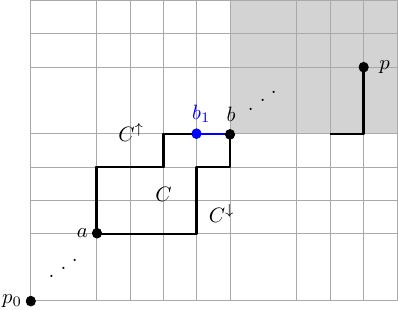}
                \caption{A \pwshort\ $H$ with a cycle $C$, which is an orthogonal polygon.}
                \label{fig:orthogonal-cycle}
            \end{subfigure}
                \hfill
            \begin{subfigure}[t]{.49\textwidth}
                \centering
                \includegraphics[width=0.77\textwidth]{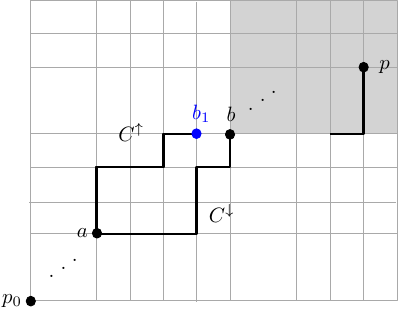}
                \caption{Obtaining another \pwshort\ $H'$ after eliminating cycle $C$ by removing edge $b_1b$.}
            \label{fig:orthogonal-cycle-kill}
            \end{subfigure}
            \caption{The gray area marks the potential positions of all terminals in $P_b$. \textbf{(a)} An orthogonal cycle $C$ in $H$, with bottom-left corner $a$, top-right corner $b$ and the two monotone subpaths $C^{\uparrow}$ and $C^{\downarrow}$ with $|C^{\uparrow}|=|C^{\downarrow}|$.
            Vertex  $b_1$ is the predecessor of $b$ along $C^{\uparrow}$.
            \textbf{(b)} $H$ after removing the edge $b_1b$; shortest paths from $p_0$ to terminals in $P_b$ are rerouted via $C^{\downarrow}$.}
            \label{fig:rsa-orthogonal-cycle}
        \end{figure}
    \end{proof}

    By repeatedly applying ~\autoref{claim:no-down-left} and~\autoref{claim:no-cycle}, we obtain a pairwise distance preserver $H$ of total length at most $k'$ that contains no top-to-bottom edges, no left-to-right edges, and no cycles.
    As a result, $H$ is a tree consisting only of bottom-to-top and left-to-right edges, and it connects $p_0$ to every terminal in $P$ via shortest paths. Hence, $H$ is a valid \rsashort of size at most $k' = k$.
    Finally, $\pwshort \in \mathrm{NP}$, since a candidate solution of size at most $k$ can be verified in polynomial time by checking that all required terminal distances are preserved.
    \end{proof}
\section{FPT by treewidth and the number of terminals}\label{sec:tw}
In this section we prove that \pwshort\ is fixed-parameter tractable parameterized by the treewidth of the input graph and the number of terminals.
More formally, we prove Theorem~\ref{thm:pw-tw-k-fpt}, restated next.
\pairwiseFPTbyTWK*

Intuitively, a graph of treewidth at most $\tau$ can be decomposed in a tree-like fashion,  where at any point only a small set of at most $\tau+1$ vertices—the \emph{bag}—needs to be considered. 
This bag acts as a separator between the portion of the graph already processed (the ``past'')  and the portion yet to be explored (the ``future''), so that all connections between past and future must go through the current bag. 

For the \pwshort\ instance, this corresponds to the property that terminals in the past can interact with the terminals in the future only via the vertices in the bag.

Our algorithm leverages this structure by performing dynamic programming over the tree decomposition, storing at each node a set of records that summarize the relevant connectivity and distance information for the vertices in the current bag.

We first recall the definition of a nice tree decomposition, which is a standard structure for performing dynamic programming in the setting of bounded treewidth.
\subparagraph*{\textbf{Nice Tree Decomposition.}}
A \emph{nice tree decomposition} of a graph $G=(V, E)$ is a pair $(\mathcal{T},\mathcal{X})$ where $\mathcal{T}$ is a rooted tree with root $r$ and $\mathcal{X}$ is a mapping that assigns to each node $t \in \mathcal{T}$ a set $X_t \subseteq V(G)$, referred to as the bag at node $t$.
A nice tree decomposition satisfies the following properties for each $t \in \mathcal{T}$:
\vspace{-0.3em}
\begin{enumerate}
    \setlength{\itemsep}{0.5em}
	\item $X_r = \emptyset$ and for every leaf $\ell$, it holds that $X_{\ell} = \emptyset$. \label{tw:leafs}
	\item For every $uv \in E$, there exists a node $t \in \mathcal{T}$ such that both $u \in X_t$ and $v \in X_t$. \label{tw:edges}
	\item For every $v \in V$, the set of nodes $t \in \mathcal{T}$ such that $v \in X_t$, induces a connected subtree in $\mathcal{T}$. \label{tw:connected}
	\item There are four kinds of nodes (aside from root and the leaves) in $\mathcal{T}$:
    \vspace{0.3em}        
	\begin{enumerate}
        \setlength{\itemsep}{0.5em}
		\item \textbf{Introduce vertex node}: a node $t$ with only one child $t'$ such that there is a vertex $v \notin X_{t'}$ satisfying $X_t = X_{t'} \cup \{v\}$. 
        We call $v$, the introduced vertex.
        \item \textbf{Introduce edge node}: a node $t$, labeled with an edge $uv \in E(G)$ such that $u, v \in X_t$, and with exactly one child $t'$ such that $X_t = X_{t'}$. 
        We say that edge $uv$ is introduced at $t$.
		\item \textbf{Forget node}: a $t$ with only one child $t'$ such that there is a vertex $v \notin X_t$ satisfying $X_t = X_{t'} \!\setminus\! \{v\}$.
        When a vertex is forgotten, then also all of its incident edges.
        We call $v$, the forgotten vertex.
		\item \textbf{Join node}: a node $t$ with exactly two children $t_1$ and $t_2$ such that $X_t=X_{t_1}=X_{t_2}$.
	\end{enumerate}
\end{enumerate}

Note that by property 3, while traversing from leaves to the root, a vertex $v \in V(G)$ cannot be introduced again after it has already been forgotten.
Otherwise the subtree of $\mathcal{T}$ induced by the nodes whose bags contain $v$ will be disconnected.
The \emph{width} of a nice tree decomposition $(\mathcal{T},\mathcal{X})$ is defined as $\max_{t \in \mathcal{T}} \bigl(|X_t|-1 \bigr)$ and the \emph{treewidth} of a graph $G$ is defined to be the smallest width of a nice tree decomposition of $G$ and is denoted by $\tw(G)$. 

One can use fixed-parameter tractable algorithms described in~\cite{bodlaender1993linear},~\cite{bodlaender2016c} and~\cite{kloks1994treewidth} to obtain a nice tree decomposition with almost optimal width and linearly many nodes.

Now we proceed with the formal proof.
Let $(G, \mathcal{P})$ be an instance of \pwshort, and assume that the treewidth of $G$ is at most $\tau$.
Let $S = \bigcup \mathcal{P} \subseteq V(G)$ denote the set of terminal vertices.

Let $\mathcal{T} = (T,\{X_t\}_{t\in V(T)})$ be a nice tree decomposition of $G$ of width $\tau$.
For a node $t\in T$, denote the subtree of $T$ rooted at $t$ by $T_t$.
We write $G_t$ for the subgraph of $G$ on all vertices and edges that appear in bags of $T_t$.
Furthermore, by $S_t^{\downarrow}$ we denote the set of all terminal vertices that appear in bags of $T_t$.
For each node $t\in T$, let $X_t$ denote its bag and define $U_t := X_t \cup S_t^{\downarrow}.$

A key property of tree decompositions is that every bag $X_t$ separates the subgraph $G_t$ from the rest of $G$.
In particular, any path that starts in $G_t \setminus X_t$ and leaves $G_t$ must pass through vertices of $X_t$.
Since every subpath of a shortest path is itself a shortest path, it follows that for any terminal vertex in $S_t^{\downarrow} \setminus X_t$,
every shortest path connecting it to terminals outside $G_t$ must intersect $X_t$.
Consequently, when processing node $t$, it suffices to remember which vertices in $U_t$ are mutually connected by shortest path within the partial solution restricted to $G_t$.
This information will be encoded by a connectivity table, together with the minimum number of edges required to realize it.

We formalize the above intuition as follows.
\begin{definition}[Connectivity Table]
For a node $t\in T$, let $\Sigma_t$ be the set of all symmetric functions $\sigma : U_t \times U_t \to \{0,1\},$
that is $\sigma(i,j)=\sigma(j,i)$ for all $i,j\in U_t$.
\end{definition}
Intuitively, a table $\sigma \in \Sigma_t$ has one row and one column for each vertex in $X_t$ and for each terminal vertex encountered so far.
The entry $\sigma(i,j)=1$ indicates that $i$ and $j$ are connected by a shortest path within the partial solution inside $G_t$,
while $\sigma(i,j)=0$ indicates that this is not the case.

Let $U$ and $U'$ be vertex sets with $U' \subseteq U$.
For a table $\sigma : U \times U \to \{0,1\}$, we denote by $\sigma\!\restriction_{U'}$ its restriction to $U'$, defined by
$ \sigma\!\restriction_{U'}(i,j) := \sigma(i,j) $
for all  $i,j \in U'$ .
Conversely, we say that a table $\sigma : U \times U \to \{0,1\}$ \emph{extends} a table $\sigma' : U' \times U' \to \{0,1\}$ if $\sigma\!\restriction_{U'} = \sigma'$.

\medskip
For a node $t$ and a table $\sigma\in\Sigma_t$, we define $w[t,\sigma]$ as the minimum number of edges in a subgraph $H_t \subseteq G_t$ such that for all $i,j\in U_t$:
\[
    \sigma(i,j) =
    \begin{cases}
        1 &\text{if and only if } d_{H_t}(i,j)=d_G(i,j),\\
        0 &\text{otherwise}.
    \end{cases}
\]
If no such subgraph exists, we set $w[t,\sigma]=+\infty$.

Any subgraph $H_t$ satisfying the above condition is said to be \emph{compatible} with $(t,\sigma)$ and to \emph{realize} the table $\sigma$.

By definition, every table $\sigma\in\Sigma_t$ is symmetric.
Throughout the remainder of this section, an entry $\sigma(i,j)$ implicitly represents both $(i,j)$ and $(j,i)$.
Accordingly, any operation on $\sigma(i,j)$ is assumed to apply symmetrically to $\sigma(j,i)$ without further mention.
\medskip
We now describe how to compute $w[t,\sigma]$ bottom-up via dynamic programming.

\begin{itemize}
    \setlength{\itemsep}{0.45em}
    \item \textbf{Leaf node:} 
        Let $t$ be a leaf, by definition $X_t=S_t^{\downarrow}=\emptyset$, so $U_t=\emptyset$.
        As $U_t$ is empty, there is only a unique table $\sigma=\emptyset$ defined at this node.
        The only subgraph compatible with $(t, \emptyset)$ is the empty graph.
        So we set $w[t,\emptyset]=0.$

    \item \textbf{Introduce vertex node:} 
        Let $t$ introduce vertex $v$ with child $t'$ so that $X_t=X_{t'}\cup\{v\}$.
        Recall that $U_t = X_t \cup S_t^{\downarrow}$ and that $U_{t'} \subseteq U_t$.
        For every $\sigma\in\Sigma_t$, we set
        \[
        w[t,\sigma] =
        \begin{cases}
        w[t',\sigma\!\restriction_{U_{t'}}],
        & \text{if } \sigma(v,u)=0 \text{ for all } u \in U_t \text{ with }u \neq v,\\[1mm]
        +\infty, & \text{otherwise}.
        \end{cases}
        \]
        By construction of a nice tree decomposition, no edge incident to $v$ appears in $G_t$.
        Hence, $v$ cannot be connected by a shortest path to any previously introduced vertex or terminal, and all corresponding entries involving $v$ must be zero.
        The remaining part of the table must agree with the state of the child.

    \item \textbf{Introduce edge node:} 
        Let $t$ introduce an edge $uv$ with child $t'$.
        We have $X_t = X_{t'}$, $S_t^{\downarrow} = S_{t'}^{\downarrow}$, and $d_G(u,v)=1$.
        Define
        \[
        U_{uv} := \bigl\{ \{p,q\} \subseteq U_t \mid d_G(p,q) = d_G(p,u) + 1 + d_G(v,q) \bigr\},
        \]
        the set of vertex pairs whose shortest path \emph{may} be realized using the new edge $uv$.  
        
        We say that a table $\sigma \in \Sigma_t$ is \emph{edge-invalid} at an introduce edge node $t$ if
        \[
            \sigma(u,v)=1 \quad \text{and there exists } \{p,q\} \in U_{uv} \text{ with } 
            \sigma(p,q)=0 \text{ and } \sigma(p,u)=\sigma(q,v)=1.
        \] 

        \medskip
        Let $\Sigma_t^{\mathrm{edge-inv}}$ denote the set of all edge-invalid tables at introduce edge node $t$.

        For each $\sigma \in \Sigma_t$, define
            \[
                \Sigma' := \bigl\{ \sigma' \in \Sigma_{t'} \mid 
                \sigma'(u,v)=0 \text{ and } \sigma'(i,j) = \sigma(i,j) \text{ for all pairs} \{i,j\} \notin U_{uv} \bigr\}.
            \]

        Then the update is
        \[
            w[t,\sigma] =
            \begin{cases}
                +\infty, & \text{if } \sigma \in \Sigma_t^{\mathrm{edge-inv}},\\[1mm]
                w[t',\sigma], & \text{if } \sigma(u,v)=0,\\[2mm]
                \displaystyle \min_{\substack{\sigma'\in \Sigma'}} w[t',\sigma'] + 1, & \text{otherwise}.
            \end{cases}
        \]
        The edge $uv$ may complete shortest paths between pairs in $U_{uv}$.
        If $\sigma$ is edge-invalid, it cannot be realized by any subgraph $H_t$.  
        If $\sigma(u,v)=0$, then the edge $uv$ is not used in $H_t$, so there is no difference between $H_{t'}$ and $H_t$.
        Otherwise if $\sigma(u,v)=1$, we consider all compatible child tables that agree outside $U_{uv}$ and add the edge to extend the solution.
        
    \item \textbf{Forget vertex node:}
        Let $t$ forget vertex $v$ with child $t'$.
        Then $S^{\downarrow}_{t}=S^{\downarrow}_{t'}$, and $X_{t}=X_{t'}\setminus \{v\}$.
        
        We distinguish the following cases:
        \begin{itemize}
            \setlength{\itemsep}{0.4em}
            \item If $v\in S^{\downarrow}_{t}$, and there exists a terminal $u \in S \setminus S^{\downarrow}_{t}$ such that $\{u, v\} \in \mathcal{P}$,
            then $v$ must remain connected to some vertex in $X_t$ via shortest path to allow future connection to $u$.
            Since $U_{t}=U_{t'}$, for every $\sigma \in \Sigma_t$ we set:
            \[
                w[t,\sigma]=
                    \begin{cases}
                    +\infty, & \text{if } \sigma(v,u)=0 \text{ for all } u\in X_t,\\[1mm]
                    w[t',\sigma] & \text{otherwise}.
                    \end{cases}
            \]
            \item If $v \in S_t^{\downarrow}$ but no such terminal $u$ exists, then $v$ has no further connectivity requirements, and for every $\sigma \in \Sigma_t$ we set
                \[
                w[t,\sigma] = w[t',\sigma].
                \]
            \item If $v \notin S_t^{\downarrow}$, then $U_t$ has one row and column fewer than $U_{t'}$ (corresponding to $v$).  
                For every $\sigma \in \Sigma_t$, define
                \[
                w[t,\sigma] = \min_{\substack{\sigma' \in \Sigma_{t'} \\ \sigma'|_{U_t} = \sigma}} w[t',\sigma'].
                \]
            For terminals that are “forgotten” but still need to connect to future terminals, we ensure they remain connected to the current bag.  
            Non-terminal vertices that are forgotten do not need to appear in future tables, so their rows and columns are removed. 
            The DP takes the minimum over all compatible child tables to preserve the minimum solution cost.
        \end{itemize}
        
    \medskip   
    \item \textbf{Join node:}
    Let $t$ be a join node with children $t_1,t_2$.
    Then $X_t = X_{t_1} = X_{t_2}, S_t^{\downarrow} = S_{t_1}^{\downarrow} \cup S_{t_2}^{\downarrow},$ and $\quad U_t = U_{t_1} \cup U_{t_2}$.
    Define $S^{\cap} = S_{t_1}^{\downarrow} \cap S_{t_2}^{\downarrow}$.
    
    For $\sigma\in\Sigma_t$, let $\sigma_i\in\Sigma_{t_i}$ be its restriction on $U_{t_i}$ for $i=1,2$, that is $\sigma_i = \sigma \restriction \! U_{t_i}$.   
    By definition, these restrictions agree with $\sigma$ on all indices in $X_t \cup S^{\cap}$:
    \[
        \sigma_1(i,j) = \sigma_2(i,j) = \sigma(i,j), \quad \forall i,j \in X_t \cup S^{\cap}.
    \]
    
    We say that $\sigma$ is \emph{join-invalid} at join node $t$ if there exist terminals $a\in S_{t_1}^{\downarrow} \setminus S_{t_2}^{\downarrow} $ and
    $b\in S_{t_2}^{\downarrow} \setminus S_{t_1}^{\downarrow} $ with $\{a, b\} \in \mathcal{P}$, and a vertex $v \in X_t$ such that:
    \[
    \sigma(a,b) = 0, \quad 
    \sigma_1(a,v) = \sigma_2(b,v) = 1, \quad
    d_G(a,b) = d_G(a,v) + d_G(v,b).
    \] 
    Let \(\Sigma_t^{\mathrm{join\text{-}inv}}\) denote the set of all such invalid tables at join node at $t$.
    
    To avoid double counting edges in the partial solution conforming with $(t, \sigma)$, define
    \[
        \delta(\sigma) := \bigl| \{ uv \in E(G[X_t]) \mid \sigma(u,v) = \sigma_1(u,v) = \sigma_2(u,v) = 1 \} \bigr|.
    \]
    
    Finally, for every $\sigma \in \Sigma_t$, we set: 
        \[
            w[t,\sigma] =
                    \begin{cases}
                            +\infty, & \text{if } \sigma \in \Sigma_{t}^{\mathrm{join\text{-}inv}},\\
                            w[t_1,\sigma_1] + w[t_2,\sigma_2] - \delta(\sigma),
                            & \text{otherwise}.
                    \end{cases}
    \]

    At a join node, we merge the partial solutions from the two subtrees.  
    A table \(\sigma\) is invalid if some terminal pair with one terminal in each subtree could form a shortest path via the bag $X_t$, but the table \(\sigma\) does not mark them as connected.  
    Otherwise, we sum the edge counts from the two children and subtract \(\delta(\sigma)\) to avoid double-counting edges entirely contained in the bag $X_t$.

    \medskip
    \item \textbf{Root node $r$: } At the root, $X_r=\emptyset$ and $S_r^{\downarrow} = S$, since all terminals have been visited.  
    We require that all terminal pairs in $\mathcal{P}$ are connected via a shortest path.
    Let
    \[
    \Sigma^* := \bigl\{ \sigma \in \Sigma_r \mid \sigma(u,v) = 1 \text{ for all } \{u,v\} \in \mathcal{P} \bigr\}.
    \]  
    The optimum size of a \pwshort\ for $(G, \mathcal{P})$ is then
    \[
    \min_{\sigma \in \Sigma^*} w[r,\sigma].
    \]
    At the root, all partial solutions have been combined, so we simply select the table that realizes shortest paths for all terminal pairs, giving the minimum number of edges in the final solution.
\end{itemize}

The dynamic programming described above computes the size of an optimal distance preserver.
By storing, for each table entry $(t,\sigma)$ with finite value, a witness choice attaining the minimum in the recurrence, one can reconstruct a corresponding subgraph by a standard backtracking procedure over the tree decomposition.
We therefore obtain not only the optimal value but also an optimal solution.

Having completed the description of the table entries and transition rules, it remains to show that the dynamic programming computes the value of an optimal solution to \pwshort.
The proof is based on a carefully maintained invariant that relates each table entry $w[t,\sigma]$ to subgraphs of $G_t$ realizing the connectivity encoded by $\sigma$.
We first restate this invariant,
\paragraph*{Dynamic Programming Invariant.}
    For every node $t$ of the tree decomposition and every table $\sigma \in \Sigma_t$ with $w[t,\sigma] < +\infty$, there exists a subgraph $H_t \subseteq G_t$ with $|E(H_t)| = w[t,\sigma]$, such that for all $i,j \in U_t$,
    \[
    \sigma(i,j) = 1 \quad \iff \quad d_{H_t}(i,j) = d_G(i,j).
    \]
Moreover, every subgraph $H_t \subseteq G_t$ satisfying this condition corresponds to exactly one table $\sigma$.  

We prove this invariant by induction over the nodes of the tree decomposition, following the bottom--up traversal used by the dynamic programming.
The induction shows that the invariant is preserved by each type of node and each transition rule.
It forms the basis for establishing both soundness (namely that every finite
table entry corresponds to a realizable partial solution in the graph) and completeness (every valid partial solution induces a table entry that survives the dynamic programming with the same cost).
Together, these properties ensure that the value returned at the root is
optimal.
\begin{lemma}[Soundness]
\label{lemma:soundness}
For every node $t \in T$ and every table $\sigma \in \Sigma_t$ with $w[t,\sigma] < +\infty$, there exists a subgraph
$H_t \subseteq G_t$ such that $|E(H_t)| = w[t,\sigma],$ and for all $i,j \in U_t$:
    \[
    \sigma(i,j) = 1 \quad \iff \quad d_{H_t}(i,j) = d_G(i,j).
    \]
\end{lemma}
\begin{proof}
We prove soundness by induction on the nodes of the nice tree decomposition.
\begin{description}
    \item[Leaf node.]
        Let $t$ be a leaf.
        Then $U_t=\emptyset$ and $G_t$ contains no vertices.
        Consequently, $\Sigma_t$ consists of only one table which is the empty table $\sigma=\emptyset$.
        By how we fill the DP table it holds that $w[t,\sigma]=0$.
        We define $H_t:=\emptyset$ to be the empty graph, which realizes $\sigma = \emptyset$ and trivially $|E(H_t)|=0=w[t,\emptyset]$.
    \item[Introduce vertex node.]
        Let $t$ introduce vertex $v$ with child $t'$; that is, $X_t = X_{t'} \cup \{v\}$.
        Consider any table $\sigma \in \Sigma_t$ with $w[t,\sigma] < +\infty$.  
        By the recurrence, its restriction $\sigma'=\sigma\!\restriction_{U_{t'}}$ satisfies $w[t',\sigma']<+\infty$, as we set $w[t,\sigma]=w[t',\sigma'].$
        By the induction hypothesis, there exists a subgraph $H_{t'} \subseteq G_{t'}$ realizing $\sigma'$ with $|H_{t'}|=w[t',\sigma']$.
        Since $v$ has no incident edges in $G_t$ at this stage, we define
        \[
            H_t := \bigl(V(H_{t'}) \cup \{v\},\, E(H_{t'})\bigr) \subseteq G_t.
        \]
        So $|E(H_t)| = |E(H_{t'})|$.
        Moreover, the DP only allows tables $\sigma$ with $ \sigma(u,v) = 0 $ for all $ u \in U_t,$
        implying that $v$ is isolated.
        Hence $H_t$ realizes $\sigma$, as the distances in $H_t$ agree with $\sigma$ for all $i,j \in U_t$, and
        \[
         |E(H_t)|= |E(H_{t'})|  = w[t',\sigma'] = w[t,\sigma].
        \]
        So the invariant holds.
        \item[Introduce edge node.]
            Let $t$ introduce an edge $uv$ with child $t'$.  
            Consider a table $\sigma \in \Sigma_t$ with $w[t,\sigma] < +\infty$.  
            By the DP recurrence, $\sigma \notin \Sigma_t^{\mathrm{edge-inv}}$.
            We distinguish two cases.

            \medskip
            \textbf{Case 1:} $\sigma(u,v) = 0$.  
            There exists a table $\sigma' \in \Sigma_{t'}$ such that $\sigma' = \sigma$ and $w[t',\sigma'] < +\infty$ as $w[t,\sigma]=w[t',\sigma']$ by the recurrence rule.
            By induction, there is a subgraph $H_{t'} \subseteq G_{t'}$ realizing $\sigma'$ with $|E(H_{t'})| = w[t',\sigma'].$  
            Since the edge $uv$ is not used, we define
            \[
            H_t := \bigl(V(H_{t'}), E(H_{t'})\bigr) \subseteq G_t.
            \]  
            Consequently $|E(H_t)| = |E(H_{t'})|$ and $H_t$ realizes $\sigma$, so the invariant holds as
            \[
                |E(H_t)| = |E(H_{t'})|= w[t',\sigma'] =w[t,\sigma].
            \]

            \medskip
            \textbf{Case 2:} $\sigma(u,v) = 1$.  
            By the DP update rule, there exists $\sigma' \in \Sigma_{t'}$ such that $\sigma$ is obtained from $\sigma'$ by possibly setting $\sigma(i,j) = 1$ for pairs $\{i,j\} \in U_{uv}$, and $\sigma(i,j) = \sigma'(i,j)$ elsewhere.  
            Let $\sigma'$ be the table minimizing $w[t',\sigma']$ among these choices.  
            By induction, there exists $H_{t'} \subseteq G_{t'}$ realizing $\sigma'$ with $|E(H_{t'})| = w[t',\sigma']$.
            
            We then define
            \[
            H_t := \bigl(V(H_{t'}), E(H_{t'}) \cup \{uv\} \bigr) \subseteq G_t,
            \]  
            which gives $|E(H_t)| = |E(H_{t'})| + 1$.
            Adding $uv$ realizes all shortest paths for pairs in $U_{uv}$, while $H_{t'}$ realizes all other distances.  
            Hence $H_t$ realizes $\sigma$ and
            \[
            |E(H_t)| = |E(H_{t'})| + 1  = w[t',\sigma'] + 1 = w[t,\sigma].
            \]

            Note the introduce-edge invalidity conditions ensure that no other shortest-path relations are created or violated; thus $H_t$ exactly realizes the connectivity table $\sigma$, preserving the invariant.

            \item[Forget vertex node.] 
                Let $t$ forget vertex $v$ with child $t'$, and let $\sigma \in \Sigma_t$ satisfy $w[t,\sigma] < +\infty$.  
                Recall that $X_t = X_{t'} \setminus \{v\}$ and $S_t^{\downarrow} = S_{t'}^{\downarrow}$.
                We distinguish two cases.
                
                \medskip
                \textbf{Case 1:} $v \in S_t^{\downarrow}$.  
                Then $U_t = U_{t'}$.  
                By the DP recurrence, any table $\sigma \in \Sigma_t$ with $w[t,\sigma] < +\infty$ is also a valid table at $t'$, that is, there exists $\sigma' \in \Sigma_{t'}$ with $\sigma' = \sigma$ and $w[t',\sigma'] < +\infty$.
                
                By the induction hypothesis, there exists a subgraph $H_{t'} \subseteq G_{t'}$ realizing $\sigma'$ with
                \[
                |E(H_{t'})| = w[t',\sigma'].
                \]
                We define $H_t := H_{t'}$.  
                Since distances between vertices in $U_t$ are unchanged, $H_t$ realizes $\sigma$, and
                \[
                 |E(H_t)| = |E(H_{t'})| = w[t',\sigma']  = w[t,\sigma].
                \]
                
                \medskip
                \textbf{Case 2:} $v \notin S_t^{\downarrow}$.  
                Then $U_t = U_{t'} \setminus \{v\}$.  
                By the DP recurrence, there exists an extension $\sigma' \in \Sigma_{t'}$ such that
                \[
                \sigma'\!\restriction_{U_t} = \sigma
                \quad \text{and} \quad
                w[t',\sigma'] < +\infty.
                \]
                By the induction hypothesis, there exists a subgraph $H_{t'} \subseteq G_{t'}$ realizing $\sigma'$ with
                \[
                |E(H_{t'})| = w[t',\sigma'].
                \]
                We again define $H_t := H_{t'}$.  
                Since $H_{t'}$ realizes $\sigma'$, it also realizes its restriction $\sigma$, and
                \[
                 |E(H_t)| = |E(H_{t'})| = w[t',\sigma']  = w[t,\sigma].
                \]
                Thus, in all cases, the invariant holds at forget nodes.

                \item[Join node.]  
                Let $t$ be a join node with children $t_1$ and $t_2$.
                Let $\sigma \in \Sigma_t$ satisfy $w[t,\sigma] < +\infty$, and let $\sigma_i := \sigma\!\restriction_{U_{t_i}} \in \Sigma_{t_i}$ for all $i=1,2.$

                By the DP recurrence, $\sigma \notin \Sigma_t^{\mathrm{join\text{-}inv}}$ and
                $w[t_i,\sigma_i] < +\infty$ for both $i=1,2$.
                
                By the induction hypothesis, there exist subgraphs
                $H_{t_i} \subseteq G_{t_i}$ realizing $\sigma_i$ such that
                \[
                |E(H_{t_i})| = w[t_i,\sigma_i]
                \quad \text{for } i=1,2.
                \]
                We define
                \[
                H_t := (V(H_{t_1}) \cup V(H_{t_2}),\; E(H_{t_1}) \cup E(H_{t_2})) .
                \]
                Since $G_{t_1}, G_{t_2} \subseteq G_t$, it follows that $H_t \subseteq G_t$.

                \medskip
                To see the realization of $\sigma$, first note that distances between pairs of vertices entirely contained in $U_{t_i}$ are preserved by $H_{t_i}$ and hence by $H_t$.
                Moreover, for pairs $a \in U_{t_1} \setminus U_{t_2}$ and $b \in U_{t_2} \setminus U_{t_1}$, the join invalidity condition guarantees that whenever $a$ and $b$ can be connected by a shortest path via some vertex $v \in X_t$,
                this is correctly encoded by $\sigma$.
                Thus $H_t$ realizes $\sigma$.

                \medskip
                For the edge count, note that an edge can belong to both $H_{t_1}$ and $H_{t_2}$ only if both its endpoints lie in $X_t$.
                Such edges are counted twice in $|E(H_{t_1})| + |E(H_{t_2})|$.
                The term $\delta(\sigma)$ subtracts exactly these edges, and therefore
                \begin{align}
                    |E(H_t)| &= |E(H_{t_1})| + |E(H_{t_2})| - |E(H_{t_1}) \cap E(H_{t_2})| \notag \\ 
                             &= w[t_1,\sigma_1] + w[t_2,\sigma_2] - \delta(\sigma) = w[t,\sigma].
                 \end{align}
                Hence the invariant holds at join nodes.
\end{description}
\end{proof}

\begin{lemma}[Completeness]
\label{lemma:completeness}
Let $H^*$ be an optimal solution to \pwshort\ for $(G, \mathcal{P}, k, \tau)$.
For every node $t \in T$, let $H^*_t := H^* \cap G_t .$
Then there exists a table $\sigma^* \in \Sigma_t$ such that $H^*_t$ realizes $\sigma^*$ and
\[
w[t,\sigma^*] \le |E(H^*_t)|.
\]
\end{lemma}

\begin{proof}
    We prove the lemma by induction on the nodes of the nice tree decomposition.
    Fix a node $t \in T$.
    Define a table $\sigma^* : U_t \times U_t \to \{0,1\}$ by setting
    \[
        \sigma^*(i,j) = \sigma^*(j,i) = 1
        \quad \text{if and only if} \quad
        d_{H^*_t}(i,j) = d_G(i,j),
    \]
    for all $i,j \in U_t$.
    By construction, $\sigma^*$ is symmetric and thus $\sigma^* \in \Sigma_t$.
    Moreover, $H^*_t$ realizes $\sigma^*$.
    
    We show that $\sigma^*$ satisfies all validity conditions of the DP recurrence and that
    \[
    w[t,\sigma^*] \le |E(H^*_t)|.
    \]
\begin{description}
    \item[Leaf node.] 
        If $t$ is a leaf, then $G_t = \emptyset$ and hence $H^*_t = \emptyset$.
        Thus $U_t = \emptyset$ and $\sigma^* = \emptyset$.
        By the DP recurrence,
        \[
        w[t,\sigma^*] = 0 \le |E(H^*_t)|.
        \]
    \item[Introduce vertex node.]
    Let $v$ be the vertex introduced at node $t$, and let $t'$ be its child.
    Since $v$ is isolated in $G_t$, it is also isolated in $H^*_t$.
    Hence, by construction, for all $u \in U_t$, we have $\sigma^*(u,v)=0$, and $\sigma^*$ satisfies the validity conditions at node $t$.
    So $w[t, \sigma^*]<+\infty$.
    Let $H^*_{t'} = H^* \cap G_{t'}$.
    By the induction hypothesis, there exists a table $\sigma' \in \Sigma_{t'}$ such that
    $H^*_{t'}$ realizes $\sigma'$ and $w[t',\sigma'] \le |E(H^*_{t'})|$.
    Since $v$ is isolated, $H^*_t$ differs from $H^*_{t'}$ only by the addition of an isolated vertex,
    and thus
    \[
    |E(H^*_t)| = |E(H^*_{t'})|.
    \]
    Moreover, all shortest-path relations between vertices in $U_{t'}$ are identical in $H^*_t$ and $H^*_{t'}$, implying that
    \[
    \sigma' = \sigma^*\!\restriction_{U_{t'}}.
    \]
    By the recurrence for introduce-vertex nodes, we have
    \[
    w[t,\sigma^*] = w[t',\sigma'] \le |E(H^*_{t'})| = |E(H^*_t)|.
    \]
    \item[Introduce edge node.]
        Let $t$ be a node with child $t'$ that introduces the edge $uv$, and let $H^*_{t'} = H^* \cap G_{t'}$.
        
        \medskip
        \textbf{Case 1:} $uv \notin H^*_t$.
        Then $H^*_t = H^*_{t'}$.
        By the induction hypothesis, there exists a table $\sigma' \in \Sigma_{t'}$ such that $H^*_{t'}$ realizes $\sigma'$ and $w[t',\sigma'] \le |E(H^*_{t'})|$.
        Since $uv \notin G_{t'}$, we have $\sigma'(u,v)=0$.
        Moreover, the shortest-path relations in $H^*_t$ and $H^*_{t'}$ agree, so $\sigma^* = \sigma'$.
        By the recurrence, $w[t,\sigma^*]=w[t',\sigma']$, and hence
        \[
        w[t,\sigma^*] \le |E(H^*_{t'})| = |E(H^*_t)|.
        \]
        \medskip
        \textbf{Case 2:} $uv \in H^*_t$.
        Then $H^*_t = H^*_{t'} \cup \{uv\}$ and $|E(H^*_t)| = |E(H^*_{t'})|+1$.
        Moreover, $\sigma^*(u,v)=1$.
        The addition of the edge $uv$ can only create new shortest paths for pairs
        $\{i,j\}$ whose shortest path in $G$ uses $uv$, and are already shortest-path connected to either $u$ or $v$.
        By definition, these are exactly the pairs contained in $U_{uv}$.
        Consequently, $\sigma^*$ and $\sigma'$ may differ only on entries $(i,j)$ with $\{i,j\} \in U_{uv}$.
        Thus $\sigma'$ is one of the tables considered in the recurrence when $\sigma(u,v)=1$, and $\sigma^* \notin \Sigma_t^{\mathrm{edge-inv}}$.
        By the recurrence,
        \[
        w[t,\sigma^*] \le w[t',\sigma'] + 1 \le |E(H^*_{t'})| + 1 = |E(H^*_t)|.
        \]
    \item[Forget vertex node.]
        Let $t$ be a node with child $t'$ that forgets vertex $v$, and let $H^*_{t'} = H^* \cap G_{t'}$.
        Since $T_t = T_{t'}$, it follows that $G_t = G_{t'}$, and hence $H^*_t = H^*_{t'}$ by definition.
        By the induction hypothesis, there exists a table $\sigma' \in \Sigma_{t'}$ such that $H^*_{t'}$ realizes $\sigma'$ and $w[t',\sigma'] \le |E(H^*_{t'})|$.
        
        \medskip
        \textbf{Case 1:} $v \in S_t^{\downarrow}$.
        Then $U_t = U_{t'}$, and by construction $\sigma^*=\sigma'$.
        Since $H^*$ is a valid \pwshort solution, for every terminal $u \in S \setminus S_t^{\downarrow}$ with $\{u,v\} \in \mathcal{P}$, vertex $v$ must remain shortest-path connected to the remainder of the graph.
        As $X_t$ separates $G_t$ from the rest of the graph, any shortest path from $v$ to such a terminal must pass through some vertex $x \in X_t$, implying that
        $\sigma^*(v,x)=1$ for some $x \in X_t$.
        Hence $\sigma^*$ satisfies the validity condition of the recurrence at $t$, and
        \[
            w[t,\sigma^*] = w[t',\sigma'] \le |E(H^*_{t'})| = |E(H^*_t)|.
        \]

        \medskip
        \textbf{Case 2:} $v \notin S_t^{\downarrow}$.
        Then $U_{t'}$ contains one additional row and column corresponding to $v$.
        By construction, $\sigma^* = \sigma'\!\restriction_{U_t}$.
        Since the recurrence takes the minimum over all such extensions, we obtain
        \[
        w[t,\sigma^*] \le w[t',\sigma'] \le |E(H^*_{t'})| = |E(H^*_t)|.
        \]
    \item[Join node.]
        Let $t$ be a join node with children $t_1$ and $t_2$.
        Let $H^*_{t_i} = H^* \cap G_{t_i}$ for $i=1,2$.
        By induction, for each $i \in \{1,2\}$ there exists a table $\sigma_i \in \Sigma_{t_i}$ such that $H^*_{t_i}$ realizes $\sigma_i$ and $w[t_i,\sigma_i] \le |E(H^*_{t_i})|$.

        Since $U_t = U_{t_1} \cup U_{t_2}$ and $H^*_t = H^*_{t_1} \cup H^*_{t_2}$, it follows by construction that $\sigma_i = \sigma^*\!\restriction_{U_{t_i}}$ for $i=1,2$.
        Moreover,
                \[
                |E(H^*_t)|
                =
                |E(H^*_{t_1})| + |E(H^*_{t_2})|
                - |E(H^*_{t_1}) \cap E(H^*_{t_2})|.
                \]
        Edges that appear in both $H^*_{t_1}$ and $H^*_{t_2}$ are exactly those edges $xy \in E(G)$ with $x,y \in X_t$ (equivalently, $d_G(x,y)=1$) that are realized in both tables $\sigma_1$ and $\sigma_2$.
        Hence,
        \[
        \delta(\sigma^*) = |E(H^*_{t_1}) \cap E(H^*_{t_2})|.
        \]
        Since $H^*$ is a \pwshort solution, for any terminal pair $\{a,b\} \in \mathcal{P}$ with
        $a \in S_{t_1}^{\downarrow} \setminus S_{t_2}^{\downarrow}$ and
        $b \in S_{t_2}^{\downarrow} \setminus S_{t_1}^{\downarrow}$,
        there exists a vertex $x \in X_t$ such that both $a$ and $b$ are shortest-path connected to $x$ in $H^*_t$.
        Therefore $\sigma^* \notin \Sigma_t^{\mathrm{join\text{-}inv}}$.

        By the recurrence,
        \begin{align}
            w[t,\sigma^*] &= w[t_1,\sigma_1] + w[t_2,\sigma_2] - \delta(\sigma^*) \notag \\ \notag
            &\leq |E(H^*_{t_1})| + |E(H^*_{t_2})| - \bigl|E(H^*_{t_1}) \cap E(H^*_{t_2})\bigr| \\ \notag
            &= |E(H^*_{t})|.
        \end{align}
\end{description}
\end{proof}

With soundness and completeness established, we now turn to the main theorem.
For clarity, we restate it here before presenting its proof.
\pairwiseFPTbyTWK*
\begin{proof}
Let $(G,\mathcal{P})$ be an instance of \pwshort\ with treewidth $\tau$.
Let $S:=\bigcup\mathcal{P}$ denote the set of terminals.
Let $\mathcal{T} = (T, \{X_t\}_{t\in V(T)})$ be a nice tree decomposition of $G$ of width at most $\tau$.

We apply the dynamic programming algorithm described above over $\mathcal{T}$.
For each node $t \in T$ and each table $\sigma \in \Sigma_t$, we compute $w[t,\sigma]$, the minimum number of edges of a subgraph $H_t \subseteq G_t$ realizing the connectivity table $\sigma$.

By \Cref{lemma:soundness}, for every table $\sigma$ with $w[t,\sigma]<+\infty$, there exists a subgraph $H_t$ realizing $\sigma$ with exactly $w[t,\sigma]$ edges.
By \Cref{lemma:completeness}, for an optimal pairwise distance preserver $H^*$, at every node $t$ there exists a table $\sigma^* \in \Sigma_t$ such that $H^*_t = H^* \cap G_t$ realizes $\sigma^*$ and
\[
w[t,\sigma^*] \le |E(H^*_t)|.
\]

At the root node $r$, we have $X_r = \emptyset$ and $S^{\downarrow}_r = S$, so that every terminal is represented in $U_r$. By definition, the set \(\Sigma^* \subseteq \Sigma_r\) contains all tables \(\sigma\) that connect every terminal pair in \(\mathcal{P}\) via shortest path. 
The dynamic programming recurrence guarantees that
\[
\min_{\sigma \in \Sigma^*} w[r,\sigma]
\]
equals the number of edges in a minimum pairwise distance preserver, by soundness and completeness.

Hence, the DP computes an optimal solution size.
By a standard backtracking over the tree decomposition, using the witness choices stored in each table, one can reconstruct an optimal subgraph.

Finally, the running time of the algorithm is
\[
|\Sigma_t| \cdot |V(T)|^{O(1)}
\le 2^{O((\tau+|S|)^2)} \cdot |V(G)|^{O(1)},
\]
since each table is a symmetric binary matrix over at most \(|U_t| \le (\tau+1) + |S|\) vertices and $|T| = O(|V(G)|)$.
The additional operations performed at each node — including checking shortest-path connectivity for sets like $U_{uv}$, updating table entries, and checking validity — are dominated by the running time of the introduce-edge node, which is $O\bigl((\tau+|S|)^2 (n+m)\bigr)$ for $n=|V(G)|$ and $m=|E(G)|$.
Hence the overall running time is $2^{O((\tau+|S|)^2)}\cdot |V(G)|^{O(1)}$, as claimed.
\end{proof}

\begin{remark}
    It is possible to further optimize the dynamic programming state by storing directly the connectivity information for the relevant terminal pairs.  
    In particular, for constant treewidth, this reduces the running time to $2^{O(|\mathcal{P}|)}$.
\end{remark}

By inclusion, we immediately get the same result for \swshort.
\begin{restatable}{corollary}{subsetFPTbyTWK}
\label{cor:sw-tw-k-fpt}
        \swshort admits an algorithm with running time $2^{O\bigl((\tw+|S|)^2\bigr)} \cdot |V(G)|^{O(1)}$, where $\tw$ is the treewidth of $G$.
\end{restatable}
\section{Vertex cover}\label{sec:vc}
In this section, we focus on the results where the parameter is the vertex cover of the graph. We start by showing the NP-hardness of \pwshort on graphs of vertex cover $3$ (Theorem~\ref{thm:pdp-vc-paraNphardness}), and then move on to the FPT algorithm for \swshort parameterized by vertex cover (Theorem~\ref{thm:sw-vc-fpt-thm}).

\subsection{PDP is NP-hard for vertex cover 3}
Here, we prove Theorem~\ref{thm:pdp-vc-paraNphardness}, restated next for convenience.
\VcParaNpHardness*

Our proof proceeds via a sequence of polynomial-time reductions.
We first reduce from \mwcfull\ (\mwcshort) to an auxiliary problem, which we denote by \alcfull\ (\alcshort), and subsequently reduce from \alcshort\ to \pwshort.

We begin by formally defining the two problems used in the reduction.
\defbox{\alcfull}{\alcshort}
{ A graph $G=(V,E)$, a set of allowed colors $\ell_v \subseteq \mathbb{N}$ with $|\ell_v| \leq 3$ for each vertex $v \in V$, and an integer $k$.}
{Is there an assignment $\mathcal{X}= \{X_v\subseteq \ell_{v} \mid v\in V(G),\ X_v\neq \emptyset\}$ such that for every edge $uv \in E(G)$, 
$X_u \cap X_v \neq \emptyset$ and the total number of assigned colors is bounded by $k$, i.e., $\sum_{v \in V} |X_v| \le k$\,?}

We call an assignment of colors $\mathcal{X}=(X_v)_{v \in V(G)}$ that satisfies these conditions a \emph{valid color assignment} for an instance $\bigl(G=(V, E), (\ell_v)_{v \in V}\bigr)$ of \alcshort.
If no valid color assignment exists, the instance is infeasible.
Throughout, we use “list” and “set” interchangeably to refer to the allowed colors of a vertex.

Finally, for a graph $G=(V,E)$ and a subset of edges $F \subseteq E$, we denote by $G-F$ the graph $(V, E \setminus F)$.
\defbox{\mwcfull}{\mwcshort}
{A connected graph $G=(V,E)$, three distinct terminals $s_1, s_2, s_3 \in V$, and an integer $k$.}
{Does there exist a set of edges $F \subseteq E$ with $|F| \le k$ such that in the graph $G - F$, each terminal is disconnected from the other two?}

Our starting point is the classical NP-hardness of \mwcshort due to Dahlhaus et al.

\begin{restatable}[\cite{3waycut}]{theorem}{ThreeWayCutHardness}
    Minimum \mwcshort\ is \nph\ even on unweighted graphs.
\end{restatable}

To prove~\autoref{thm:pdp-vc-paraNphardness}, which is the main result of this section, we combine two polynomial-time reductions.
First, we give a polynomial-time reduction from \mwcshort\ to \alcshort. 
Then, we reduce from \alcshort\ to a restricted version of \pwshort, in which the graph is bipartite, has vertex cover of at most $3$, and the terminal pairs satisfy a specific structure. 
Together, these two reductions establish that \pwshort\ is at least as hard as \mwcshort, which is \nph. 

For clarity and readability, each reduction is presented in a separate unnumbered subsection as a lemma (Lemmas~\ref{lem:mwc-to-alc} and~\ref{lem:alc-to-bip}), 
including the corresponding construction. 
Finally, the main theorem is restated and proved by combining the results of these lemmas.

\paragraph*{Reduction from \mwcshort\ to \alcshort}
We begin by reducing \mwcshort\ to \alcshort.
Let $I_{\mwcshort}=\bigl(G, \{s_1, s_2, s_3\}, k\bigr)$ be an instance of \mwcshort.
We construct a corresponding instance of \alcshort, $I_{\alcshort}=\bigl(G', (\ell_{v'})_{v' \in V(G')}, k'\bigr)$ where $k'=(n+1)m+k$ and $|\ell_{v'}| \leq 3$ for all $v' \in V(G')$, as follows.

Let $n=|V(G)|$ and $m=|E(G)|$ denote the number of vertices and edges in $G$.
For each vertex $v \in V(G)$, we add to $G'$ a clique $K_v =\{x_{v,1}, x_{v,2}, \dots, x_{v,m}\}$ of size $m$.
Next, for each edge $e = uv \in E(G)$, we introduce a new vertex $w'_e \in V(G')$ and connect it to every vertex in both $K_u$ and $K_v$.
See \autoref{fig:mwc-to-alc:edge-gadget} for an illustration of this construction.
\begin{figure}
    \centering
    \begin{subfigure}[b]{.64\textwidth}
          \centering
          \includegraphics[width=.95\textwidth, center]{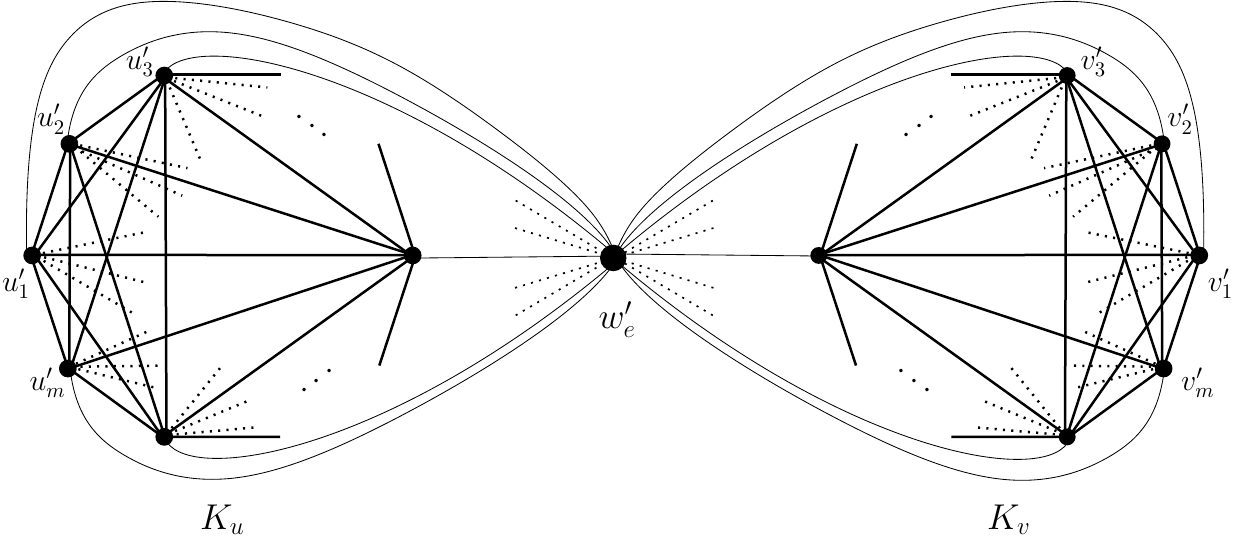}
          \caption{The construction in $G'$ corresponding to an edge $e = (u,v)$ of $G$.}
          \label{fig:mwc-to-alc:edge-gadget-cliques}
        \end{subfigure}
        \hfill
        \begin{subfigure}[b]{.33\textwidth}
          \centering
          \includegraphics[width=1\textwidth]{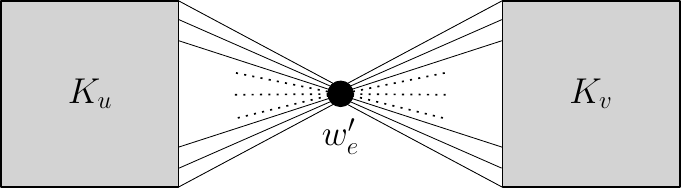}
          \caption{A schematic representation of the same gadget, with cliques depicted as squares.}
          \label{fig:mwc-to-alc:edge-gadget-square}
    \end{subfigure}
    \caption{ In (a), the vertices $u$ and $v$ are replaced by cliques $K_u$ and $K_v$, each consisting of $m$ vertices, and a vertex $w'_e$ is introduced adjacent to all vertices in both cliques. In (b), the same construction is shown schematically, with cliques represented by squares; this representation is used in subsequent figures for clarity.}
    \label{fig:mwc-to-alc:edge-gadget}
\end{figure}

We now define the lists of allowed colors for the vertices of $G'$.
Initially, set $\ell_{v'} = \emptyset$ for all $v' \in V(G')$.
For each terminal $s_i$ with $i \in [3]$, assign $\ell_{v'} = \{i\}$ to every vertex $v' \in K_{s_i}$.  
For all the remaining vertices $v' \in V(G') \setminus \bigcup_{i=1}^3 V(K_{s_i})$, set $\ell_{v'}=\{1,2,3\}$.

Thus, in $G'$, each clique $K_{s_i}$ is forced to use color $i$, while all other vertices may use any of the three colors.
An example illustrating this construction is shown in \autoref{fig:mwc-to-alc}.

\begin{restatable}{lemma}{MwcToAlc}
\label{lem:mwc-to-alc}
Let $I_{\mwcshort}=\bigl(G=(V, E), \{s_1, s_2, s_3\}, k\bigr)$ be an instance of \mwcshort,\ where $G$ has $n$ vertices and $m$ edges.
Let $I_{\alcshort}=\bigl(G'=(V', E'),(\ell_{v'})_{v' \in V(G')}, k'\bigr)$ be the instance of \alcshort\ constructed as described above.
Then $I_{\mwcshort}$ admits a \mwcshort\ of at most $k$ edges if and only if $I_{\alcshort}$ admits a valid color assignment of total size at most $k'=(n+1)m+k$.
\end{restatable}

\begin{proof}
We first prove the forward direction.
Assume that $I_{\mwcshort}=\bigl(G=(V, E), \{s_1, s_2, s_3\}, k\bigr)$ admits a \mwcshort\ $F \subseteq E(G)$ of size at most $k$.
From $F$, we construct  a valid \alcshort\ assignment $\mathcal{X}=(X_{v'})_{v' \in G'}$ for the instance $I_{\alcshort}$ with total size at most $k'=(n+1)m+k$.

Let $\mathcal{C} = \{C_1, C_2, \dots, C_p\}$ be the connected components of $G - F$, and let $C_{s_i}$ denote the component containing terminal $s_i$ for $i \in [3]$.

For every vertex $v \in V(G)$ and for every vertex $u \in K_v$ in the corresponding clique of $G'$, define the assigned color set $X_{u}$ as follows:
\[
    X_{u} =
        \begin{cases}
             \{i\}, & v \in C_{s_i} \text{ for some } i \in [3], \\
            \{1\}, & v \in C \text{ with } C \in \mathcal{C} \setminus \{C_{s_1}, C_{s_2}, C_{s_3}\}
        \end{cases}
\]
For each clique $K_v \subseteq V(G')$, we define
\[
    \gamma(K_v) := \bigcup_{u \in K_v} X_u,
\]
which we refer to as the \emph{color of the clique $K_v$}.  
By construction, all vertices in a clique $K_v$ receive the same color, so $\gamma(K_v)$ is always a singleton set.  

Using this notation, for each edge $e = uv \in E(G)$, we assign to the corresponding vertex $w'_e \in V(G')$ the color set
\[
    X_{w'_e} := \gamma(K_u) \cup \gamma(K_v).
\]
This completes the construction of the color assignment $\mathcal{X}$.\\
We now show that $\mathcal{X}$ is a valid color assignment of colors for $I_{\alcshort}$.

First, we verify that $X_{v'} \subseteq \ell_{v'}$ for all $v' \in V(G')$.
If $v' \in K_{s_i}$ for some $i \in [3]$, then $X_{v'} = \{i\} = \ell_{v'}$ by construction.
Otherwise, $v' \notin \bigcup_i K_{s_i}$, so $\ell_{v'} = \{1,2,3\}$ and $X_{v'} \subseteq \{1,2,3\}$.

Next, let $W'_E := \{w'_e \in V(G') \mid e \in E(G)\}$ be the set of vertices corresponding to edges of $G$.
Consider the edges of $G'$. 
For edges entirely within a clique $K_u$, all vertices are assigned the same color $\gamma(K_u)$, so for any edge $e' = u'_i u'_j \in K_u$, we have $X_{u'_i} \cap X_{u'_j} = \gamma(K_u) \neq \emptyset.$

The remaining edges of $G'$ have one endpoint in $W'_E$ and the other in a clique.
Consider a vertex $w'_e \in W'_E$ corresponding to an edge $e = uv \in E(G)$.
By construction, $w'_e$ is adjacent exactly to the vertices of the two cliques $K_u$ and $K_v$, that is, $N_{G'}(w'_e) = V(K_u) \cup V(K_v)$.
According to the color assignment described earlier, $X_{w'_e} = \gamma(K_u) \cup \gamma(K_v)$.
Moreover, all vertices of $K_u$ (resp., $K_v$) are assigned the color $\gamma(K_u)$ (resp., $\gamma(K_v)$).
Therefore, for every $z \in K_u \cup K_v$ it follows that:
\[
    X_{w'_e} \cap X_{z} \neq \emptyset
\]
Thus, all adjacency constraints are satisfied, and $\mathcal{X}$ is a valid color assignment for $I_{\alcshort}$.

To bound the total size of the color assignment, $\sum_{v' \in V(G')} |X_{v'}|$, we proceed as follows:
\[
\sum_{v' \in V(G')} |X_{v'}| = \sum_{K_v \in V(G')} \sum_{v' \in K_v} |X_{v'}| + \sum_{w'_e \in V'_{E}} |X_{w'_e}| 
\]
By construction, for every clique $K_v \subseteq G'$, all of its vertices are assigned the same color $\gamma(K_v)$.  
Hence, for each $v' \in K_v$, we have $|X_{v'}| = 1$.  
Since each clique has $m$ vertices and there are $n$ cliques, the total contribution of all clique vertices is $\sum_{K_v \in V(G')} \sum_{v' \in K_v} |X_{v'}| = n\cdot m.$

For each edge $e = (u,v) \in E(G)$, the corresponding vertex $w'_e \in V(G')$ is assigned the union of the colors of its endpoints’ cliques:
$X_{w'_e} = \gamma(K_u) \cup \gamma(K_v), \quad \text{so } |X_{w'_e}| \le 2.$
If $u$ and $v$ belong to the same connected component in $G - F$, then $\gamma(K_u) = \gamma(K_v)$ and $|X_{w'_e}| = 1$.
Otherwise, $\gamma(K_u) \neq \gamma(K_v)$ and $|X_{w'_e}| = 2$.

Recall that $|F| \le k$.
By definition of a \mwcshort, since $G$ is connected, all edges connecting distinct connected components in $G - F$ must belong to the cut $F$.
So there are at most $k$ edges connecting different connected components in $G - F$.  
Each of these edges contributes at most $2$ colors, while the remaining $m - k$ edges lie entirely within a component, each contributing exactly $1$ color.
Hence, the total contribution from the edge vertices satisfies
\[\sum_{w'_e \in W'_E} |X_{w'_e}| \le 2|F| + (m - |F|) = m + |F| \le m + k.\]

Combining the contributions from the cliques and the edge vertices, we obtain
\[
\sum_{v' \in V(G')} |X_{v'}| \le n \cdot m + (m + k) = (n+1)\cdot m + k = k',
\]
which gives the desired upper bound on the total size of the color assignment.

\begin{figure}
    \centering
    \begin{subfigure}[b]{.35\textwidth}
          \centering
          \includegraphics[width=1\textwidth, center]{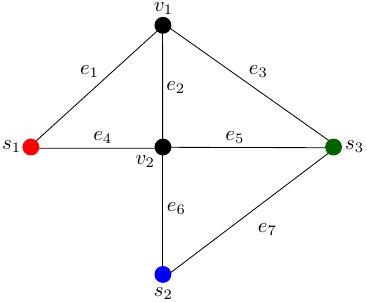}
          \caption{An instance of \mwcshort\ $I_{\mwcshort}=\bigl(G=(V, E), \{s_1, s_2, s_3\}, k\bigr)$.}
          \label{fig:mwc}
        \end{subfigure}
        \hfill
        \begin{subfigure}[b]{.6\textwidth}
          \centering
          \includegraphics[width=.9\textwidth, height=.34\textheight]{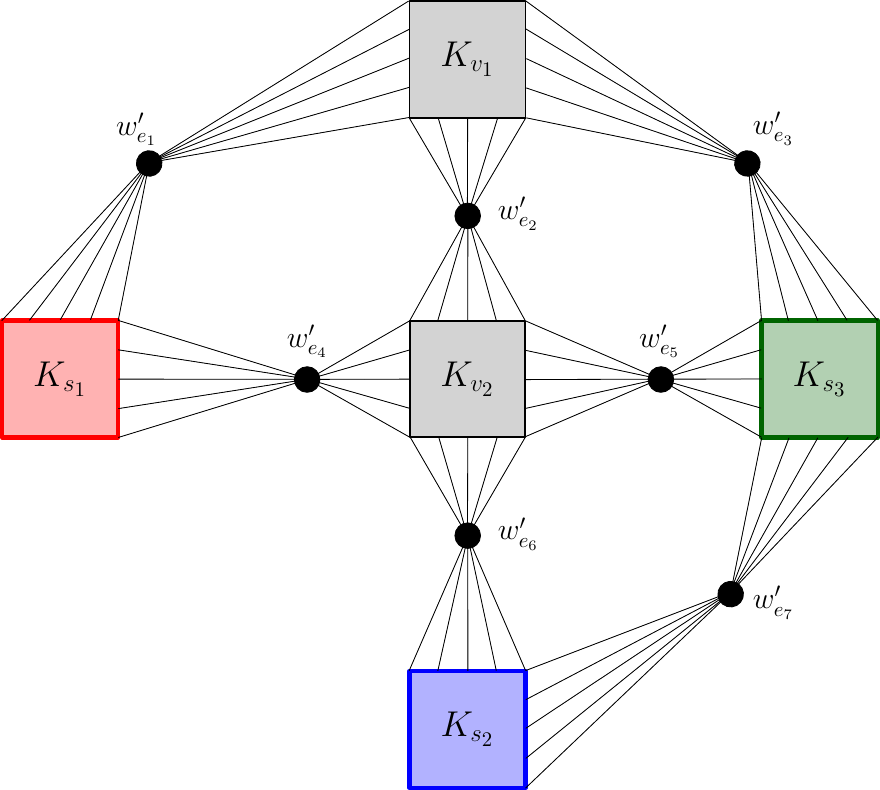}
          \caption{The corresponding instance of \alcshort\ $I_{\alcshort}=(G'=(V', E'),(\ell_{v'})_{v' \in V(G')}, k')$ constructed from the instance in (a)}
          \label{fig:alc1}
    \end{subfigure}
    \caption{Each vertex of $G$ is replaced by a clique of size $m=7$, illustrated as squares. For every edge $e=(u,v)$ in $G$, we introduce a vertex $w'_e$ that is adjacent to all vertices in the cliques $K_u$ and $K_v$. Colors red, blue, and green represent the three permissible colors $1, 2$, and $3$ respectively. The lists of allowed colors for vertices in $K_{s_1}$, $K_{s_2}$ and $K_{s_3}$ are restricted only to red, blue, and green respectively, while the remaining vertices can have all the three colors.}
    \label{fig:mwc-to-alc}
\end{figure}

To see the other direction, assume that $I_{\alcshort}$ admits a valid color assignment $\mathcal{X}=(X_{v'})_{v' \in G'}$ of total size $k'=(n+1)m+k$.
From $\mathcal{X}$ we construct a solution $F \subseteq E(G)$ for the instance $I_{\mwcshort} = \bigl(G=(V, E),\{s_1,s_2,s_3\}, k\bigr)$ of size at most $k$.
For every $w'_e \in V(G')$ with $|X_{w'_e}|\geq 2$, include its corresponding edge $e \in E(G)$ into the cut, that is $F := F \cup \{e\}$.

We show that $F$ is a valid \mwcshort\ by contradiction.  
Assume, for the sake of contradiction, that there exist distinct $a,b \in [3]$ such that
$s_a$ and $s_b$ are connected in $G - F$.  
Let $Q = (s_a, e_1, v_1, \dots, v_{q-1}, e_q, s_b)$ be a path connecting $s_a$ and $s_b$ in $G - F$.

By construction of $G'$, the existence of $Q$ implies that:  
all vertices of $K_{s_a}$ and $K_{v_1}$ are adjacent to $w'_{e_1}$,  
all vertices of $K_{v_1}$ and $K_{v_2}$ are adjacent to $w'_{e_2}$,  
and so on, until all vertices of $K_{v_{q-1}}$ and $K_{s_b}$ are adjacent to $w'_{e_q}$.

By the color assignment constraints, every vertex $v' \in K_{s_a}$ is assigned the unique color $a$, i.e., $X_{v'} = \{a\}$, and every vertex $u' \in K_{s_b}$ is assigned the unique color $b$, i.e., $X_{u'} = \{b\}$.  
Moreover, since none of the edges $e_1, \dots, e_q$ belong to $F$, we have $|X_{w'_{e_i}}| = 1$ for all $i \in [q]$.

Hence, for the color assignment $\mathcal{X}$ to be valid, there must exist some clique $K_{v_t}$ along the path where the vertices are assigned a set containing both colors $a$ and $b$.  
This implies that $K_{v_t}$ contributes $2m$ to the cardinality of the assignment $\mathcal{X}$, and the remaining vertices of $G'$ each contribute at least $1$ color to the cardinality of $\mathcal{X}$.
Thus, since there are $n-1$ cliques each with $m$ vertices and $|W'_{E}|=m$, the total size of the color assignment satisfies
\[
\sum_{v' \in G'} |X_{v'}| \geq 2m+(n-1) \cdot m + m = (n+2) \cdot m > (n+1)\cdot m+k
\]
which contradicts the assumption that $\sum_{v' \in V(G')} |X_{v'}| \le k' = (n+1)\cdot m + k$.
Therefore, such a path cannot exist, and $F$ is indeed a \mwcshort\ for $I_{\mwcshort} = \bigl(G=(V, E),\{s_1,s_2,s_3\}, k\bigr)$.

It remains to bound the size of $F$.  
This reduces to bounding the number of vertices $w'_e \in V(G')$ with $|X_{w'_e}| \ge 2$.  
The graph $G'$ contains $(n+1)\cdot m$ vertices, and each vertex is assigned at least one color.  
Since the total size of $\mathcal{X}$ is at most $(n+1)\cdot m + k$, there can be at most $k$ vertices with color sets of size greater than $1$.  
By construction, each such vertex corresponds to an edge in $F$, so we conclude that $|F| \le k$.
\end{proof}

\paragraph*{Reduction from \alcshort\ to \pwshort}

As the second step of our reduction, we introduce a restricted variant of \pwshort\ and reduce from \alcshort\ to this variant.

\defboxsmalltitle{\bipwfull}{\bipwshort}
{An instance of \pwshort consisting of a bipartite graph $G=(V,E)$ with bipartition $V = L \cup R$, where $|L| = 3$, together with a set of $t$ terminal pairs and an integer $k$, where:
\[
\p = \Bigl\{ p_i = \{u_i,v_i\} \mid i \in [t],\ u_i, v_i \in R,\ d_G(u_i,v_i) = 2 \Bigr\}.
\]
}
{Is there a \pwshort\ for $\p$ with at most $k$ edges?}

Observe that in \bipwshort, the input graph admits a vertex cover of size at most~$3$, namely the set $L$, since $G$ is bipartite. 

The reduction from \alcshort\ to \bipwshort\ constructs, for each instance of \alcshort, a bipartite graph with vertex cover 3 and terminal pairs at distance 2, such that a minimum \pwshort\ corresponds exactly to a minimum valid color assignment in the original \alcshort\ instance.

Let $I_{\alcshort}=\bigl(G, (\ell_{v})_{v \in V(G)}, k\bigr)$ be an instance of \alcshort\ such that $\ell_v \subseteq \{1, 2, 3\}$ and $\ell_v \neq \emptyset$ for every vertex $v \in V(G)$.
We construct an instance of \bipwshort, $I_{\textsc{VC-3-Bip}}=(G', \p', k)$ as follows.

The graph $G'$ is bipartite with partitions $L'$ and $R'$.
We add three vertices $a_1, a_2, a_3$ to $L'$, which correspond to the three colors.
For each vertex $v \in V(G)$, we add a corresponding vertex $v'$ to $R'$ and connect $v'$ to $a_i \in L'$ if and only if $i \in \ell_v$.
Intuitively, the neighborhood of $v'$ in $L'$ encodes the set of allowed colors of $v$.
Thus, each vertex in $R'$ is adjacent exactly to those vertices in $L'$ that represent the allowed colors of its corresponding vertex in $G$.
Finally, for every edge $uv \in E(G)$, we add the terminal pair $(u',v')$ to $\p'$.
Hence, the number of terminal pairs in $\p'$ equals $|E(G)|$.
Refer to \autoref{fig:alc-to-pair} for an illustration of the construction.
\begin{figure}[h]
    \centering
    \begin{subfigure}[b]{.48\textwidth}
          \centering
          \includegraphics[width=.9\textwidth, center]{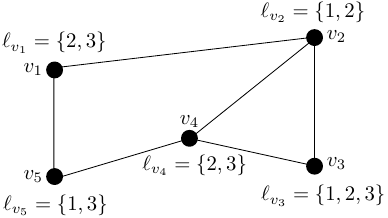}
          \caption{An instance of \alcshort\ where each vertex is assigned a list of at most three allowed colors.}
          \label{fig:alc2}
        \end{subfigure}
        \hfill
        \begin{subfigure}[b]{.48\textwidth}
          \centering
          \includegraphics[width=.5\textwidth]{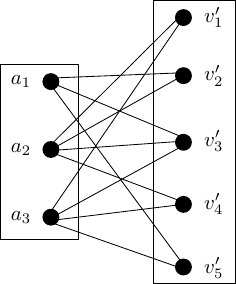}
          \caption{The corresponding instance of \bipwshort\ constructed from the \alcshort\ instance in (a)}
          \label{fig:bipw}
    \end{subfigure}
    \caption{Illustration of the reduction from an instance of \alcshort\ $\bigl(G, (\ell_v)_{v \in V(G)}, k\bigr)$ to an instance of \bipwshort\ $\bigl(G'=(L', R', E'), \p', k\bigr)$. Each vertex in $G$ is mapped to a vertex in the right partition of $G'$, while the three vertices in the left partition, $\{a_1, a_2, a_3\}$, represent the available colors. Terminal pairs $\p'= \bigl\{(v'_1, v'_2), (v'_1, v'_5), (v'_2, v'_3), (v'_2, v'_4), (v'_3, v'_4), (v'_4, v'_5) \bigr\}$ correspond to edges of $G$.}
    \label{fig:alc-to-pair}
\end{figure}
\begin{restatable}{lemma}{AlcToBip}
\label{lem:alc-to-bip}
Let $I_{\alcshort}=\bigl(G,(\ell_v)_{v\in V(G)}, k\bigr)$ be an instance of \alcshort\ with $\ell_v\subseteq\{1,2,3\}$ and $\ell_v\neq\emptyset$ for all $v\in V(G)$.
Let $I_{\textsc{VC-3-Bip}}=(G'=(L',R',E'),\p', k)$ be the instance of \bipwshort\ constructed as above.
Then $I_{\alcshort}$ has a valid color assignment of total size at most $k$ if and only if $I_{\textsc{VC-3-Bip}}$ has a \pwshort with at most $k$ edges.
\end{restatable}

\begin{proof}
To see the forward direction, assume that $I_{\alcshort}$ admits a valid color assignment $\mathcal{X}=(X_v)_{v\in V(G)}$ with $\sum_{v \in V} |X_v| = k$.

We construct a \pwshort $H \subseteq G'$ for the terminal pairs in $\p'$ with $k$ edges as follows.
For each vertex $v \in V(G)$ and each color $i \in X_v$, we include the edge $v'a_i$ in $H$.
Note that since $i \in \ell_v$, by construction, edge $v'a_i$ exists in $E(G')$. 

We now show that $H$ is a \pwshort.
Consider any terminal pair $(u', v') \in \p'$.
By definition of $\p'$, every terminal pair in $G'$ corresponds to an edge in $G$, so the vertices $u$ and $v$ are adjacent in $G$.
Since $\mathcal{X}$ is a valid color assignment, it holds that $(X_u \cap X_v )\neq \emptyset$.
Let $i \in (X_u \cap X_v)$.
By the construction of $G'$, both $u'$ and $v'$ are adjacent to $a_i$, so $d_{G'}(u', v') = 2$.
As a result, it is enough to prove that, there is a path of length $2$ between $u'$ and $v'$ in $H$.
Since $i \in (X_u \cap X_v)$, the edges $u' a_i$ and $v' a_i$ are included in $H$.
Therefore, $H$ contains a path of length $2$ between $u'$ and $v'$, and thus preserves at least one shortest path between pair $(u', v')$.

It remains to bound the size of $H$.
For each vertex $v' \in V(H)$ corresponding to $v \in V(G)$, we include into $H$ exactly $|X_v|$ edges incident to it.
Hence, $|E(H)| = \sum_{v \in V(G)} |X_v| = k.$
Since $G'$ (and consequently $H$) is bipartite and all terminal pairs lie in the right partition, every edge of $H$ is counted exactly once in this sum.
Thus, $H$ is a \pwshort for $\p'$ with $k$ edges.

To prove the reverse direction, assume that $I_{\textsc{VC-3-Bip}}$ has a \pwshort $H$ with $k$ edges.
We construct a valid color assignment $\mathcal{X}=(X_v)_{v \in V(G)}$ for $I_{\alcshort}$ such that $\sum_{v \in V(G)} |X_v| = k$.

We first handle a degenerate case.
If there exists a terminal pair $(u',v') \in \p'$ such that $ N_{G'}(u') \cap N_{G'}(v') = \emptyset$, then $I_{\alcshort}$ is a trivial \no-instance.
Indeed, the pair $(u',v')$ corresponds to an edge $uv \in E(G)$, and the condition $N_{G'}(u') \cap N_{G'}(v') = \emptyset$ implies $(\ell_u \cap \ell_v) = \emptyset$.
Hence, no valid color assignment of $G$ exists.

We may therefore assume that for every terminal pair $(u',v') \in \p'$ it holds that
$d_H(u',v') = d_{G'}(u',v') =2$.
Since $G'$ is bipartite, every such path has the form $u' - a_i - v'$ for some $i \in [3]$.

For each terminal pair $(u',v') \in \p'$, add every color $i \in \bigl(N_{H}(u') \cap N_{H}(v')\bigr)$ to both $X_u$ and $X_v$ only if it is not already present. 
That is, set $X_w := X_w \cup \{i\}$ for $w \in \{u,v\},$ if $i \notin X_w.$ 

We now show that $\mathcal{X}$ is a valid color assignment.
It suffices to prove that for every edge $(u,v) \in E(G)$, we have $X_u \subseteq \ell_u$, $X_v \subseteq \ell_v$, and $(X_u \cap X_v) \neq \emptyset$.

Consider any edge $uv\in E(G)$.
The corresponding terminal pair $(u',v')$ satisfies $d_H(u',v') = 2$, and hence there exists at least one $i \in [3]$ such that both $u'a_i$ and $v'a_i$ belong to $E(H)$.
By the construction of $G'$ and since $E(H) \subseteq E(G')$, this implies $i \in (\ell_u \cap \ell_v)$.
Since $i$ is added to both $X_u$ and $X_v$, we obtain $i \in (X_u \cap X_v) \neq \emptyset$.
Moreover, colors are added to $X_u$ and $X_v$ only when they belong to the respective allowed lists.
Therefore, $X_u \subseteq \ell_u$ and $X_v \subseteq \ell_v$.
Thus, $\mathcal{X}$ is a valid color assignment.

It remains to bound the total size of the color assignment.
For each color $i \in [3]$, let $\alpha(i) = \bigl|\{v \in V(G) \mid i \in X_v\}\bigr|$ denote the number of vertices in $G$ whose assigned color set contains color~$i$.
By double counting, we have $\sum_{v \in V(G)} |X_v| = \sum_{i \in [3]} \alpha(i)$.

Note that by the construction of $\mathcal{X}$, color $i$ is added to the set $X_v$ of vertex $v$ only if there exists a terminal pair $(u', v')$ with $u' \in R'$ such that $a_i \in \bigl(N_H(v') \cap N_H(u') \bigr)$.
Moreover, each color is added to any set $X_v$ at most once.
Observe that $a_i \in N_H(v')$ holds if and only if $v'a_i \in E(H)$.

Consequently, each occurrence of color $i$ in some set $X_v$ corresponds uniquely to one edge of $H$ incident to $a_i$, which is $(a_i, v')$.
Conversely, every edge incident to $a_i$ gives rise to exactly one such occurrence.
Therefore, the number of vertices whose assigned sets contain color $i$ is precisely the degree of $a_i$ in $H$, that is,
\[\sum_{v \in V(G)} |X_v| = \sum_{i \in [3]} \alpha(i) = \sum_{i \in [3]} | N_H(a_i)| = |E(H)| = k.\]
Thus, $I_{\alcshort}$ admits a valid color assignment of total size at most $k$, completing the proof.

Note that, instead of adding all colors in $N_{G'}(u') \cap N_{G'}(v')$, one may choose a single index $i \in N_H(u') \cap N_H(v')$ and add only this color. 
The analysis remains unchanged, as the bound depends solely on the degrees of the vertices $a_i$ in $H$. 

Finally, since we have shown that a \bipwshort\ is \nph, it immediately follows that the general \pwshort\ problem is also \nph.
\end{proof}

We restate the main theorem of this section and give a concise proof using the reductions established in the previous subsections.

\VcParaNpHardness*
\begin{proof}
Let $I_{\mwcshort}$ be an instance of \mwcshort.  
By Lemma~\ref{lem:mwc-to-alc}, $I_{\mwcshort}$ can be transformed in polynomial time into an instance $I_{\alcshort}$ of \alcshort\ such that $I_{\mwcshort}$ is a yes-instance if and only if $I_{\alcshort}$ admits a valid color assignment of bounded size.  

Next, by Lemma~\ref{lem:alc-to-bip}, $I_{\alcshort}$ can be transformed in polynomial time into a restricted instance $I_{\bipwshort}$ of \pwshort\ on a bipartite graph with vertex cover 3, such that $I_{\alcshort}$ is a yes-instance if and only if $I_{\bipwshort}$ admits a pairwise distance preserver of bounded size.  

Since $I_{\bipwshort}$ is a restricted version of \pwshort, this chain of reductions implies that \pwshort\ is at least as hard as \mwcshort.  
Because \mwcshort\ is known to be \nph~\cite{3waycut}, it follows that Minimum \pwshort\ is \nph even on bipartite graphs with vertex cover 3.   
\end{proof}
\subsection{SDP is FPT by vertex cover}
In this section, we establish that \swfull\ (\swshort) is fixed-parameter tractable (FPT) when parameterized by the vertex cover size of the input graph. 
Formally, we prove Theorem~\ref{thm:sw-vc-fpt-thm}, restated next for convenience.

\SwVCFPT*

It is a well-known result that a minimum vertex cover $C$ of size at most $k$ can be computed in $O(1.2738^k + kn)$ time~\cite{ChenKX06}.

Let $(G,S,k)$ be an instance of \swshort\ such that $G$ admits a vertex cover of size at most $k$.
Fix a vertex cover $C \subseteq V(G)$ with $|C| \le k$, and let $I := V(G) \setminus C$.
Note that $I$ is an independent set; hence, for every non-isolated vertex $v \in I$, all neighbors of $v$ must lie in $C$, that is, $N_G(v) \subseteq C$.

We partition $I$ into equivalence classes $\{I_1, \dots, I_q\}$ based on the equivalence relation
\[ u \sim v \iff N_G(u) = N_G(v). \]
Note that  the number of such classes is bounded by $q \le 2^k$. 
For each $i \in [q]$, let $N_G(I_i)$ denote the common neighborhood of the vertices in $I_i$ in $G$, and let $S_i := S \cap I_i$ denote the set of terminals contained in $I_i$.

The following lemma demonstrates that there exists an optimal solution where all terminals in the same equivalence class are connected 
symmetrically.
\begin{restatable}{lemma}{sw-vc-fpt-same-connection}
\label{lemma:sw-vc-fpt-same-connection}
    For every \swshort\ $H$ of size at most $h$, there exists a \swshort\ $H'$ of size at most $h$ such that for every equivalence class $I_i$,  all terminals $t \in S_i$ satisfy $N_{H'}(t) = X_i$ for some $X_i \subseteq N_G(I_i)$.
\end{restatable}
\begin{proof}

    Let $H$ be a \swshort\ of size at most $h$. 
    If for every equivalence class $I_i$, all terminals in $S_i$ already have identical neighborhoods in $H$, then we are done.
    Otherwise, suppose there exist $t_1, t_2 \in S_i$  such that $N_H(t_1) \neq N_H(t_2)$.
    Assume $|N_H(t_1)| \le |N_H(t_2)|$ without 
    loss of generality.
    
    We construct $H'$ by setting $N_{H'}(t_2) = N_H(t_1)$ and $N_{H'}(v) = N_H(v)$ for all $v \neq t_2$. 
    That is, we remove all edges incident to $t_2$ in $H$ and add edges between $t_2$ and every vertex in $N_H(t_1)$.
    Since $N_H(t_1) \subseteq N_G(t_1) = N_G(t_2)$, it follows that $H' \subseteq G$.
    
    \medskip

    We claim that $H'$ is a valid \swshort.
    Fix an arbitrary terminal $t_j \in S$.
    
    If $t_j \in N_G(I_i)$, then $d_G(t_1, t_j)=d_G(t_2, t_j)=1$.
    Since $H$ is a \swshort, both edges $t_1t_j$ and $t_2t_j$ belong to $H$, and hence $t_j \in N_H(t_1) \cap N_H(t_2)$.
    By construction, these edges are preserved in $H'$, implying $d_{H'}(t_1,t_j) = d_{H'}(t_2,t_j) = 1$ and that $H'$ is a valid \swshort. 
    
    If $t_j \in S \setminus N_G(I_i)$, since $H$ is a \swshort\ and $I$ is an independent set, there exist vertices $u \in N_H(t_2)$ and $v \in N_H(t_1)$ such that
    \[
    d_H(t_j,t_2) = d_H(t_j,u) + 1 = d_G(t_j,t_2) \qquad \text{ and }\qquad d_H(t_j,t_1) = d_H(t_j,v) + 1 = d_G(t_j,t_1).
    \]    
    Suppose for contradiction that $d_H(t_j,u) < d_H(t_j,v)$.
    Then connecting $t_1$ to $u$ instead of $v$ would yield a strictly shorter path between $t_1$ and $t_j$ in $H$, contradicting the fact that $H$ preserves shortest-path distances.
    Therefore, $d_H(t_j,u) = d_H(t_j,v),$ and replacing the edge $t_2u$ by the edge $t_2v$ preserves the distance between $t_2$ and $t_j$.
    
    Since $t_1 \sim t_2$, we have $N_G(t_1) = N_G(t_2)$, and hence $v \in N_G(t_2)$.
    Therefore, the edge $t_2v$ exists in $G$ and is admissible in the construction.
    
    If either $u$ or $v$ is a terminal, then it is adjacent to $t_1$ in $H$.
    By construction, the same adjacency is introduced for $t_2$ in $H'$, and the
    argument continues to hold without modification.
    
    \medskip
    Finally, since $|N_H(t_1)| \le |N_H(t_2)|$, we obtain
    \[
    |E(H')| = |E(H)| - |N_H(t_2)| + |N_H(t_1)| \le |E(H)|.
    \]
    Therefore, $H'$ is a \swshort of size at most $h$.
    By exhaustively applying this transformation, we obtain a graph in which all terminals belonging to the same equivalence class have identical neighborhoods.
    See Figure~\ref{fig:sw-vc-fpt-same-connection} for an illustration of the transformation from $H$ to $H'$.
\end{proof}

\begin{figure}[h]
    \centering
    \begin{subfigure}[t]{.49\textwidth}
      \centering
      \includegraphics[width=0.62\textwidth]{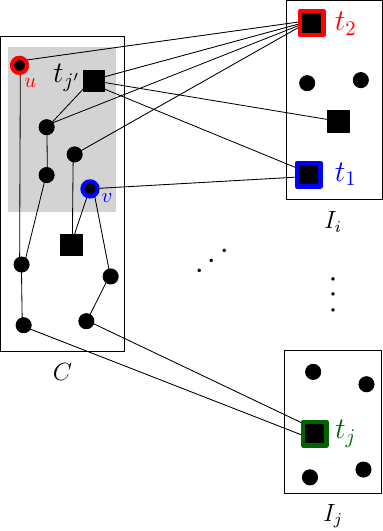}
      \caption{\swshort\ $H$ where terminals $t_1 \sim t_2$, but have different neighborhoods, i.e., $N_H(t_1) \neq N_H(t_2)$.  }
      \label{fig:sw-vc-fpt-same-connection-1}
    \end{subfigure}
    \hfill
    \begin{subfigure}[t]{.49\textwidth}
      \centering
      \includegraphics[width=0.62\textwidth]{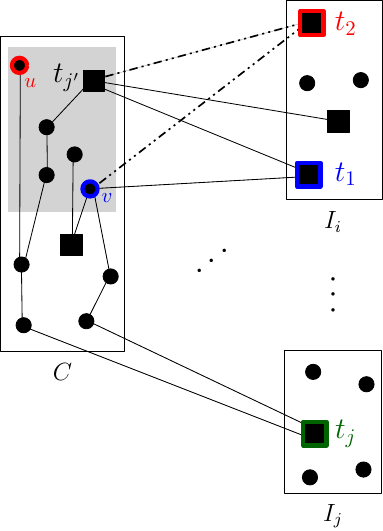}
      \caption{\swshort\ $H'$ obtained by removing all edges in $\delta_H(t_2)$ and adding edges between $t_2$ and all vertices in $N_H(t_1)$ (dotted edges). }
      \label{fig:sw-vc-fpt-same-connection-2}
    \end{subfigure}
    \caption{
        Illustration of the proof of Lemma~\ref{lemma:sw-vc-fpt-same-connection}.
        For clarity, not all vertices and edges are shown.
        In each figure, vertex cover $C$ is depicted as a rectangle on the left, and the equivalence classes on the right. 
        Terminals are represented by squares and non-terminals by circles. 
        The shaded region indicates the common neighborhood $N_G(I_i)$. 
        The edges belong to the \swshort\ $H$ in (a) and \swshort\ $H'$ in (b) respectively.
        All possible relevant placements of terminals in $S \setminus \{t_1,t_2\}$ are depicted.
        Vertices $v$ and $u$ lie on shortest paths between $t_j$ and $t_1$ and between $t_j$ and $t_2$, respectively.
        \textbf{(a)} Terminals in $N_G(I_i)$ are adjacent to both $t_1$ and $t_2$ in $H$ and remain so in $H'$.
        \textbf{(b)} Replacing the edges incident to $t_2$ preserves all shortest-path distances between terminals.
    }
    \label{fig:sw-vc-fpt-same-connection}
\end{figure}

By \autoref{lemma:sw-vc-fpt-same-connection}, we may therefore assume without loss of generality that there exists an optimal \swshort\ in which all terminals in the same equivalence class have identical connections.

The next lemma shows that, for each equivalence class, an optimal \swshort\ can be assumed to contain at most one non-terminal vertex.

\begin{restatable}{lemma}{sw-vc-fpt-one-non-terminal}
\label{lemma:sw-vc-fpt-one-non-terminal}
For every \swshort\ $H$ of size at most $h$, there exists a \swshort\ $H'$ of size at most $h$ such that for every $i \in [q]$, it holds that 
$\bigl|V(H') \cap (I_i \setminus S_i)\bigr| \le 1.$
\end{restatable}
\begin{proof}
    Let $H$ be a \swshort\ of size at most $h$.
    Suppose that there exists an index $i \in [q]$ such that $|V(H) \cap (I_i \setminus S_i)| \ge 2$.
    Let $u,v \in V(H) \cap (I_i \setminus S_i)$ be two distinct non-terminals from this set.

    \medskip
    
    We now construct $H'$ as follows.
    Delete the vertex $u$ together with all edges incident to it in $H$.
    For every vertex $x \in N_H(u)$, add the edge $vx$.
    All other vertices and edges of $H$ remain unchanged.
    More formally, we define $H' = (V(H'), E(H'))$ where
    \[
    V(H') := V(H) \setminus \{u\}
    \quad\text{and}\quad
    E(H') := \bigl(E(H) \setminus \delta_H(u)\bigr) \cup \{vx \mid x \in N_H(u)\}.
    \]

    \medskip
    We claim that $H'$ is a valid \swshort.
    By ~\autoref{obs:remove-nonterminals-not-on-shortest-path}, without loss of generality we may assume  that every vertex in $V(H) \cap (I_i \setminus S_i)$ lies on at least one shortest path between a pair of terminals.
    Let $T_u := \bigl\{(t_i, t_j) \mid t_i, t_j \in S,\ i \neq j,  \text{ and } u \text{ belongs to some } t_i-t_j \text{ shortest path}  \bigr\}$ 
    be the set of terminal pairs such that $u$ lies on their shortest path in $H$. 
    Define $T_v$ similarly. 
    
   \medskip

    \medskip
    
    Fix an arbitrary terminal pair $(t_1^u,t_2^u) \in T_u$, and consider a shortest path between them in $H$.
    Since $u$ is a non-terminal and $u \in I_i$, this path enters and leaves $u$ through two distinct neighbors $u_1,u_2 \in N_H(u) \subseteq C$, and hence
    \begin{align}
    \label{eq:sw-vc-fpt-via-u}
        d_H(t_1^u,t_2^u)
        = d_H(t_1^u,u_1) + d_H(u_1,u) + d_H(u,u_2) + d_H(u_2,t_2^u) 
        = d_H(t_1^u,u) + 2 + d_H(u_2,t_2^u).
    \end{align}
    
    \medskip

    Since $u,v \in I_i$, it holds that $N_G(u) = N_G(v).$
    In particular, $u_1,u_2 \in N_G(v)$ and $d_G(u_1, v) = d_G(v, u_2) = 1$.
    Therefore, adding the edges $vu_1$ and $vu_2$ yields a $t_1^u$--$t_2^u$ path in $H'$ of the same length:
    \begin{align}
    \label{eq:sw-vc-fpt-via-v}
    d_{H'}(t_1^u,t_2^u)
    &= d_{H'}(t_1^u,u_1) + d_{H'}(u_1,v) + d_{H'}(v,u_2) + d_{H'}(u_2,t_2^u) \notag \\
    &= d_H(t_1^u,u) + 2 + d_H(u_2,t_2^u) \overset{(\ref{eq:sw-vc-fpt-via-u})}{=} d_H(t_1^u,t_2^u).
    \end{align}
    Thus, the shortest-path distance between $t_1^u$ and $t_2^u$ is preserved in $H'$.
    
    \medskip
    
    Finally, note that no edge incident to $v$ was removed in the construction of $H'$.
    Hence, every terminal pair whose shortest path in $H$ used $v$ still has a shortest path of the same length in $H'$.
    All other distances remain unchanged.
    
    \medskip
    
    We conclude that $H'$ is a \swshort\ of size at most:
    \[
        |E(H')| = |E(H)| - |\delta_H(u)| + |\delta_H(u)| = |E(H)| \leq h.
    \]
    By repeating this operation exhaustively, we obtain a \swshort\ in which each equivalence class $I_i$ contains at most one non-terminal.
    Refer to~\autoref{fig:sw-vc-fpt-one-non-terminal} for an illustration.

\end{proof}

\begin{figure}[h]
    \centering
    \begin{subfigure}[t]{.49\textwidth}
      \centering
      \includegraphics[width=0.62\textwidth]{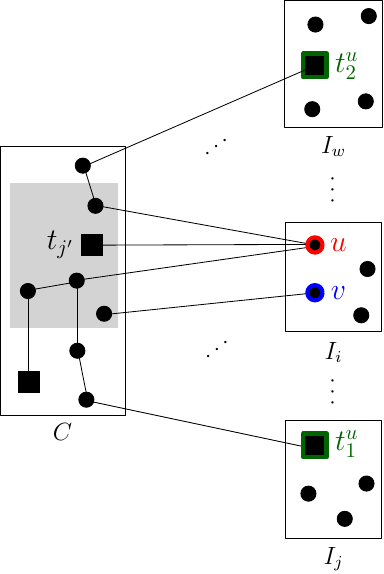}
      \caption{\swshort\ $H$ where at least two non-terminals $u \sim v$ from the same equivalence class  are taken into the solution.}
      \label{fig:sw-vc-fpt-one-non-terminal-1}
    \end{subfigure}
    \hfill
    \begin{subfigure}[t]{.49\textwidth}
      \centering
      \includegraphics[width=0.62\textwidth]{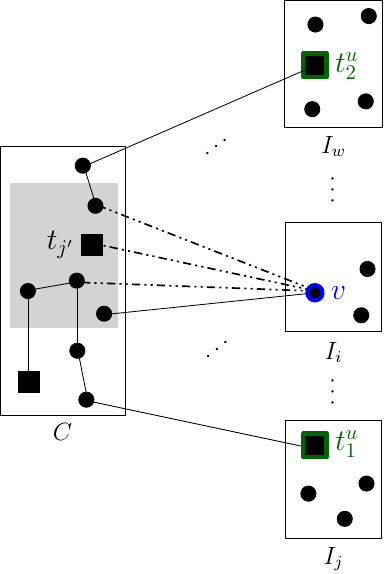}
      \caption{\swshort\ $H'$ obtained by removing $u$ and all its adjacent edges and connecting $v$ to all vertices in $N_H(u)$ (dotted edges). }
      \label{fig:sw-vc-fpt-one-non-terminal-2}
    \end{subfigure}
    \caption{
        Illustration of the proof of Lemma~\ref{lemma:sw-vc-fpt-one-non-terminal}.
        \textbf{(a)} Two non-terminal vertices $u$ and $v$ from the same equivalence class $I_i$ are used in the solution.
        \textbf{(b)} The construction removes $u$ and connects $v$ to all vertices in $N_H(u)$.
        All terminal-to-terminal shortest-path distances are preserved.
    }
    \label{fig:sw-vc-fpt-one-non-terminal}
\end{figure}

The next lemma shows that if an equivalence class contains a terminal, then there exists an optimal \swshort\ that uses only terminals from this class.

\begin{lemma}
\label{lemma:sw-vc-remove-nonterminals}
For every \swshort\ $H$ of size at most $h$, if there exists an index $i \in [q]$ with $S_i \neq \emptyset$ such that $V(H) \cap I_i \not\subseteq S_i$, then there exists a \swshort\ $H'$ of size at most $h$ satisfying $V(H') \cap I_i \subseteq S_i$.
\end{lemma}
\begin{proof}
    Let $H$ be a \swshort\ of size at most $h$, where $S_i \neq \emptyset$ for some $i\in[q]$ and $V(H) \cap I_i \not\subseteq S_i$.
    By~\autoref{lemma:sw-vc-fpt-one-non-terminal}, we may assume $|V(H) \cap I_i \not\subseteq S_i| = 1$.
    Let $v \in V(H) \cap I_i $ be the non-terminal vertex from equivalence class $I_i$ that is part of $H$.
    Let $t_i \in S_i$ be an arbitrary terminal in $I_i$.

    We construct $H'$ by removing $v$ and all the edges in $\delta_H(v)$ and connecting $t_i$ to all the vertices in $N_H(v)$.

    The argument that all terminal distances are preserved in $H'$ are identical to those used in the proof of Lemma~\ref{lemma:sw-vc-fpt-one-non-terminal}.
    The only difference is that $v$ is replaced by a terminal $t_i \in S_i$, which belongs to the same equivalence class and therefore has the same neighborhood in $G$.
\end{proof}
Note that ~\autoref{lemma:sw-vc-fpt-one-non-terminal} shows that each equivalence class needs at most one non-terminal vertex in an optimal solution.
~\autoref{lemma:sw-vc-remove-nonterminals} strengthens this statement by showing that if a class contains a terminal, then no non-terminal from that class is required.

We introduce two reduction rules based on the structural properties established in these lemmas. 
For each equivalence class $I_i$ where $i \in [q]$, we define the following:

\begin{itemize}
    \item[\textbf{(i)}] \textbf{Non-terminal Deletion}: If $I_i \cap S \neq \emptyset$, discard all 
    non-terminal vertices in $I_i$. 
    \label{rule:non-terminal-deletion}
    \item [\textbf{(ii)}] \textbf{Representative Selection}: If $I_i \cap S = \emptyset$, retain an arbitrary representative non-terminal vertex in $I_i$ and remove all other vertices of the class. \label{rule:representative}
\end{itemize}

With all the properties and definitions set, we prove the main theorem of this section, which we restate for convenience:
\SwVCFPT*
\begin{proof}
Let $(G, S, k)$ be an instance of \swshort\ such that $G$ admits a vertex cover of size at most $k$.

We describe a branching algorithm that constructs all candidate solutions consistent with the structural properties established
in Lemmata~\ref{lemma:sw-vc-fpt-same-connection}, \ref{lemma:sw-vc-fpt-one-non-terminal}, and \ref{lemma:sw-vc-remove-nonterminals} and chooses the one with the smallest size.

First, using a standard FPT algorithm for \textsc{Vertex Cover}, we compute a vertex cover $C \subseteq V(G)$ with $|C| \le k$.
Let $I := V(G) \setminus C$.
Then, we partition $I$ into equivalence classes $\{I_1,\dots,I_q\}$ with $q \in \N$ according to their neighborhoods in $C$, where $u \sim v$ if and only if $N_G(u) = N_G(v)$.

Next, we exhaustively apply the reduction rules \textbf{Non-terminal Deletion} and \textbf{Representative Selection}.
After these reductions, for every equivalence class $I_i$, either $I_i \subseteq S$ or $V(H) \cap I_i$ contains exactly one
non-terminal vertex.

Then, we enumerate all possible subgraphs induced by the vertex cover.
Let $G[C]$ denote the subgraph of $G$ induced by $C$, and let $\mathcal{H}_C$ be the family of all subgraphs of $G[C]$.
For each $H_C \in \mathcal{H}_C$, we construct candidate solutions that extend $H_C$ to the vertices in $I$.

For each equivalence class $I_i$, all vertices in $I_i$ may only connect to vertices in $N_G(I_i) \subseteq C$.
We branch on all possible subsets of $N_G(I_i)$, including the empty set $\emptyset$.
Formally, define
\[
\mathcal{H}_i := \{ X \mid X \subseteq N_G(I_i) \}.
\]
For each $X \in \mathcal{H}_i$, the branch corresponds to choosing $X$ as the neighborhood in $C$ of the equivalence class $I_i$ in the final solution.
That is, every vertex of $I_i$ that is present in the solution graph is connected exactly to the vertices in $X$.

A leaf of the branching process corresponds to a graph
\[
H = H_C \cup \bigcup_{i=1}^q E_i,
\]
where for each equivalence class $I_i$ and each chosen subset $X_i \in \mathcal{H}_i$ along the branch, we define
\[
E_i := \{ vx \mid v \in I_i,\ x \in X_i \}.
\]

Let $r_i$ denote the number of vertices of $I_i$ that appear in $H$ after the exhaustive application of the reduction rules.
By Lemmas~\ref{lemma:sw-vc-fpt-one-non-terminal} and~\ref{lemma:sw-vc-remove-nonterminals}, we have $r_i = |S_i|$ if $S_i \neq \emptyset$, and $r_i \in \{0,1\}$ otherwise. 
Only the edges incident to these $r_i$ vertices are retained in $H$.
Consequently, the contribution of $I_i$ to the edge set of $H$ is $|E_i| = r_i \cdot |X_i|.$
Thus, the total size of $H$ is
\[
|E(H)| = |E(H_C)| + \sum_{i=1}^q r_i \cdot |X_i|.
\]

Hence, although the branching decision is made once per equivalence class, the resulting graph $H$ contains all edges induced by this decision for every vertex of $I_i$, which allows us to correctly evaluate both the \swshort\ property and the size of $H$.
Among all such graphs, we select one of minimum size that preserves all pairwise shortest-path distances between terminals.

Equivalently, for each fixed subgraph $H_C \subseteq G[C]$, the branching process can be viewed as a rooted search tree of depth $q$, where the node at depth $i$ branches over all subsets in $\mathcal{H}_i$, corresponding to the possible choices of connections between the equivalence class $I_i$ and the vertex cover.

We argue that the algorithm constructs an optimal solution.

\medskip
\noindent
\textbf{Completeness.}
Let $H^\star$ be an optimal \swshort\ of $G$.
By Lemma~\ref{lemma:sw-vc-fpt-same-connection}, we may assume that all
terminals belonging to the same equivalence class have identical
neighborhoods in $H^\star$.
By Lemma~\ref{lemma:sw-vc-fpt-one-non-terminal}, each equivalence class
contains at most one non-terminal vertex in $H^\star$.
Finally, by Lemma~\ref{lemma:sw-vc-remove-nonterminals}, if $S_i \neq
\emptyset$, then $H^\star$ contains no non-terminal vertex from $I_i$.

Therefore, $H^\star$ is fully described by:
(i) its restriction to the vertex cover $C$, and
(ii) for each equivalence class $I_i$, the subset of $N_G(I_i)$ to which
the vertices of $I_i$ are connected.
Hence, $H^\star$ corresponds to one of the branches considered by the
algorithm.

\medskip
\noindent
\textbf{Optimality.}
For every candidate graph $H$ constructed by the algorithm, we explicitly
verify whether it preserves all pairwise shortest-path distances between
terminals.
Among all valid candidates, the algorithm outputs one of minimum size.
Since every optimal \swshort\ satisfying the above structural properties
is considered, the returned solution is optimal.

Next, we bound the running time.
There are at most $2^{|C|}$ choices for $H_C$.
For each equivalence class $I_i$, we branch into at most $2^{|N_G(I_i)|} \le 2^{|C|}$ possibilities.
Since $q \le 2^{|C|}$, the total number of candidates is at most $2^{|C| \cdot 2^{|C|}}.$
For each candidate graph, we can verify the \swshort\ property by computing all-pairs distances between terminals in $O((n+m)|S|^2)$ time.
Thus, the total running time is $O(2^{|C|(2^{|C|}+1)})\cdot O((n+m)|S|^2) = 2^{O(2^{|C|})}\cdot O(k^2(n+m))$. This concludes the proof of Theorem~\ref{thm:sw-vc-fpt-thm}.
\end{proof}

\section{\texorpdfstring{\wh{1}}{W[1]-hard}ness by solution size} \label{sec:solution_size}
In this section, we show that \swshort\ parameterized by solution size is \wh{1}.
We prove this via a parameterized reduction from \mccfull\ parameterized by solution size.

We first recall the problem.

\defbox{\mccfull}{\mccshort}
{A graph $G=(V, E)$, an integer $k$, and a partition of $V$ into $k$ independent sets $(V_1, V_2, \dotso, V_k)$.}
{Does there exist a subset of the vertices $X \subseteq V$ such that $|X \cap V_i| = 1$ for every $i \in [k]$ and $G[X]$ is a clique?}

The \mccshort\ problem parameterized by solution size is known to be \wh{1}~\cite{CyganFKLMPPS15}.
We show a parameterized reduction from \mccshort\ parameterized by solution size to \swshort\ and obtain the result of Theorem~\ref{thm:sw-w1-hard-soln}, restated next.
\SwSolnWHard*
\begin{proof}
    Let $I_{\mccshort} = \bigl(G=(V,E),(V_1,\dots,V_k), k\bigr)$ be an instance of \mccshort, where $k$ is the parameter.
    We construct an instance of \swshort\ denoted by $I_{\swshort}=\bigl(G'=(V',E'), S, k' \bigr)$, where $k'=k+\binom{k}{2}$.

    First we construct the graph $G'=(V',E')$ as follows.
    For each $i \in [k]$, create a copy of vertex set $V_i$, denoted by $V'_i$.
    For every vertex $v \in V_i$, let $v' \in V'_i$ denote its copy.
    For every edge $uv \in E(G)$ with $u \in V_i$ and $v \in V_j$, add the edge $u'v'$ to $E'$. 
    For each $i \in [k]$, add a terminal vertex $t_i$ to $V'_i$ and connect $t_i$ to all vertices in $V'_i$.
    Let $S=\{t_i \mid i \in [k]\}$ be the set of terminals.
    $(G'=(V',E'),S,k')$ is the constructed instance of \swshort.
    Refer to~\autoref{fig:mcc-to-subsetwise} for an illustration of this reduction.

    \begin{figure}[h]
        \centering
        \begin{subfigure}[t]{.49\textwidth}
          \centering
          \includegraphics[width=.6\textwidth, ]{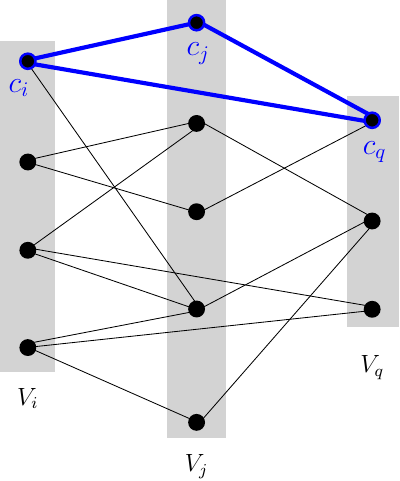}
          \caption{An instance of \mccshort.}
          \label{fig:kclique}
        \end{subfigure}
        \hfill
        \begin{subfigure}[t]{.49\textwidth}
          \centering
          \includegraphics[width=.85\textwidth]{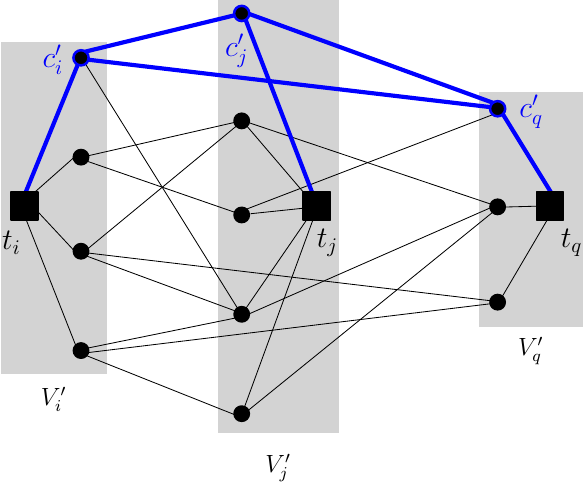}
          \caption{The corresponding instance of \swshort.}
          \label{fig:subsetwise}
        \end{subfigure}
        \caption{Illustration of the reduction from an instance of \mccshort\  $\bigl(G=(V, E),(V_1,\dots,V_k),k\bigr)$ to an instance of \swshort\ $\bigl(G'=(V', E'), S, k'\bigr)$.
        The vertex sets \(V_i, V_j,\) and \(V_q\) in \(G\) are highlighted in gray, each forming an independent set.
        For every vertex set \(V_i\) in \(G\), we create a corresponding copy \(V'_i\) in \(G'\) and introduce a terminal \(t_i\).
        Each terminal \(t_i\) is connected to all vertices in \(V'_i\) via edges.
        The squares represent the terminals \(S=\{t_i, t_j, t_q\}\) added to \(G'\), and the circles denote non-terminals.
        In (a), the blue edges and vertices indicate the solution subgraph for the \mccshort\ instance, while in (b), the blue edges and vertices indicate the solution subgraph for the \swshort\ instance.}
        \label{fig:mcc-to-subsetwise}
    \end{figure}

    We prove that $I_{\mccshort}$ is a \yes-instance if and only if $I_{\swshort}$ is a \yes-instance.
    Equivalently, $I_{\mccshort}$ has a solution of size at most $k$ if and only if $I_{\swshort}$ has a solution of size at most $k'=k+\binom{k}{2}$.

    To see the forward direction, assume that $G$ contains a multicolored clique $C$ with $k$ vertices, where $V(C)$ and $E(C)$ denote the set of vertices and edges of $C$, respectively.
    For each $i \in [k]$, by definition, we have $|V(C) \cap V_i| = 1$.
    Let $c_i$ denote the unique vertex in $V(C) \cap V_i$.
    Then for every $i,j \in [k]$ with $i \neq j$, $C$ contains the edge $c_ic_j$, which is the unique edge of $C$ between $V_i$ and $V_j$.
    
    We construct a \sw\ distance preserver $H$ for $(G',S)$ as follows.
    
    For each $c_i \in V(C)$, let $c'_i \in V'_i$ be its corresponding vertex in $G'$.
    Set $V(H)= S \cup \{c'_1,c'_2,\ldots,c'_k\}$.
    For each edge $c_ic_j \in E(C)$, add the corresponding edge $c'_ic'_j$ to $E(H)$.
    Moreover, for each $i \in [k]$, add the edge $c'_it_i$ to $H$.
    Equivalently,
    \[
        E(H)= \{c'_it_i \mid i \in [k] \} \cup \{c'_ic'_j \mid i,j\in [k] \text{ with } i \neq j\}.
    \]
    This concludes the construction of $H$.

    To see that $H$ is a valid \sw\ distance preserver, first observe that for any $i,j \in [k]$ with $i \neq j$, the shortest path in $G'$ 
    between $t_i$ and $t_j$ is $\pi_{G'}(t_i, t_j) = \langle t_i, c'_i, c'_j, t_j \rangle$, so $d_{G'}(t_i, t_j) = 3$.
    Since all these edges are included in $H$, we have $d_H(t_i, t_j) = 3 = d_{G'}(t_i, t_j)$, so $H$ is a valid \sw\ distance preserver.

    To bound the size of $H$, note that $H$ contains one edge $c'_it_i$ for each $i \in [k]$ and one edge $c'_ic'_j$ for each $i,j \in [k]$ with $i \neq j$, so $|E(H)|=k+\binom{k}{2}=k'$.

    To see the other direction, assume that $G'$ has a \sw\ distance preserver $H$ with $k+\binom{k}{2}$ edges.
    First note that for every $i \neq j$, any path between $t_i$ and $t_j$ in $G'$, must use at least one edge incident to $t_i$, an edge incident to $t_j$, and an edge between $V'_i$ and $V'_j$, so $d_H(t_i,t_j) \ge 3$.
    On the other hand, if $d_H(t_i,t_j) \ge 4$, then $H$ contains no edge between the vertex sets $V'_i$ and $V'_j$.
    By the construction of $G'$, this implies that there is no edge between $V_i$ and $V_j$ in $G$, and therefore $G$ is a \no-instance for \mccshort.
    Since in $G'$ each terminal $t_i$ is adjacent to every vertex in its corresponding set $V'_i$, the existence of an edge between $V'_i$ and $V'_j$ would yield a path of length $3$ between $t_i$ and $t_j$ in $G'$, and, as $H$ is a \sw\ distance preserver, also in $H$.
    Hence, we may restrict attention to the case where $d_H(t_i,t_j)=3$ for all $i,j \in [k]$ with $i \neq j$.

    To achieve distance $3$, $H$ must contain at least one edge between each pair $V'_i$ and $V'_j$, yielding at least $\binom{k}{2}$ edges.
    Since $|E(H)|=k+\binom{k}{2}$, it should hold that each terminal $t_i$ is incident to exactly one edge in $H$.
    Let $c'_i$ be the unique neighbor of $t_i$ in $V'_i$, and let $c_i$ be the corresponding vertex in $V_i$.
    Define $C=\{c_i : i \in [k]\}$.
    Because $H$ preserves distances, for every $i \neq j$ there is an edge $c'_ic'_j$ in $H$, which implies the existence of an edge $c_ic_j$ in $G$.
    Thus $C$ induces a clique with exactly one vertex from each $V_i$, and therefore $G$ contains a multicolored clique with $k$ vertices.
\end{proof}
Since \pwshort\ generalizes \swshort, we obtain the following.
\begin{corollary}
\label{coll:pw-w2-hard}
\pwshort\ is \wh{1} parameterized by the solution size.
\end{corollary}

\section{Conclusion}\label{sec:conclusion}

In this work, we investigated the parameterized complexity of finding minimum distance preservers, based on the number of terminals and the structural properties of the input graph. In addition to complementing the existing literature on upper and lower bounds for distance preservers, our results draw interesting parallels with known disjoint-path problems.

First, finding $k$ shortest paths between specified terminal pairs to maximize the resulting intersection---the \pwshort problem parameterized by $|\mathcal{P}|$---is mirrored by \textsc{$k$-Disjoint Shortest Paths}, which is also W[1]-hard and XP when parameterized by the number of terminal pairs. The algorithm for the latter problem is considerably more involved~\cite{Lochet21}, having stayed an open problem for a while. On the other hand, our hardness result for \pwshort and even \swshort holds in a very restricted setting, showcasing the non-trivial structure of distance preservers on grids and aligning with the known bottlenecks for distance-preserving minors~\cite{Krauthgamer2014}.

Second, \textsc{NP}-hardness of \pwshort on graphs of vertex cover $3$ can be compared to the known NP-hardness of \textsc{Edge Disjoint Paths} on the same graph class~\cite{Fleszar2018}, where in both cases the arbitrary structure of terminal pairs renders even the very restricted graph structure unhelpful. The reduction for \pwshort turns out to be more challenging however, as the desired paths have to be shortest, but need not be disjoint.

The discussion above leads also naturally to the open questions that we leave for future work. While we concluded that both \pwshort and \swshort are FPT on the grid when parameterized by the number of terminals, and that \pwshort is \nph\ on the grid, we are aware of no such result for \swshort. Hence, the status of \swshort in this case remains one of the open questions concerning the arrangements of shortest paths on the grid.

Another open question arises from the gap between \pwshort and \swshort on graphs with small separators. While our negative result for \pwshort rules out FPT or XP algorithms for other most common decompositional parameters such as treewidth, for \swshort we are only able to provide an FPT algorithm parameterized by vertex cover. The status of \swshort with respect to parameters that dominate vertex cover remains open---including treewidth, but also the parameters between treewidth and vertex cover, for example feedback vertex set or tree-depth. Drawing on the parallel with disjoint-path problems, we observe that such results are not known there either---\swshort would correspond to \textsc{Edge Disjoint Paths} with the additional restriction that the set of terminal pairs contains all pairs from the set of terminals. To the best of our knowledge, this problem has not been investigated with respect to the decompositional structure of the graph.

Finally, while our work focuses on finding the exact minimum distance preservers, it would be curious to explore FPT approximation for \pwshort and \swshort. While polynomial-time approximation for distance preservers hits hard inapproximability barriers~\cite{CDKL20,abdolmaleki2020minimum}, exponential time in the parameter and/or restricted graph structure may be helpful in combination with approximation. Especially if one aims to approximate both the size of the subgraph and the length of the shortest paths, which would correspond to FPT approximation for general spanners.

\bibliographystyle{plainurl}
\bibliography{biblio}

@inproceedings{BBC2003,
author = {Bollob\'{a}s, B\'{e}la and Coppersmith, Don and Elkin, Michael},
title = {Sparse distance preservers and additive spanners},
year = {2003},
isbn = {0898715385},
publisher = {Society for Industrial and Applied Mathematics},
address = {USA},
booktitle = {Proceedings of the Fourteenth Annual ACM-SIAM Symposium on Discrete Algorithms},
pages = {414–423},
numpages = {10},
location = {Baltimore, Maryland},
series = {SODA '03}
}

@article{CD2006,
author = {Coppersmith, Don and Elkin, Michael},
title = {Sparse Sourcewise and Pairwise Distance Preservers},
journal = {SIAM Journal on Discrete Mathematics},
volume = {20},
number = {2},
pages = {463-501},
year = {2006},
doi = {10.1137/050630696}
}

@inproceedings{bodwin2016better,
  title={Better distance preservers and additive spanners},
  author={Bodwin, Greg and Williams, Virginia Vassilevska},
  booktitle={Proceedings of the twenty-seventh annual ACM-SIAM symposium on Discrete algorithms},
  pages={855--872},
  year={2016},
  organization={SIAM}
}

@InProceedings{gajjar2017,
  author    = {Gajjar, Kshitij and Radhakrishnan, Jaikumar},
  title     = {Distance-Preserving Subgraphs of Interval Graphs},
  booktitle = {25th Annual European Symposium on Algorithms (ESA 2017)},
  pages     = {39:1--39:13},
  series    = {Leibniz International Proceedings in Informatics (LIPIcs)},
  year      = {2017},
  volume    = {87},
  publisher = {Schloss Dagstuhl -- Leibniz-Zentrum f{\"u}r Informatik},
  address   = {Dagstuhl, Germany},
  doi       = {10.4230/LIPIcs.ESA.2017.39},
  URL       = {https://drops.dagstuhl.de/entities/document/10.4230/LIPIcs.ESA.2017.39}
}

@article{abdolmaleki2020minimum,
  title={Minimum Weight Pairwise Distance Preservers},
  author={Abdolmaleki, Mojtaba and Yin, Yafeng and Masoud, Neda},
  journal={arXiv preprint arXiv:2007.07554},
  year={2020}
}

@article{bodwin2021new,
  title={New results on linear size distance preservers},
  author={Bodwin, Greg},
  journal={SIAM Journal on Computing},
  volume={50},
  number={2},
  pages={662--673},
  year={2021},
  publisher={SIAM}
}

@inproceedings{kogan2025having,
  title={Having Hope in Missing Spanners: New Distance Preservers and Light Hopsets},
  author={Kogan, Shimon and Parter, Merav},
  booktitle={Proceedings of the 2025 Annual ACM-SIAM Symposium on Discrete Algorithms (SODA)},
  pages={4352--4374},
  year={2025},
  organization={SIAM}
}

@inproceedings{KoganP22,
  author       = {Shimon Kogan and
                  Merav Parter},
  title        = {Having Hope in Hops: New Spanners, Preservers and Lower Bounds for
                  Hopsets},
  booktitle    = {63rd {IEEE} Annual Symposium on Foundations of Computer Science, {FOCS}
                  2022, Denver, CO, USA, October 31 - November 3, 2022},
  pages        = {766--777},
  publisher    = {{IEEE}},
  year         = {2022},
  url          = {https://doi.org/10.1109/FOCS54457.2022.00078},
  doi          = {10.1109/FOCS54457.2022.00078},
  bibsource    = {dblp computer science bibliography, https://dblp.org}
}

@article{Krauthgamer2014,
author = {Krauthgamer, Robert and Nguy\~{\^e}n, Huy L. and Zondiner, Tamar},
title = {Preserving Terminal Distances Using Minors},
journal = {SIAM Journal on Discrete Mathematics},
volume = {28},
number = {1},
pages = {127-141},
year = {2014},
doi = {10.1137/120888843},
}

@article{KOBAYASHI201888,
    title = {NP-hardness and fixed-parameter tractability of the minimum spanner problem},
    journal = {Theoretical Computer Science},
    volume = {746},
    pages = {88-97},
    year = {2018},
    issn = {0304-3975},
    doi = {https://doi.org/10.1016/j.tcs.2018.06.031},
    url = {https://www.sciencedirect.com/science/article/pii/S030439751830447X},
    author = {Yusuke Kobayashi}    
}

@article{AHMED2020100253,
title = {Graph spanners: A tutorial review},
journal = {Computer Science Review},
volume = {37},
pages = {100253},
year = {2020},
issn = {1574-0137},
doi = {https://doi.org/10.1016/j.cosrev.2020.100253},
url = {https://www.sciencedirect.com/science/article/pii/S1574013719302539},
author = {Reyan Ahmed and Greg Bodwin and Faryad Darabi Sahneh and Keaton Hamm and Mohammad Javad {Latifi Jebelli} and Stephen Kobourov and Richard Spence}
}

@article{FominGLMSS22,
  author       = {Fedor V. Fomin and
                  Petr A. Golovach and
                  William Lochet and
                  Pranabendu Misra and
                  Saket Saurabh and
                  Roohani Sharma},
  title        = {Parameterized Complexity of Directed Spanner Problems},
  journal      = {Algorithmica},
  volume       = {84},
  number       = {8},
  pages        = {2292--2308},
  year         = {2022},
  url          = {https://doi.org/10.1007/s00453-021-00911-x},
  doi          = {10.1007/S00453-021-00911-X},
  bibsource    = {dblp computer science bibliography, https://dblp.org}
}

@inproceedings{Lochet21,
  author       = {William Lochet},
  editor       = {D{\'{a}}niel Marx},
  title        = {A Polynomial Time Algorithm for the \emph{k}-Disjoint Shortest Paths
                  Problem},
  booktitle    = {Proceedings of the 2021 {ACM-SIAM} Symposium on Discrete Algorithms,
                  {SODA} 2021, Virtual Conference, January 10 - 13, 2021},
  pages        = {169--178},
  publisher    = {{SIAM}},
  year         = {2021},
  url          = {https://doi.org/10.1137/1.9781611976465.12},
  doi          = {10.1137/1.9781611976465.12},
  bibsource    = {dblp computer science bibliography, https://dblp.org}
}

@article{BentertFGKLPRSS26,
  author       = {Matthias Bentert and
                  Fedor V. Fomin and
                  Petr A. Golovach and
                  Tuukka Korhonen and
                  William Lochet and
                  Fahad Panolan and
                  M. S. Ramanujan and
                  Saket Saurabh and
                  Kirill Simonov},
  title        = {Packing Short Cycles},
  journal      = {{ACM} Trans. Algorithms},
  volume       = {22},
  number       = {1},
  pages        = {8:1--8:35},
  year         = {2026},
  url          = {https://doi.org/10.1145/3765285},
  doi          = {10.1145/3765285},
  bibsource    = {dblp computer science bibliography, https://dblp.org}
}

@article{RSA,
author = {Shi, Weiping and Su, Chen},
title = {The Rectilinear Steiner Arborescence Problem Is NP-Complete},
journal = {SIAM Journal on Computing},
volume = {35},
number = {3},
pages = {729-740},
year = {2005},
doi = {10.1137/S0097539704371353},
URL = {https://doi.org/10.1137/S0097539704371353},
}

@article{3waycut,
author = {Dahlhaus, E. and Johnson, D. S. and Papadimitriou, C. H. and Seymour, P. D. and Yannakakis, M.},
title = {The Complexity of Multiterminal Cuts},
journal = {SIAM Journal on Computing},
volume = {23},
number = {4},
pages = {864-894},
year = {1994},
doi = {10.1137/S0097539792225297},
URL = { https://doi.org/10.1137/S0097539792225297 },
eprint = { https://doi.org/10.1137/S0097539792225297},
}

@InProceedings{ChenKX06,
    author="Chen, Jianer
    and Kanj, Iyad A.
    and Xia, Ge",
    title="Improved Parameterized Upper Bounds for Vertex Cover",
    booktitle="Mathematical Foundations of Computer Science 2006",
    year="2006",
    publisher="Springer Berlin Heidelberg",
    pages="238--249",
}

@article{bodlaender1993linear,
  author       = {Hans L. Bodlaender},
  title        = {A Linear-Time Algorithm for Finding Tree-Decompositions of Small Treewidth},
  journal      = {{SIAM} J. Comput.},
  volume       = {25},
  pages        = {1305--1317},
  year         = {1996},
  doi          = {10.1137/S0097539793251219},
}

@book{kloks1994treewidth,
  author       = {Ton Kloks},
  title        = {Treewidth, Computations and Approximations},
  series       = {Lecture Notes in Computer Science},
  volume       = {842},
  publisher    = {Springer},
  year         = {1994},
  doi          = {10.1007/BFB0045375},
}

@article{bodlaender2016c,
	title={A $c^{k}n$ 5-approximation algorithm for treewidth},
	author={Bodlaender, Hans L and Drange, P{\aa}l Gr{\o}n{\aa}s and Dregi, Markus S and Fomin, Fedor V and Lokshtanov, Daniel and Pilipczuk, Micha{\l}},
	journal={SIAM Journal on Computing},
	volume={45},
	pages={317--378},
	year={2016},
	publisher={SIAM}
}

@Book{CyganFKLMPPS15,
  author	= {Marek Cygan and Fedor V. Fomin and Lukasz Kowalik and
		  Daniel Lokshtanov and D{\'{a}}niel Marx and Marcin
		  Pilipczuk and Michal Pilipczuk and Saket Saurabh},
  doi		= {10.1007/978-3-319-21275-3},
  publisher	= {Springer},
  title		= {Parameterized Algorithms},
  year		= {2015},
}

@article{ChekuriC16,
    author = {Chekuri, Chandra and Chuzhoy, Julia},
    title = {Polynomial Bounds for the Grid-Minor Theorem},
    year = {2016},
    issue_date = {December 2016},
    publisher = {Association for Computing Machinery},
    address = {New York, NY, USA},
    volume = {63},
    number = {5},
    issn = {0004-5411},
    url = {https://doi.org/10.1145/2820609},
    doi = {10.1145/2820609},
    journal = {J. ACM},
    month = dec,
    articleno = {40},
    numpages = {65}
}

@article{CDKL20,
    author = {Chlamt\'{a}\v{c}, Eden and Dinitz, Michael and Kortsarz, Guy and Laekhanukit, Bundit},
    title = {Approximating Spanners and Directed Steiner Forest: Upper and Lower Bounds},
    year = {2020},
    issue_date = {July 2020},
    publisher = {Association for Computing Machinery},
    address = {New York, NY, USA},
    volume = {16},
    number = {3},
    issn = {1549-6325},
    url = {https://doi.org/10.1145/3381451},
    doi = {10.1145/3381451},
    journal = {ACM Trans. Algorithms},
    month = jun,
    articleno = {33},
    numpages = {31}
}

@article{PelegS89,
    author = {Peleg, David and Schäffer, Alejandro A.},
    title = {Graph spanners},
    journal = {Journal of Graph Theory},
    volume = {13},
    number = {1},
    pages = {99-116},
    doi = {https://doi.org/10.1002/jgt.3190130114},
    url = {https://onlinelibrary.wiley.com/doi/abs/10.1002/jgt.3190130114},
    eprint = {https://onlinelibrary.wiley.com/doi/pdf/10.1002/jgt.3190130114},
    year = {1989}
}

@article{PaulChew89,
    title = {There are planar graphs almost as good as the complete graph},
    journal = {Journal of Computer and System Sciences},
    volume = {39},
    number = {2},
    pages = {205-219},
    year = {1989},
    issn = {0022-0000},
    doi = {https://doi.org/10.1016/0022-0000(89)90044-5},
    url = {https://www.sciencedirect.com/science/article/pii/0022000089900445},
    author = {L. {Paul Chew}},
}

@Article{Dobkin1990,
    author={Dobkin, David P.
    and Friedman, Steven J.
    and Supowit, Kenneth J.},
    title={Delaunay graphs are almost as good as complete graphs},
    journal={Discrete {\&} Computational Geometry},
    year={1990},
    month={Aug},
    day={01},
    volume={5},
    number={4},
    pages={399-407},
    issn={1432-0444},
    doi={10.1007/BF02187801},
    url={https://doi.org/10.1007/BF02187801}
}

@article{Pet09,
    author = {Pettie, Seth},
    title = {Low distortion spanners},
    year = {2010},
    issue_date = {December 2009},
    publisher = {Association for Computing Machinery},
    address = {New York, NY, USA},
    volume = {6},
    number = {1},
    issn = {1549-6325},
    url = {https://doi.org/10.1145/1644015.1644022},
    doi = {10.1145/1644015.1644022},
    journal = {ACM Trans. Algorithms},
    month = dec,
    articleno = {7},
    numpages = {22}
}

@article{ElkinFN17,
  author       = {Michael Elkin and
                  Arnold Filtser and
                  Ofer Neiman},
  title        = {Terminal embeddings},
  journal      = {Theor. Comput. Sci.},
  volume       = {697},
  pages        = {1--36},
  year         = {2017},
  url          = {https://doi.org/10.1016/j.tcs.2017.06.021},
  doi          = {10.1016/J.TCS.2017.06.021},
  bibsource    = {dblp computer science bibliography, https://dblp.org}
}

@article{ElkinS16,
    author = {Elkin, Michael and Pettie, Seth},
    title = {A Linear-Size Logarithmic Stretch Path-Reporting Distance Oracle for General Graphs},
    year = {2016},
    issue_date = {September 2016},
    publisher = {Association for Computing Machinery},
    address = {New York, NY, USA},
    volume = {12},
    number = {4},
    issn = {1549-6325},
    url = {https://doi.org/10.1145/2888397},
    doi = {10.1145/2888397},
    journal = {ACM Trans. Algorithms},
    month = aug,
    articleno = {50},
    numpages = {31}
}

@inproceedings{ElkinS23,
  author       = {Michael Elkin and
                  Idan Shabat},
  title        = {Path-Reporting Distance Oracles with Logarithmic Stretch and Size
                  O(n log log n)},
  booktitle    = {64th {IEEE} Annual Symposium on Foundations of Computer Science, {FOCS}
                  2023, Santa Cruz, CA, USA, November 6-9, 2023},
  pages        = {2278--2311},
  publisher    = {{IEEE}},
  year         = {2023},
  url          = {https://doi.org/10.1109/FOCS57990.2023.00141},
  doi          = {10.1109/FOCS57990.2023.00141},
  bibsource    = {dblp computer science bibliography, https://dblp.org}
}

@article{AwerbuchBCP98,
    author = {Awerbuch, Baruch and Berger, Bonnie and Cowen, Lenore and Peleg, David},
    title = {Near-Linear Time Construction of Sparse Neighborhood Covers},
    journal = {SIAM Journal on Computing},
    volume = {28},
    number = {1},
    pages = {263-277},
    year = {1998},
    doi = {10.1137/S0097539794271898},
    URL = {https://doi.org/10.1137/S0097539794271898},
    eprint = {https://doi.org/10.1137/S0097539794271898},
}

@article{Elkin05,
    author = {Elkin, Michael},
    title = {Computing almost shortest paths},
    year = {2005},
    issue_date = {October 2005},
    publisher = {Association for Computing Machinery},
    address = {New York, NY, USA},
    volume = {1},
    number = {2},
    issn = {1549-6325},
    url = {https://doi.org/10.1145/1103963.1103968},
    doi = {10.1145/1103963.1103968},
    journal = {ACM Trans. Algorithms},
    month = oct,
    pages = {283–323},
    numpages = {41}
}

@article{Cohen00,
author = {Cohen, Edith},
title = {Polylog-time and near-linear work approximation scheme for undirected shortest paths},
year = {2000},
issue_date = {Jan. 2000},
publisher = {Association for Computing Machinery},
address = {New York, NY, USA},
volume = {47},
number = {1},
issn = {0004-5411},
url = {https://doi.org/10.1145/331605.331610},
doi = {10.1145/331605.331610},
journal = {J. ACM},
month = jan,
pages = {132–166},
numpages = {35}
}

@article{Cai94,
title = {NP-completeness of minimum spanner problems},
journal = {Discrete Applied Mathematics},
volume = {48},
number = {2},
pages = {187-194},
year = {1994},
issn = {0166-218X},
doi = {https://doi.org/10.1016/0166-218X(94)90073-6},
url = {https://www.sciencedirect.com/science/article/pii/0166218X94900736},
author = {Leizhen Cai}
}

@inproceedings{BrandesH97,
  author       = {Ulrik Brandes and
                  Dagmar Handke},
  editor       = {Rolf H. M{\"{o}}hring},
  title        = {NP-Completness Results for Minimum Planar Spanners},
  booktitle    = {Graph-Theoretic Concepts in Computer Science, 23rd International Workshop,
                  {WG} '97, Berlin, Germany, June 18-20, 1997, Proceedings},
  series       = {Lecture Notes in Computer Science},
  volume       = {1335},
  pages        = {85--99},
  publisher    = {Springer},
  year         = {1997},
  url          = {https://doi.org/10.1007/BFb0024490},
  doi          = {10.1007/BFB0024490},
  bibsource    = {dblp computer science bibliography, https://dblp.org}
}

@article{VenkatesanRMMP97,
title = {Restrictions of Minimum Spanner Problems},
journal = {Information and Computation},
volume = {136},
number = {2},
pages = {143-164},
year = {1997},
issn = {0890-5401},
doi = {https://doi.org/10.1006/inco.1997.2641},
url = {https://www.sciencedirect.com/science/article/pii/S0890540197926419},
author = {G. Venkatesan and U. Rotics and M.S. Madanlal and J.A. Makowsky and C.Pandu Rangan}
}

@Article{Fleszar2018,
author={Fleszar, Krzysztof
and Mnich, Matthias
and Spoerhase, Joachim},
title={New algorithms for maximum disjoint paths based on tree-likeness},
journal={Mathematical Programming},
year={2018},
month={Sep},
day={01},
volume={171},
number={1},
pages={433-461},
issn={1436-4646},
doi={10.1007/s10107-017-1199-3},
url={https://doi.org/10.1007/s10107-017-1199-3}
}

@article{Hanan66,
author = {Hanan, M.},
title = {On Steiner’s Problem with Rectilinear Distance},
journal = {SIAM Journal on Applied Mathematics},
volume = {14},
number = {2},
pages = {255-265},
year = {1966},
doi = {10.1137/0114025},
URL = {https://doi.org/10.1137/0114025},
eprint = {https://doi.org/10.1137/0114025}
}
\end{document}